\documentclass[conference]{IEEEtran}
\IEEEoverridecommandlockouts
\usepackage{mathtools}
\usepackage{array}
\usepackage{xcolor}
\usepackage{framed}
\usepackage{bm}
\usepackage[mathscr]{euscript}
\usepackage{multirow}
\usepackage{xspace}
\usepackage{booktabs}
\usepackage{times,epsfig,epsf,psfrag,array,amsmath,amssymb,graphicx,graphics,verbatim}
\usepackage[ruled,vlined,linesnumbered, commentsnumbered]{algorithm2e}
\usepackage{url}
\usepackage{cite}
\usepackage{comment}
\usepackage{enumerate}
\usepackage{float}
\usepackage{mathtools}
\usepackage{amsthm}
\usepackage{tabularx}
\usepackage{subfigure}
\usepackage[most]{tcolorbox}

\newtcolorbox{mybox}[2][]
{enhanced, attach boxed title to top center={yshift=-2mm},
	colback=white,
	coltitle=black,
	colbacktitle=white,
	title=#2}

\def\BibTeX{{\rm B\kern-.05em{\sc i\kern-.025em b}\kern-.08em
    T\kern-.1667em\lower.7ex\hbox{E}\kern-.125emX}}

\newcommand{\eat}[1]{}
\newcommand{\eg}{{\textrm{e.g.}}\xspace}
\newcommand{\etc}{{\textrm{etc.}}\xspace}
\newcommand{\ie}{{\textrm{i.e.}}\xspace}
\newcommand{\etal}{{\textrm{et al.}}\xspace}
\newcommand{\resp}{{\textrm{resp.}}\xspace}
\newcommand{\calD}{\mathcal{D}}
\newcommand{\calF}{\mathcal{F}}
\newcommand{\simplex}{\mathbb{S}}
\newcommand{\real}{\mathbb{R}\xspace}

\newtheorem{theorem}{Theorem}
\newtheorem{lemma}{Lemma}

\newtheorem{example}{Example}
\newtheorem{problem}{Problem}

\SetKwFunction{skyprob}{\textsc{$kd$-ASP$^*$}}

\makeatletter
\newcommand{\removelatexerror}{\let\@latex@error\@gobble}
\makeatother

\begin{document}

\title{Computing All Restricted Skyline Probabilities \\on Uncertain Datasets\\
\thanks{This work is supported by the National Natural Science Foundation of China (NSFC) Grant NOs. 61832003.}
}

\author{\IEEEauthorblockN{Xiangyu Gao}
\IEEEauthorblockA{
\textit{Harbin Institute of Technology}\\
Harbin, China \\
gaoxy@hit.edu.cn}
\and
\IEEEauthorblockN{Jianzhong Li} 
\IEEEauthorblockA{
\textit{Shenzhen Institute of Advanced Technology} \\
\textit{Chinese Academy of Sciences}\\
Shenzhen, China \\
lijzh@siat.ac.cn}
\and
\IEEEauthorblockN{Dongjing Miao}
\IEEEauthorblockA{
\textit{Harbin Institute of Technology}\\
Harbin, China \\
miaodongjing@hit.edu.cn}
}

\maketitle

\begin{abstract}
	Restricted skyline (rskyline) query is widely used in multi-criteria decision making.
    It generalizes the skyline query by additionally considering a set of personalized scoring functions $\calF$.
    Since uncertainty is inherent in datasets for multi-criteria decision making, we study rskyline queries on uncertain datasets from both complexity and algorithm perspective.
    We formalize the problem of computing rskyline probabilities of all data items and show that no algorithm can solve this problem in truly subquadratic-time, unless the orthogonal vectors conjecture fails.
    Considering that linear scoring functions are widely used in practical applications, we propose two efficient algorithms for the case where $\calF$ is a set of linear scoring functions whose weights are described by linear constraints, one with near-optimal time complexity and the other with better expected time complexity.
	For special linear constraints involving a series of weight ratios, we further devise an algorithm with sublinear query time and polynomial preprocessing time.
	Extensive experiments demonstrate the effectiveness, efficiency, scalability, and usefulness of our proposed algorithms.
\end{abstract}

\begin{IEEEkeywords}
Uncertain data, probabilistic restricted skyline
\end{IEEEkeywords}

\section{Introduction}

%1.现在的方法没考虑不确定性，但是不确定性在数据中是固有的
%2.举不确定数据rskyline的例子，例子中说明为什么具有不确定性
%3.说明为什么现在的算法不能在不确定说句上用
%4.和这个两个应用最相关的问题是不确定skyline，说明为什么不确定skyline用不了
%5.说明我们做了什么，可以写的具体一些
Restricted skyline (rskyline) query is a powerful tool for supporting multi-criteria decision making, which extends the skyline query by serving the specific preferences of an individual user.
Given a dataset of multidimensional objects and a set of monotone scoring functions $\calF$, the rskyline query retrieves the set of objects that are not $\calF$-dominated by any other object.
Here an object $t$ is said to $\calF$-dominate another object $s$ if $t$ scores better than $s$ under all functions in $\calF$.
It was shown that objects returned by the rskyline query preserve the best score with respect to any function in $\calF$, and the result size is usually smaller compared to the skyline query~\cite{DBLP:journals/pvldb/CiacciaM17}.
%Therefore, the rskyline query is more effective at identifying objects of interest.
Due to its effectiveness and wide applications, many efficient algorithms have been proposed to efficiently answer rskyline queries on datasets where no uncertainty is involved~\cite{DBLP:journals/pvldb/CiacciaM17, DBLP:conf/icde/Liu0ZP021}.
% where no uncertainty is involved.
%However, none of them considered uncertain datasets.

However, uncertainty is inherent in datasets used for multi-criteria decision making caused by limitations of measuring equipment, privacy issues, data incompleteness, outdated data sources,~\etc~\cite{DBLP:journals/fcsc/LiWLG20}.
Below are two application scenarios that involve answering rskyline queries on uncertain datasets.

\underline{\it E-commerce Scenario:}
%As an example, probabilistic selling is a popular sales strategy in e-commerce~\cite{DBLP:journals/mktsci/FayX08}.
Probabilistic selling is a novel sales strategy in e-commerce~\cite{DBLP:journals/mktsci/FayX08}.
Sellers create probabilistic products by setting the probability of getting any one from a set of products.
%This essentially forms an uncertain dataset for buyers to make multi-criteria decisions.
%Theses probabilistic products form an uncertain dataset for buyers to make multi-criteria decisions.
A typical case is renting cars on {\it Hotwire}~(\url{www.hotwire.com/car-rentals/}).
%Probabilistic cars are created by grouping cars with varying horsepower (HP) and miles per gallon (MPG) collected from car rental companies (Dollar, AVIS, \etc) by type (compact SUV, median sedan, \etc).
%Probabilistic cars are created by grouping cars collected from car rental companies (Dollar, AVIS, \etc) into categories with same type (compact SUV, median sedan, \etc).
%These vehicles are similar in type (\eg, compact suv, midsize, \etc) but have different horsepower (HP) and miles per gallon (MPG).
The platform groups cars with varying horsepower (HP) and miles per gallon (MPG) by categories (\eg, compact SUV, median sedan) and abstracts these groups as probabilistic cars.
When customers choose a probabilistic car, the platform will provide any car from the corresponding group to them with a predetermined probability.
All probabilistic cars form an uncertain dataset for customers to make multi-criteria decisions.
Suppose the score of a car is defined as the weighted sum of its attributes.
It is unrealistic to expect customers to precisely determine weights of attributes.
They can only specify rough demands like MPG is more important than HP.
Then performing rskyline queries on such uncertain dataset with $\calF = \{\omega_1 \textsc{HP} + \omega_2 \textsc{MPG} \mid \omega_1 \le \omega_2\}$ can retrieve choices with high probabilities getting a car with good fuel economy to aid decision-making.
% provide an overview of potentially interesting probabilistic cars with respect to any $f \in \mathcal{F}$ to assist in decision-making.
%A typical case is booking hotels on \url{hotwire.com}.
%Probabilistic hotels are created by collecting hotels within different areas for reservation (see Fig~\ref{label}).
%When customers pay for a probabilistic hotel, they have a probability of getting any hotel within the area.
%Assume the two criteria considered by customers are price (\$) and distance to the beach (mile), and they require that if the price increases by \$100, then the distance must decrease by at least 1,000 miles.
%Performing rskyline queries on the uncertain dataset of probabilistic hotels with $\calF = \{\omega_1 \textsc{Price} + \omega_2 \textsc{Distance} \mid \omega_1 \ge 10 \cdot \omega_2\}$ helps customers to identify satisfactory probabilistic hotels.

%学习算法预测了某些事情，在预测结果上执行rskyline操作
\underline{\it Prediction Service:}
With the rapid development of machine learning, prediction services are commonly provided in fields such as finance~\cite{DBLP:journals/eswa/ZhangCXLL18}, disease control~\cite{zoabi2021machine}, healthcare~\cite{mujumdar2019diabetes}, \etc
For example, given historical data of stock market, algorithms like~\cite{DBLP:journals/eswa/ZhangCXLL18} can predict the price (P) and growth rate (GR) of a socket.
Such prediction is usually associated with a confidence value (\ie, the probability) and all predictions form an uncertain dataset.
By performing rskyline queries over this uncertain dataset with $\calF = \{\omega_1 \text{P} + \omega_2 \text{GR} \mid 0.5 \times \omega_2 \le \omega_1 \le 2 \times \omega_2\}$, we can mine an overview of stocks with high probabilities of having advantages in both price and growth rates.

Motivated by these applications, in this paper, we investigate how to conduct rskyline queries on uncertain datasets.
Similar to previous work on uncertain datasets~\cite{DBLP:conf/sigmod/GeZM09, DBLP:conf/pods/AtallahQ09, DBLP:conf/vldb/PeiJLY07, DBLP:journals/tods/AtallahQY11, DBLP:journals/mst/AfshaniAALP13, DBLP:journals/tkde/KimIP12, DBLP:conf/pakdd/NguyenC15}, we model uncertainty in a dataset by describing each uncertain object with a discrete probability distribution over a set of instances.
Then, we adopt the possible world semantics~\cite{DBLP:conf/sigmod/AbiteboulKG87} and define the rskyline probability of an object as the accumulated probabilities of all possible worlds that have one of its instances in their rskylines.
%expectation that one of its instances occurs and is not $\calF$-dominated by occurring instances of other objects.
Instead of identifying objects with top-$k$ rskyline probabilities or rskyline probabilities above a given threshold, we study the problem of computing rskyline probabilities of all objects.
This overcomes the difficulty of selecting an appropriate threshold and is convenient for users to retrieve results with different sizes.
%We refer to this problem as the all rskyline probabilities (ARSP) problem.

%TODO: Improve the motivation of probabilistic rskyline query with reference to existing works in probabilistic top-k skyline queries and ranking queries on uncertain datasets.
To our knowledge, no work has been done to address this problem to date.
Previous researches on uncertain datasets performed different tasks, such as skyline queries (\eg,~\cite{DBLP:conf/vldb/PeiJLY07, DBLP:journals/is/ZhangZLJP11, DBLP:journals/ijon/YangLZMG18, DBLP:journals/isci/YongLKH14}), top-$k$ queries (\eg,~\cite{DBLP:conf/icde/HuaPZL08, DBLP:conf/icde/SolimanIC07, DBLP:journals/dpd/WangSY16, DBLP:conf/icde/YiLKS08, DBLP:conf/sigmod/GeZM09}), \etc
%TODO: distinguish with uncertain top-k or skyline
These work only considered uncertainty in datasets but not in user preferences.
However, as stated in~\cite{DBLP:journals/pvldb/MouratidisT18, DBLP:conf/sigmod/MouratidisL021}, determining a precise scoring function for the user, which is required for top-$k$ queries, is hardly realistic, and skyline queries lack personalization.
%However, it is not reasonable to expect a user to precisely determine the scoring function like top-$k$ queries, even though he/she may roughly know what to look for~\cite{DBLP:journals/pvldb/MouratidisT18, DBLP:conf/sigmod/MouratidisL021}.
%Meanwhile, skyline queries lack of personalization.
%The input set $\calF$ describes the fuzzy scoring function of a single user o
The input $\calF$ to our problem overcomes this by considering all possible scoring functions of the user, which can be efficiently learned by algorithms like~\cite{DBLP:journals/pvldb/QianGJ15}.
%which is different from the top-$k$ query that requires a specific scoring function and the definition of the skyline query based on the dominance relationship.
%uncertain top-$k$ queries and uncertain skyline queries only considered uncertainty in datasets~\cite{DBLP:conf/icde/HuaPZL08, DBLP:journals/is/ZhangZLJP11, DBLP:journals/ijon/YangLZMG18, DBLP:journals/isci/YongLKH14} or uncertainty in user preferences~\cite{DBLP:journals/pvldb/CiacciaM17, DBLP:journals/pvldb/MouratidisT18, DBLP:conf/icde/Liu0ZP021}.
%See Fig~\ref{label} for an illustration.
Meanwhile, due to the neglect of uncertainty in datasets, existing algorithms for answering rskyline queries~\cite{DBLP:journals/pvldb/CiacciaM17, DBLP:conf/icde/Liu0ZP021} can not be applied to our problem.
An alternative approach is to convert the uncertain dataset into a certain one by representing each attribute of an object with an aggregate function like weighted sum.
However, such aggregated values lose important distribution information.
Objects with equal aggregated values but different instances will be treated equally.
Our experiments show that this makes the aggregated result ignore objects with slightly lower aggregated values, but still appear in the rskyline result of a great number of possible worlds.
We also observe that objects in the aggregated result with low rskyline probabilities will have many instances $\calF$-dominated by others' instances.
They are actually less attractive as they only belong to the rskyline result in a small set of possible worlds.

The work most related to ours is researches on computing skyline probabilities of all objects on uncertain datasets~\cite{DBLP:conf/pods/AtallahQ09, DBLP:journals/tods/AtallahQY11, DBLP:journals/mst/AfshaniAALP13, DBLP:journals/tkde/KimIP12}.
This problem is a special case of our problem because the dominance relation is equivalent to the $\calF$-dominance relation when $\calF$ contains all monotone scoring functions~\cite{DBLP:journals/pvldb/CiacciaM17}.
Although efficient algorithms were proposed for computing all skyline probabilities in~\cite{DBLP:conf/pods/AtallahQ09, DBLP:journals/tods/AtallahQY11, DBLP:journals/mst/AfshaniAALP13, DBLP:journals/tkde/KimIP12}, none of them investigated the hardness of this problem.
By establishing a fine-grained reduction from the orthogonal vector problem~\cite{DBLP:conf/stacs/Bringmann19}, we prove that no algorithm can compute rskyline probabilities of all objects within truly subquadratic time.
This also proves the near optimality of algorithms proposed in~\cite{DBLP:conf/pods/AtallahQ09, DBLP:journals/tods/AtallahQY11, DBLP:journals/mst/AfshaniAALP13} for computing all skyline probabilities.

In practice, one of the most common ways of specifying $\calF$ is to impose linear constraints on weights in linear scoring functions~\cite{DBLP:journals/pvldb/CiacciaM17}.
%这句话转的不太好
%which is one the most commonly used scoring functions~\cite{DBLP:journals/ior/DyerS79}.
Unfortunately, existing algorithms for computing all skyline probabilities~\cite{DBLP:conf/pods/AtallahQ09, DBLP:journals/tods/AtallahQY11, DBLP:journals/mst/AfshaniAALP13, DBLP:journals/tkde/KimIP12} do not suit for this case.
The reason is that the constraints on weights makes the instance's dominance region, \ie, the region contains all instances $\calF$-dominated by this instance, irregular.
We overcome this obstacle by mapping instances into a higher dimensional data space.
With this methodology, we propose a near-optimal algorithm with time complexity $O(n^{2 - 1/d'})$, where $d'$ is the dimensionality of the mapped data space.
Furthermore, by conducting the mapping on the fly and designing effective pruning strategies, we propose an algorithm with better expected time complexity based on the branch-and-bound paradigm.

%Moreover, we also consider a special linear constraints called ratio bound constraints, which has the form $l_i \le \omega[i]/\omega[d] \le h_i$ for $1 \le i < d$.
Then, we focus on a special linear constraint called weight ratio constraint, which is also studied in~\cite{DBLP:conf/icde/Liu0ZP021} for rskyline queries on certain datasets.
In such case, we improve the time complexity of the $\calF$-dominance test from $O(2^{d-1})$ to $O(d)$.
This newly proposed test condition implies a Turing reduction from the problem of computing rskyline probabilities of all objects to the half-space reporting problem~\cite{agarwal2017simplex}.
Based on this reduction, we propose an algorithm with polynomial preprocessing time and $O(2^dmn\log{n})$ query time, where $m$ and $n$ is the number of objects and instances, respectively.
Subsequently, we introduce the {multi-level} strategy and the {data-shifting} strategy to further improve the query time complexity to $O(2^{d-1}\log{n} + n)$.
%Note that the additional linear time is only required for reporting the final results.
%This algorithm matters from the following two aspects.
%This proves that the \textit{online rskyline probability query} belongs to the complexity class $\mathrm{PsL}$~\cite{DBLP:journals/tcs/GaoLML20}, which can be further used to design efficient algorithms for other queries in $\mathrm{PsL}$.
Although this algorithm is somewhat inherently theoretical, experimental results shows that its extension for this special rskyline query on certain datasets outperforms the state-of-the-art index-based method proposed in~\cite{DBLP:conf/icde/Liu0ZP021}.
%To the best of our knowledge, this paper is the first work conducting rskyline analysis on uncertain datasets.
The main contributions of this paper are summarized as follows.
\begin{enumerate}[$\bullet$]
	\item We formalize the problem of computing rskyline probabilities of all objects and prove that there is no algorithm can solve this problem in $O(n^{2-\delta})$ time for any $\delta > 0$, unless the orthogonal vectors conjecture fails.
	\item When $\calF$ is a set of linear scoring functions whose weights are described by linear constraints, we propose an near-optimal algorithm with time complexity $O(n^{2 - 1/d'})$, where $d'$ is the number of vertices of the preference region, and an algorithm with expected time complexity $O(mn\log{n})$.
	\item When $\calF$ is a set of linear scoring functions whose weights are described by weight ratio constraints, we propose an algorithm with polynomial preprocessing time and $O(2^{d-1}\log{n} + n)$ query time.
    %For online rskyline probability query, we propose an algorithm with polynomial preprocessing time and $O(2^{d-1}\log{n})$ query time.
	\item We conduct extensive experiments over real and synthetic datasets to demonstrate the effectiveness of the problem studied in this paper and the efficiency and scalability of the proposed algorithms.
\end{enumerate}

%The rest of this paper is organized as follows.
%Section~\ref{sec:relatedwork} overviews the related work.
%We review the related work in Section~\ref{sec:relatedwork}.
%Section~\ref{sec:preliminary} formally defines the problem and proves its conditional lower bound.
%We formally define the problem studied in this paper and study its conditional lower bound in Section~\ref{sec:preliminary}.
%Then, to handle linear scoring functions, two efficient algorithms are proposed for general linear constraints in Section~\ref{sec:rskyprobalg}, and an algorithm with sublinear query time and polynomial preprocessing time is proposed for special weight ratio constraints in Section~\ref{sec:eclprobalg}.
%Then, we propose two efficient algorithms in Section~\ref{sec:rskyprobalg}, and design an algorithm with sublinear query time and polynomial preprocessing time for ratio bound constraints in Section~\ref{sec:eclprobalg}.
%The experimental results are reported in Section~\ref{sec:experiments}.
%Finally, Section~\ref{sec:conclusions} concludes the paper.
\section{Problem Definition and Hardness}\label{sec:preliminary}

\subsection{Restricted Skyline}

Let $D$ be a $d$-dimensional dataset consisting of $n$ objects.
Each object $t \in D$ has $d$ numeric attributes, denoted by $t = (t[1], \cdots, t[d])$.
%W.l.o.g., we assume the numeric domain of each attribute is normalized to $[0, 1]$ and the lower values are preferred than higher ones.
W.l.o.g., we assume lower values are preferred than higher ones.
Given a {scoring function} $f : \mathbb{R}^d \to \mathbb{R}^+$, the value $f(t[1], \cdots, t[d])$ is called the {score} of $t$ under $f$, also written as $f(t)$.
Function $f$ is called {monotone} if for any two objects $t$ and $s$, it holds that $f(t) \le f(s)$ if $\forall 1 \le i \le d$, $t[i] \le s[i]$.
Let $\calF$ be a set of monotone scoring functions, an object $t$ $\calF$-{dominates} another object $s \ne t$, denoted by $t \prec_\calF s$, if $\forall f \in \calF$, $f(t) \le f(s)$.
The {restricted skyline} (rskyline) of $D$ with respect to $\calF$ is the set of objects that are not $\calF$-dominated by any other object, \ie, ${\rm RSKY}(D, \calF) = \{t \in D \mid \nexists s \in D, s \prec_\calF t\}$.

\subsection{Restricted Skyline Probability}

Let $\calD$ denote a $d$-dimensional uncertain dataset including $m$ objects.
%Following the uncertain data model used in previous related work~\cite{DBLP:conf/pods/AtallahQ09}, a $d$-dimensional uncertain dataset $\calD$ consists of $m$ objects $\{T_1, \cdots, T_m\}$.
Each uncertain object $T_i \in \calD$ is a discrete probability distribution over the $d$-dimensional data space.
In other word, the sample space of $T_i$ is a set of points $\{t_{i,1}, \cdots, t_{i, n_i}\}$ in the $d$-dimensional data space.
Each point $t_{i, j}$ is called an {instance} of $T_i$ and $T_i$ has probability $p(t_{i, j})$ to occur as $t_{i, j}$.
We also use $T_i$ to denote the set of its instances $\{t_{i, 1}, \cdots, t_{i, n_i}\}$ and write $t \in T_i$ to mean that $t$ is an instance of $T_i$.
For any object $T_i$, we assume $\sum_{t \in T_i} p(t_i) \le 1$ and $T_i$  can only take one instance at a time.
Let $I = \cup_{i = 1}^m T_i$ denote the set of all instances and $n = |I| = \sum_{i = 1}^m n_i$.
To cope with datasets of large scale, we use a spatial index R-tree to organize $I$.

Similar to previous work~\cite{DBLP:conf/vldb/PeiJLY07, DBLP:conf/pods/AtallahQ09, DBLP:journals/tods/AtallahQY11, DBLP:journals/mst/AfshaniAALP13, DBLP:journals/tkde/KimIP12, DBLP:conf/cikm/LiuZXLL15, DBLP:journals/tkde/LiuYYL13}, we adopt the possible world semantics~\cite{DBLP:conf/sigmod/AbiteboulKG87} and assume objects are independent of each other.
The uncertain dataset $\mathcal{D}$ is interpreted as a probability distribution over a set of datasets $D \sqsubseteq \mathcal{D}$ obtained by sampling each object $T_i$.
And the probability of observing the possible world $D$ is
\begin{equation}\label{eq:pwprob}
	%\small
	\Pr(D) = \prod_{t \in D} p(t) \cdot \prod_{1 \le i \le m, |T_i \cap D| = 0} (1 - \sum_{t \in T_i}p(t)).
\end{equation}
%Moreover, each object can only take one instance at a time and objects are independent of each other.

Given an uncertain dataset $\calD$ and a set of monotone scoring functions $\calF$, the rskyline probability of an instance $t \in T_i$ is the accumulated possible world probabilities of all possible worlds that have $t$ in their rskyline with respect to $\calF$. Formally,
\begin{equation}\label{eq:rskyprob-pw}
	\mathrm{Pr}_{\rm rsky}(t) = \sum_{D \sqsubseteq \mathcal{D}} \times \mathbf{1}[t \in {\rm RSKY}(D, \calF)]
\end{equation}
where $\mathbf{1}[\cdot]$ is the indicator function.
%Given an uncertain dataset $\calD$ and a set of monotone scoring functions $\calF$, an instance $t \in T_i$ belongs to the rskyline of $\calD$ if and only if $T_i$ occurs as $t$ and none of other objects appears as an instance that $\calF$-dominates $t$.
%We refer to such probability as the \textit{rskyline probability} of $t$, denoted by $\Pr_{\rm rsky}(t)$.
%With the above assumption, $\Pr_{\rm rsky}(t)$ can be computed as follows,
And the rskyline probability of an object $T_i$, denoted by $\Pr_{\rm rsky}(T_i)$, is defined as the sum of rskyline probabilities of all its instances.%, \ie,
%\begin{equation}
%	\mathrm{Pr}_{\rm rsky}(T_i) = \sum_{t \in T_i} \mathrm{Pr}_{\rm rsky}(t).
%\end{equation}

\begin{figure}[htp]
	\includegraphics[width=\linewidth]{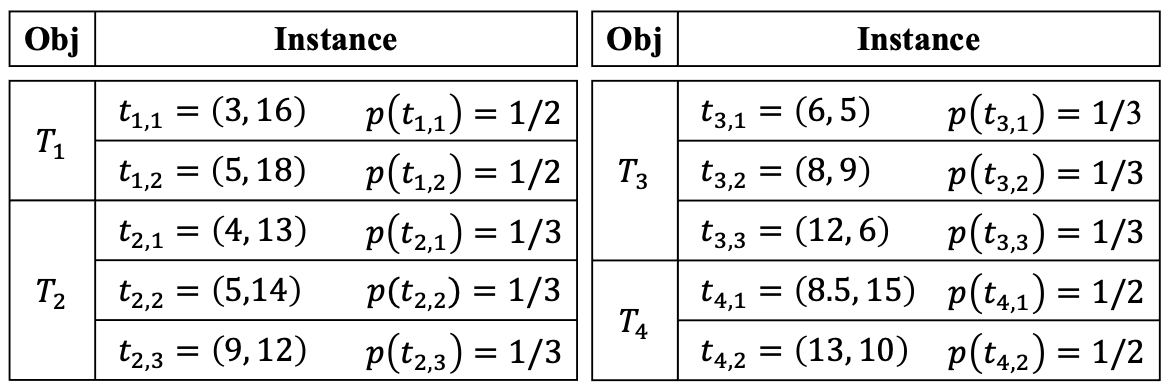}\vspace{-2mm}
	\caption{{An uncertain dataset $\calD$ of 4 objects and 10 instances.}}
	\label{fig:model-illustration}
\end{figure}

\begin{example}\label{ep}
	Consider the uncertain dataset $\calD$ shown in Fig.~\ref{fig:model-illustration}.
	$D = \{t_{1,1}, t_{2,1}, t_{3,1}, t_{4,1}\}$ is a possible world of $\calD$ and $\Pr(D) = p(t_{1,1})\times p(t_{2,1}) \times p(t_{3,1}) \times p(t_{4,1}) = 1/36$.
	Let $\calF = \{\omega[1]t[1] + \omega[2]t[2] \mid 0.5\times \omega[2] \le \omega[1] \le 2 \times \omega[2]\}$, the set of possible worlds that have $t_{1,1}$ in their rskyline with respect to $\calF$ is $S = \{t_{1,1}\} \times \{t_{2,2}, t_{2,3}\} \times \{t_{3,2}, t_{3,3}\} \times \{t_{4,1}, t_{4,2}\}$.
	Therefore, $\Pr_{\rm rsky}(t_{1,1}) = \sum_{D \in S} \Pr(D) = 2/9$.
	Similarly, we know $\Pr_{\rm rsky}(t_{1,2}) = 0$.
	Hence, $\Pr_{\rm rsky}(T_1) = \Pr_{\rm rsky}(t_{1,1}) + \Pr_{\rm rsky}(t_{1,2}) = 2/9$.
	%There are 4 objects in this uncertain dataset and their instances and existence probabilities are shown in the table.
	%Given a set of scoring functions $\calF = \{\omega[1]t[1] + \omega[2]t[2] \mid \omega[1] \ge \omega[2]\}$, regions containing all instances $\calF$-dominating $b_3$ and being $\calF$-dominated by $b_3$ are shaded in gray and green, respectively.
	%Thus, the rskyline probability of $b_3$ can be computed as $\Pr_{\rm rsky}(b_3) = p(b_3) \times (1 - p(c_1)) = 0.18$.
	%Similarly, we can derive that $\Pr_{\rm rsky}(b_1) = 0.3$ and $\Pr_{\rm rsky}(b_2) = 0.018$.
	%The rskyline probability of object $B$ is $\Pr_{\rm rsky}(B) = \Pr_{\rm rsky}(b_1) + \Pr_{\rm rsky}(b_2) + \Pr_{\rm rsky}(b_3) = 0.498$.
\end{example}

%In this paper, we study the problem of computing rskyline probabilities of all instances, from which the rskyline probabilities of all objects can also be computed.
The main problem studied in this paper is as follows.

\begin{problem}[{All RSkyline Probabilities} (ARSP)]
	Given an uncertain dataset $\calD = \{T_1, \cdots, T_m\}$ and a set of monotone scoring functions $\calF$, compute rskyline probabilities of all instances in $I = \cup^m_{i = 1} T_i$, \ie, return the set
	\[{\rm ARSP} = \{(t, {\rm Pr}_{\rm rsky}(t)) \mid t \in I\}.\]
\end{problem}

%\noindent{\bf All RSkyline Probabilities (ARSP) Problem}\\
%\noindent{\bf Input:} an uncertain dataset $\calD$ and a set of monotone scoring functions $\calF$. \\
%\noindent{\bf Output:} rskyline probabilities of all instances in $I$, \ie,
%\[{\rm ARSP} = \{(t, {\rm Pr}_{\rm rsky}(t)) \mid t \in I\}.\]

\subsection{Conditional Lower Bound}

We show that no algorithm can solve the ARSP problem in truly subquadratic time without preprocessing, unless the orthogonal vectors conjecture fails.

\noindent{$\blacktriangleright$ \bf Orthogonal Vectors Conjecture~\cite{DBLP:conf/stacs/Bringmann19}.} Given two sets $A, B$, each of $n$ vectors in $\{0, 1\}^d$, for every $\delta > 0$, there is a $c \ge 1$ such that no $O(n^{2 - \delta})$-time algorithm can determine if there is a pair $(a, b) \in A\times B$ such that $a \times b = 0$ with $d = c\log{n}$.

\begin{theorem}\label{thm:lower-bound}
	Given an uncertain dataset $\calD$ and a set of monotone scoring functions $\calF$, no algorithm can compute rskyline probabilities of all instances within $O(n^{2-\delta})$ time for any $\delta > 0$, unless the Orthogonal Vectors conjecture fails.
\end{theorem}

%For lack of space, all missing proofs can be found in~\cite{gao2023computing}.

%\eat{
\begin{proof}
	We establish a fine-grained reduction from the orthogonal vectors problem to the ARSP problem.
	Given two sets $A, B$, each of $n$ vectors in $\{0, 1\}^d$, we construct an uncertain dataset $\calD$ and a set $\calF$ of monotone scoring functions as follows.
	First, for each vector $b \in B$, we construct an uncertain tuple $T_b$ with a single instance $b$ and $p(b) = 1$.
	Then, we construct an uncertain tuple $T_A$ with $n$ instances $\xi(a)$ and $p(\xi(a)) = \frac{1}{n}$ for all vectors $a \in A$, where $\xi(a)[i] = \frac{3}{2}$ if $a[i] = 0$ and $\xi(a)[i] = \frac{1}{2}$ if $a[i] = 1$ for $1 \le i \le d$.
	Finally, let $\calF$ consists of $d$ linear scoring functions $f_i(t) = t[i]$ for $1 \le i \le d$, which means instance $t$ $\calF$-dominates another instance $s$ if and only if $t[i] \le s[i]$ for $1 \le i \le d$.
	We claim that for each instance $\xi(a) \in T_A$, there exists an instance $b$ from other uncertain tuple $T_b$ $\calF$-dominating $\xi(a)$ if and only if $a$ is orthogonal to $b$.
	Suppose there is a pair $(a, b) \in A \times B$ such that $a \times b = 0$, then $a[i] = 0$ or $b[i] = 0$ for $1 \le i \le d$.
	If $a[i] = 0$, then $b[i]$ can be either 0 or 1 and $\xi(a)[i] = \frac{3}{2} > b[i]$.
	Or if $b[i] = 0$, then $a[i]$ can be either 0 or 1 and $\xi(a)[i] \ge \frac{1}{2} > b[i]$.
	That is $b \prec_\calF \xi(a)$.
	On the other side, suppose there is a pair of instances $b$ and $\xi(a)$ such that $b \prec_\calF \xi(a)$.
	For each $1 \le i \le d$, $b[i]$ is either 0 or 1 and $\xi(a)[i]$ is either $\frac{3}{2}$ and $\frac{1}{2}$.
	If $b[i] = 0$, then $b[i]\cdot a[i] = 0$.
	Or if $b[i] = 1$, then $\xi(a)[i] = \frac{3}{2}$ since $b[i] \le \xi(a)[i]$.
	So $a[i] = 0$  according to the mapping $\xi(\cdot)$.
	Hence $a[i] \cdot b[i] = 0$.
	Thus we conclude that there is a pair $(a, b) \in A \times B$ such that $a \times b = 0$ if and only if there exists an instance $\xi(a) \in T_A$ with $\Pr_{\rm rsky}(\xi(a)) = 0$.
	Since $\calD$ can be constructed in $O(nd)$ time and whether such instance exists can be determined in $O(n)$ time, any $O(n^{2-\delta})$-time algorithm for all rskyline probabilities computation for some $\delta > 0$ would yield an algorithm for Orthogonal Vectors in $O(nd + n^{2 - \delta} + n) = O(n^{2 - \delta'})$ time for some $\delta' > 0$ when $d = \Theta(\log{n})$, which contradicts the Orthogonal Vectors conjecture.
\end{proof}
%}
\section{Algorithms for ARSP Problem\\with Linear Scoring Functions}\label{sec:rskyprobalg}

The linear scoring function is one of the most commonly used scoring functions in practice~\cite{DBLP:journals/ior/DyerS79}.
Given a weight $\omega$, the {score} of an object $t$ is defined as $S_\omega(t) = \sum^d_{i = 1} \omega[i]t[i]$.
Since ordering any two objects by $S_\omega(\cdot)$ is independent from the magnitude of $\omega$, we assume $\omega$ belongs to the unit $(d-1)$-simplex $\simplex^{d-1}$, \ie, $\forall 1 \le i \le d$, $\omega[i] \ge 0$, and $\sum^d_{i = 1} \omega[i] = 1$.
To serve the specific preferences of an individual user, a notable approach is to add linear constraints $A \times \omega \le b$ on $\simplex^{d-1}$, where $A$ is a $c \times d$ matrix and $b$ is a $c \times 1$ matrix.
In this section, we propose two efficient algorithms to compute ARSP in case of $\calF = \{S_\omega(\cdot) \mid \omega \in \simplex^{d-1} \wedge A \times \omega \le b\}$.

\subsection{Baseline Algorithms}\label{subsec:bsl}

According to equation~(\ref{eq:rskyprob-pw}), a baseline algorithm to compute ARSP is to enumerate each possible worlds $D \sqsubseteq \calD$, compute ${\rm RSKY}(D, \calF)$, and add $\Pr(D)$ to $\Pr_{\rm rsky}(t)$ for each $t \in {\rm RSKY}(D, \calF)$.
However, this brute force algorithm is infeasible due to the exponential time complexity.

Note that for any $D \sqsubseteq \calD$, an instance $t \in T_i$ belongs to ${\rm RSKY}(D, \calF)$ if and only if $T_i$ occurs as $t$ in $D$ and none of other objects appears as an instance that $\calF$-dominates $t$ in $D$.
Thus, $\Pr_{\rm rsky}(t)$ can be equivalently represented as
\begin{equation}\label{eq:rskyprob-def}
	\mathrm{Pr}_{\rm rsky}(t) = p(t) \cdot \prod_{j = 1, j \ne i}^m(1 - \sum_{s \in T_j, s \prec_\calF t} p(s)).
\end{equation}
%Based on equation~(\ref{eq:rskyprob-def}), another algorithm to calculate $\Pr_{\rm rsky}(t)$ for each $t \in I$ is to compute the product of probabilities that all other objects occur as instances that do not $\calF$-dominate $t$.
The major challenge of equation~(\ref{eq:rskyprob-def}) is to compute the product of probabilities that all other objects occur as instances that do not $\calF$-dominate $t$.
A straight approach is to perform $\calF$-dominance tests between $t$ and all instances from other objects.
With the fact that the {preference region} $\Omega = \{\omega \in \simplex^{d-1} \mid A \times \omega \le b\}$ is a {closed convex polytope}, the $\calF$-dominance relation between two instances can be determined by comparing their scores under the set of vertices $V$ of $\Omega$.
Here a weight $\omega$ is called a vertex of $\Omega$ if and only if it is the unique solution to a $d$-subset inequalities of $A \times \omega \le b$.

\begin{theorem}[$\calF$-{dominance test}~\cite{DBLP:journals/pvldb/CiacciaM17}]
	Given a set of linear scoring functions $\calF = \{S_\omega(\cdot) \mid \omega \in \simplex^{d-1} \wedge A \times \omega \le b\}$, let $V$ be the set of vertices of the preference region $\Omega = \{\omega \in \simplex^{d-1} \mid A \times \omega \le b\}$, an instance $t$ $\calF$-dominates another instance $s$ if and only if $S_\omega(t) \le S_\omega(s)$ holds for all weights $\omega \in V$.
	\label{thm:F-dominace-V}
\end{theorem}

With Theorem~\ref{thm:F-dominace-V}, we construct another baseline algorithm as follows.
Since the preference region $\Omega$ is closed, the set of linear constraints can be transformed into a set of points using the \textit{polar duality}~\cite{preparata2012computational} such that the intersection of the linear constraints is the dual of the convex hull of the points.
After the transformation, the baseline invokes the quickhull algorithm proposed in~\cite{DBLP:journals/toms/BarberDH96} to compute the set of vertices $V$ of $\Omega$.
Then it sorts the set of instances using a scoring function $S_\omega(\cdot)$ for some $\omega \in V$.
This guarantees that if an instance $t$ precedes another instance $s$ in the sorted list, then $s \nprec_\calF t$.
After that, for each instance $t$, the baseline tests $t$ against every instance of other objects preceding $t$ in the sorted list to compute $\Pr_{\rm rsky}(t)$ according to equation~(\ref{eq:rskyprob-def}).
Since $V$ can be computed in $O(c^2)$ time~\cite{greenfield1990proof}, where $c$ is the number of linear constraints, and each $\calF$-dominance test can be performed in $O(dd')$ time, where $d' = |V|$, the time complexity of the baseline algorithm is $O(c^2 + dd'n^2)$.
Although the theoretical upper bound of $d'$ is $\Theta(c^{\lfloor d/2 \rfloor})$~\cite{henk2017basic}, the actual size of $V$ is experimentally observed to be small.
%Hence we conclude that the time complexity of the baseline algorithm is $O(n^2)$.	

\subsection{Tree-Traversal Algorithm}\label{subsec:tt}

We say an object $t$ dominates another object $s \ne t$, denoted by $t \preceq s$, if $\forall 1 \le i \le d, t[i] \le s[i]$.
Given an uncertain dataset $\calD$, the skyline probability of an instance $t \in T_i$ is defined as
\[{\rm Pr}_{\rm sky}(t) = p(t) \cdot \prod^m_{j = 1, j \ne i}(1 - \sum_{s \in T_j, s \preceq t} p(s)).\]
The all skyline probabilities (ASP) problem aims to compute skyline probabilities of all instances~\cite{DBLP:conf/pods/AtallahQ09, DBLP:journals/tods/AtallahQY11, DBLP:journals/mst/AfshaniAALP13, DBLP:journals/tkde/KimIP12}. 
In case of $\calF = \{S_\omega(\cdot) \mid \omega \in \simplex^{d-1} \wedge A \times \omega \le b\}$, we show how to transform the ARSP problem to the ASP problem.
%we overcome this challenge by reducing the ARSP problem to the ASP problem~\cite{DBLP:conf/pods/AtallahQ09, DBLP:journals/tods/AtallahQY11, DBLP:journals/mst/AfshaniAALP13, DBLP:journals/tkde/KimIP12} in a higher dimensional data space.
%Then calling the state-of-the-art method~\cite{DBLP:journals/mst/AfshaniAALP13} for the ASP problem yields an algorithm with near-optimal time complexity.
%We say that an object $t$ dominates another object $s \ne t$, denoted as $t \prec s$, if $\forall i \in [1, d], t[i] \le s[i]$.
%And the formal definition of the ASP problem is stated as follows.

\begin{figure}[!t]
    \removelatexerror
    \begin{algorithm}[H]
	\caption{KDTree-Traversal Algorithm}
	\label{alg:trans-alg}
	\KwIn{an uncertain dataset $\calD$, a set of linear scoring functions $\calF = \{S_\omega(\cdot) \mid \omega \in \simplex^{d-1} \wedge A \times \omega \le b\}$}
	\KwOut{ARSP}\BlankLine
	Compute vertices $V$ of $\Omega = \{\omega \in \simplex^{d-1} \mid A \times \omega \le b\}$\;
	Construct the uncertain dataset $\calD'$\;
	${\rm ARSP} \gets \emptyset$; $\chi \gets 0$; $\beta \gets 1$\;
	\lForEach{$i \gets 1$ \emph{\bf to} $m$}{$\sigma[i] \gets 0$}
	\skyprob{$I', I'$}\;
	\Return{${\rm ARSP}$}\;\BlankLine
	\SetKwProg{proc}{Procedure}{}{}
	\proc{\skyprob{$P$,  $C$}}{
		$C_{par} \gets C$; $C \gets \emptyset$; $D \gets \emptyset$\;
		\ForEach{$S_{_V}(t) \in C_{par}$}{
			\If{$S_{_V}(t) \preceq P_{\min}$ \emph{(say $t \in T_i$)}}{
				Insert $S_{_V}(t)$ into $D$\;
				$\sigma[i] \gets \sigma[i] + p(t)$\;
				\If{$\sigma[i] = 1$}{
					$\chi \gets \chi + 1$;
					$\beta \gets \beta/p(t)$\;
				}
				\Else{
					$\beta \gets \beta \times (1 - \sigma[i]) / (1 - \sigma[i] + p(t))$\;
				}
			}
			\ElseIf{$S_{_V}(t) \preceq P_{\max}$}{
				Insert $S_{_V}(t)$ into $C$\;
			}
		}
		\If{$\chi = 0$ \emph{\bf and} $|P| = 1$} {
            \tcp{\small suppose $P = \{S_{_V}(t)\}$ and $t \in T_i$}
			Insert $(t, \beta\times p(t)/(1 - \sigma[i]))$ into ${\rm ARSP}$\;
		}\ElseIf{$\chi = 0$ \emph{\bf and} $|P| > 1$}{
			Partition $P$ into $P_l$ and $P_r$ with selected $axis$\;
			%			Compute MBR $R_l$ (\resp, $R_r$) of $I_l$ (\resp, $I_r$)\;
			\skyprob{$P_l$, $C$}\;
			\skyprob{$P_r$, $C$}\;
		}
		\ForEach{$t \in D$}{
			Undo the changes to restore $\sigma$, $\chi$, $\beta$\;
		}
		$C \gets C_{par}$\;
	}
\end{algorithm}
\end{figure}

%\noindent{\bf All Skyline Probabilities (ASP) Problem}~\cite{DBLP:journals/mst/AfshaniAALP13}\\
%\noindent{\bf Input:} an uncertain dataset $\calD$. \\
%\noindent{\bf Output:} skyline probability $\Pr_{\rm sky}(t)$ of each instance $t \in I$, where if $t \in T_i$, then
%\[{\rm Pr}_{\rm sky}(t) = \Pr(t) \cdot \prod^m_{j = 1, j \ne i}(1 - \sum_{s \in T_j, s \prec t} \Pr(s)).\]

Given a $d$-dimensional uncertain dataset $\calD$ and a set of linear scoring functions $\calF = \{S_\omega(\cdot) \mid \omega \in \simplex^{d-1} \wedge A \times \omega \le b\}$, let $V = \{\omega_1, \cdots, \omega_{d'}\}$ be the set of vertices of the preference region $\Omega = \{\omega \in \simplex^{d-1} \mid A \times \omega \le b\}$ and $d' = |V|$.
For each $t \in I$, $S_{_V}(t) = (S_{\omega_1}(t), \cdots, S_{\omega_{d'}}(t))$ is a $d'$-dimensional point whose $i$-th coordinate is the score of instance $t$ under $\omega_i \in V$.
We construct a $d'$-dimensional uncertain dataset $\calD'$ as follows.
For each uncertain object $T_i \in \calD$, we create an uncertain object $T'_i$ in $\calD'$.
Then, for each instance $t \in T_i$, we compute $S_{_V}(t)$ as an instance of $T'_i$ and set $p(S_{_V}(t)) = p(t)$.
From Theorem~\ref{thm:F-dominace-V}, it is directly to know that for any two instance $t, s \in I$, $t \prec_\calF s$ if and only if $S_{_V}(t) \preceq S_{_V}(s)$.
This means, for each $t \in I$, $\Pr_{\rm rsky}(t) = \Pr_{\rm sky}(S_{_V}(t))$.
Thus, after constructing $\calD'$, we employ the procedure \skyprob on $\calD'$ to compute skyline probabilities of all instances in $I' = \cup^m_{i=1} T'_i$.

\skyprob is an optimized implementation of the state-of-the-art algorithm for the ASP problem proposed in~\cite{DBLP:journals/mst/AfshaniAALP13}.
The original algorithm first constructs a $kd$-tree $T$ on $I'$, and then progressively computes skyline probabilities of all instances by performing a preorder traversal of $T$.
We optimized it by integrating the preorder traversal into the construction of $T$ and pruning the construction of a subtree if all instances included in the subtree have zero skyline probabilities.
Although these optimizations does not improve the time complexity, they do enhance its experimental performance.

Concretely, \skyprob always keeps a path from the root of $T$ to the current reached node in the main memory.
And for each node $N$ in the path, let $P$ be the set of instances contained in $N$ and $P_{\min}$ ($P_{\max}$) denote the minimum (maximum) corner of the minimum bounding rectangle of $P$, \skyprob maintains the following information,
(1) a set $C$ including instances that dominates $P_{\max}$,
(2) an array $\sigma = \langle \sigma[1], \sigma[2], \cdots \sigma[m] \rangle$, where $\sigma[i] = \sum_{t \in T_i, S_{_V}(t) \preceq P_{\min}} p(t)$, \ie, the sum of existence probabilities over instances of $T'_i$ that dominate $P_{\min}$,
(3) a value $\beta = \prod_{1 \le i \le m, \sigma[i] \ne 1}(1 - \sigma[i])$,
and (4) a counter $\chi = |\{i \mid \sigma[i] = 1\}|$.

At the beginning, \skyprob initializes $C = I'$, $\sigma[i] = 0$ for $1 \le i \le m$, $\beta = 1$, and $\chi = 0$ at the root node of $T$.
Supposing the information of all nodes in the maintained path is available, \skyprob constructs the next arriving node $N$ as follows.
Again, let $P$ denote the set of instances in $N$.
For each instance $S_{_V}(t) \in C_{par}$, where $C_{par}$ is the set $C$ of the parent node of $N$, it tests $S_{_V}(t)$ against $P_{\min}$.
If $S_{_V}(t) \preceq P_{\min}$, say $t \in T_i$, it updates $\sigma[i]$, $\beta$, and $\chi$ as stated in lines 12-16 of Algorithm~\ref{alg:trans-alg}.
Otherwise, it further tests $S_{_V}(t)$ against $P_{\max}$ and inserts $S_{_V}(t)$ into the set $C$ of $N$ if $S_{_V}(t) \preceq P_{\max}$.
When $\chi$ becomes to one, we know that $\Pr_{\rm sky}(P_{\min}) = 0$, and so are all instances in $N$ due to the transitivity of dominance.
Therefore, \skyprob prunes the construction of the subtree rooted at $N$ and returns to its parent node.
Otherwise, \skyprob keeps growing the path by partitioning set $P$ like a $kd$-tree until it reaches a node including only one instance $S_{_V}(t)$, then computes $\Pr_{\rm sky}(S_{_V}(t))$ based on $\beta$ and $\sigma$.

\begin{figure}[t]
	\subfigure[$kd$-tree for $I'$.]{
		\includegraphics[width=.42\linewidth]{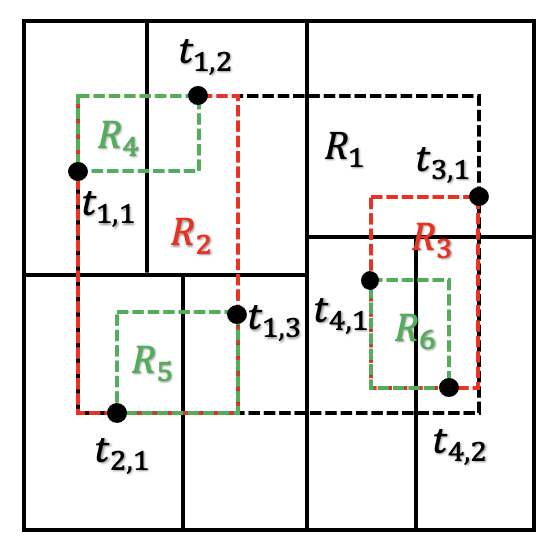}
		\label{fig:data-space}
	}
    \hfill
	\subfigure[Pruning at node $R_3$.]{
		\includegraphics[width=.5\linewidth]{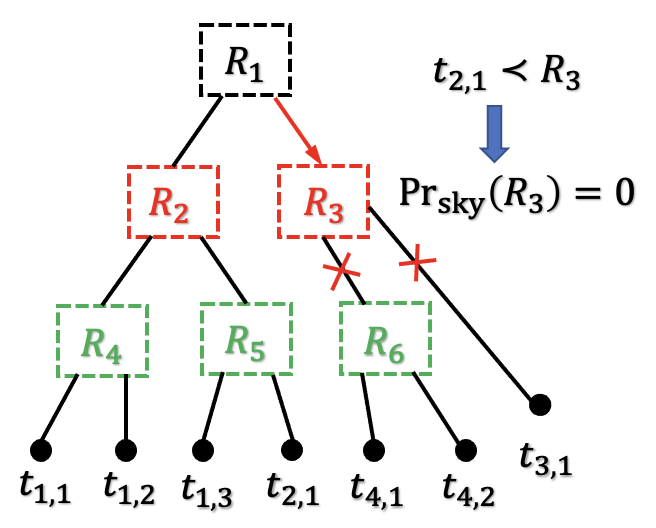}
		\label{fig:kdtree}
	}\vspace{-2mm}
	\caption{Running example for $kd$-\textsc{ASP}$^*$.}\vspace{-2mm}
	\label{fig:alg-trans}
	\vspace{-2mm}
\end{figure}

%TODO: running example of algorihtm 
\begin{example}
	As shown in Fig.~\ref{fig:alg-trans}, suppose all instances of an object occur with the same probability.
    The original algorithm keeps a whole $kd$-tree in the main memory but \skyprob only maintains a path from the root node, \eg, $R_1 \to R_2 \to R_5$.
	Moreover, when \skyprob traverses from $R_1$ to $R_3$, it updates $\sigma[2]$ to $1$ and $\chi$ to $1$ since $t_{2,1} \preceq R_3$.
	This indicates that the skyline probabilities of all instances in the subtree rooted at $R_3$ is zero, thus \skyprob prunes the construction of the subtree rooted at $R_3$ as shown in Fig.~\ref{fig:kdtree}.
\end{example}

%TODO: time complexity analyze
The pseudocode of the entire algorithm is shown in Algorithm~\ref{alg:trans-alg}.
As stated previously, the computation of $V$ takes $O(c^2)$ time, where $c$ is the number of linear constraints.
For each instance $t \in I$, $S_{_V}(t)$ can be computed in $O(dd')$ time, where $d' = |V|$.
According to~\cite{DBLP:journals/mst/AfshaniAALP13}, the time complexity of \skyprob on a set of $n$ instances in $d'$-dimensional data space is $O(n^{2-1/d'})$.
Therefore, the overall time complexity of Algorithm~\ref{alg:trans-alg} is $O(c^2 + d'dn + n^{2-1/d'}) = O(n^{2-1/d'})$.

Next, we claim that Theorem~\ref{thm:lower-bound} still holds even if we limit $\calF$ into linear scoring functions whose weights are described by linear constraints.
%TODO: optimaility
Let $\calF$ be the set of all linear scoring functions.
Given two instances $t$ and $s$, if $t \prec_\calF s$, then $t[i] \le s[i]$ for $1 \le i \le d$ since $\omega_i \in \Omega$ where $\omega_i[i] = 1$ and $\omega_i[j] = 0$ for all $1 \le j \ne i \le d$.
If $t[i] \le s[i]$ for $1 \le i \le d$, it is known that $t \prec_\calF s$ since all linear scoring functions are monotone.
Hence, we can also conclude that $t \prec_\calF s$ if and only if $t[i] \le s[i]$ for $1 \le i \le d$.
Thus, with the same reduction established in the proof of Theorem~\ref{thm:lower-bound}, it is known that there is no subquadratic-time algorithm for the ARSP problem even if $\calF$ is limited into linear scoring functions whose weights are described by linear constraints.
This proves that Algorithm~\ref{alg:trans-alg} achieves a near-optimal time complexity.

%TODO: replace kd-tree with other space partition tree also works for this method
\noindent{\bf Remark.}
Algorithm~\ref{alg:trans-alg} also works when \skyprob adopts any other space-partitioning tree.
The only detail that needs to be modified is the method to partition the data space (line 23-25 of Algorithm~\ref{alg:trans-alg}).
In our experimental study, we implement a variant of Algorithm~\ref{alg:trans-alg} based on the quadtree, which partitions the data space in all dimensions each time.
It is observed that choosing an appropriate space-partitioning tree can improve the performance of Algorithm~\ref{alg:trans-alg}.
For example, the quadtree-based implementation works well in low-dimensional data spaces, while the $kd$-tree-based implementation have better scalability for data dimensions.

\subsection{Branch-and-Bound Algorithm}\label{subsec:bb}

A drawback of Algorithm~\ref{alg:trans-alg} is that it needs to map $\calD$ to $\calD'$ in advance, in this subsection, we show how to do the mapping on the fly so that unnecessary computations can be avoided.

%TODO: transform on the fly and pruning rule
Recall that if instances in $I$ are sorted in ascending order according to their scores under some $f \in \calF$, then an instance $t$ will not be $\calF$-dominated by any instance $s$ after $t$.
Assuming that instances are processed in the sorted order, $S_{_V}(t)$ may only be involved in the computation of instances after $t$.
%According to this observation, we design efficient pruning strategies to tell whether an instance or a set of instances can be safely ignored during the mapping, and if so, their computations can be avoided.
If we know ignoring $S_{_V}(t)$ has no effect on those computations, then the mapping can be avoided.
Unlike conducting probabilistic rskyline analysis under top-$k$ or threshold semantics, maintaining upper and lower bounds on each instance's rskyline probability as pruning criteria is helpless since our goal is to compute exact rskyline probabilities of all instances.
Thus, the only pruning strategy can be utilized is that if an instance $t$ is $\calF$-dominated by another instance $s$ and $\Pr_{\rm rsky}(s)$ is zero, then $\Pr_{\rm rsky}(t)$ is also zero due to the transitivity of $\calF$-dominance.
A straight method for efficiently performing this pruning strategy is to keep a rskyline of all instances processed so far whose rskyline probabilities are zero and compare the next instance to be processed against all instances in the rskyline beforehand.
However, the maintained rskyline may suffer from huge scale on anti-correlated datasets.
In what follows, we prove that all instances with zero rskyline probability can be safely ignored and a set $P$ of size at most $m$ is sufficient for pruning tests.

\begin{theorem}
	All instances with zero rskyline probability can be safely discarded.
\end{theorem}
%\eat{
\begin{IEEEproof}
	Let $t \in T_i$ be an instance with $\Pr_{\rm rsky}(t) = 0$.
	Recall the formulation of rskyline probability in equation~\ref{eq:rskyprob-def}, all other instances of $T_i$ will not be affected by $t$.
	This also holds for instances of other objects $T_j$ that are not $\calF$-dominated by $t$.
	Now, suppose $s$ is an instance of $T_{j \ne i}$ and $s$ is $\calF$-dominated by $t$.
	Since $t \prec_\calF s$ and $\Pr_{\rm rsky}(t) = 0$, it is easy to see that there exists a set of objects $\mathcal{T} = \{T_k \mid k \ne j \wedge k \ne i\}$ such that all instances of each object $T_k \in \mathcal{T}$ $\calF$-dominate $t$.
	Moreover, because $\calF$-dominance is asymmetric, it is known that there exists at least one object $T_k \in \mathcal{T}$, all instances of which have non-zero rskyline probability.
	Therefore, according to the transitivity of $\calF$-dominance, $s$ is also $\calF$-dominated by all instances of $T_k$ and thus $\Pr_{\rm rsky}(s) = 0$.
\end{IEEEproof}
%}
\begin{theorem}
	Let $V = \{\omega_1, \cdots, \omega_{d'}\}$ be the set of vertices of the preference region $\Omega = \{\omega \in \simplex^{d-1} \mid A \times \omega \le b\}$, there is a set $P$ such that for any instance $t$, ${\rm Pr}_{\rm rsky}(t) = 0$ if and only if $S_{_V}(t)$ is dominated by some instance $p \in P$ and $|P| \le m$.
\end{theorem}
%\eat{
\begin{IEEEproof}
	We start with the construction of the pruning set $P$.
	For each object $T_i$ with $\sum_{t \in T_i} p(t) = 1$, we insert an instance $p_i = (\max_{t \in T_i}S_{\omega_1}(t), \cdots, \max_{t \in T_i}S_{\omega_{d'}}(t))$ into $P$.
	Note that the above construction also requires to map all instances into the score space in advance in order to facilitate the understanding of the proof.
	However, in the proposed algorithm, we construct $P$ incrementally during the computation.
	It is straight to verify that $|P| \le m$ from the construction of $P$.
	Then, let $t$ denote an instance of object $T_i$, we prove that $\Pr_{\rm rsky}(t) = 0$ if and only if $S_{_V}(t)$ is dominated by some $p_{j \ne i} \in P$.
	From equation~\ref{eq:rskyprob-def}, it is easy to see that $\Pr_{\rm rsky}(t) = 0$ if and only if there must exist an object $T_{j \ne i}$ such that every instance $s \in T_j$ $\calF$-dominates $t$ and $\sum_{s \in T_j} p(s) = 1$.
	That is $S_{_V}(s) \preceq S_{_V}(t)$ holds for all instances $s \in T_j$ according to Theorem~\ref{thm:F-dominace-V}.
	Moreover, since a set of instances dominates another instance if and only if the maximum corner of their minimum bounding rectangle dominates that instance, it is derived that $\Pr_{\rm rsky}(t) = 0$ if and only if $p_j = (\max_{s \in T_j}S_{\omega_1}(s), \cdots, \max_{s \in T_j}S_{\omega_{d'}}(s)) \preceq t$.
	Based on the construction of $P$, it is known that all $p_j$ are included in $P$, thus completing the proof.
\end{IEEEproof}
%}
\begin{figure}[!t]
    \removelatexerror 
    \begin{algorithm}[H]
    	\caption{Branch-and-Bound Algorithm}
    	\label{alg:bbs-alg}
    	\KwIn{an uncertain dataset $\calD$, a set of linear scoring functions $\calF = \{S_\omega(\cdot) \mid \omega \in \simplex^{d-1} \wedge A \times \omega \le b\}$}
    	\KwOut{ARSP}\BlankLine
    	Compute vertices $V$ of $\Omega = \{\omega \in \simplex^{d-1} \mid A \times \omega \le b\}$\;
    	Initialize a min-heap $H$ with respect to $S_{\omega}(\cdot)$ and $m$ $d'$-dimensional aggregated R-trees $R_1, \cdots, R_m$\;
    	$P \gets \emptyset$; ${\rm ARSP} \gets \emptyset$\;
    	Insert the root of R-tree on $I$ into $H$\;
    	\While{$H$ is not empty}{
    		Let $N$ be the top node in $H$\;
    		\If{$N$ is not pruned by $P$}{
    			\If{$N = \{t\}$ is a leaf node \emph{(say $t \in T_i$)}}{
    %				$t' \gets (S_{\omega_1}(t), \cdots, S_{\omega_{d'}}(t))$\;
    				Compute $S_{_V}(t)$\;
    				$\Pr_{\rm rsky}(t) \gets p(t)$\;
    				\ForEach{aggregated R-tree $R_{j \ne i}$}{
    					$\sigma[j] \gets $ perform window query with the orign and $S_{_V}(t)$ on $R_j$\;
    					$\Pr_{\rm rsky}(t) \gets \Pr_{\rm rsky}(t) \times (1 - \sigma[j])$\;
    				}
    				Insert $S_{_V}(t)$ into $R_i$\;
                    Insert $(t, \Pr_{\rm rsky}(t))$ into ${\rm ARSP}$\;
    				$p(T_i) \gets p(T_i) + p(t)$\;
    				\ForEach{$j \gets 1$ \emph{\bf to} $|V|$}{
    					$p_i[j] \gets \max(p_i[j], S_{_V}(t)[j])$\;
    				}
    				\lIf{$p(T_i) = 1$}{
    					Insert $p_i$ into $P$
    				}
    			}
    			\Else{
    				\ForEach{child node $N'$ of $N$}{
    					%\lIf{$e_i$ is not $\F$-dominated by any instance in $SP_0$}{
    					\If{$N'$ is not pruned by $P$}{
    						Insert $N'$ into $H$\;
    					}
    				}
    			}
    		}
    	}
    	\Return{${\rm ARSP}$}\;
    \end{algorithm}
\end{figure}

%TODO: algorithm description
Now, we propose an algorithm with the pruning strategy.
The pseudocode is shown in Algorithm~\ref{alg:bbs-alg}.
%Without loss of generality, it is assumed that all instances in $I$ are organized into an R-tree $R$ in advance.
The algorithm first computes the set of vertices $V$ of the preference region $\Omega$ and initializes $m$ aggregated R-trees $R_1, \cdots, R_m$, where $R_i$ is used to incrementally index $S_{_V}(t)$ for all instances $t \in T_i$ with $\Pr_{\rm rsky}(t) > 0$ that have been processed so far.
After that, the algorithm traverses the index $R$ on $I$ in a {best-first} manner.
Specifically, it first inserts the root of R-tree into a {minimum heap} $H$ sorted according to its score under some $S_{\omega \in V}(\cdot)$, where the score of a node $N$ is defined as $S_\omega(N_{\min})$.
Then, at each time, it handles the top node $N$ popped from $H$.
If $S_{_V}(N_{\min})$ is dominated by some instance in $P$, then the algorithm ignores all instances in $N$ since their rskyline probabilities are zero due to the transitivity of $\calF$-dominance.
Otherwise, if $N$ is a leaf node, say $t \in T_i$ is contained in $N$, the algorithm computes $S_{_V}(t)$ and issues the window query with the origin and $S_{_V}(t)$ on each aggregated R-tree $R_{j \ne i}$ to compute $\sigma[j] = \sum_{s \in T_j, s \prec_\calF t} p(s)$ and inserts $S_{_V}(t)$ into $R_i$.
Then it updates $p_i$, which records the maximum corner of the minimum bounding rectangle of $S_{_V}(t)$ for all instances $t \in T_i$ with $\Pr_{\rm rsky}(t) > 0$ that have been processed so far, and inserts $p_i$ into $P$ if all instances in $T_i$ have non-zero rskyline probability.
Or if $N$ is an internal node, it inserts all non-pruned child nodes of $N$ into $H$ for further computation. 

With the fact that Algorithm~\ref{alg:bbs-alg} only visits the nodes which contain instances $t$ with $\Pr_{\rm rsky}(t) > 0$ and never access the same node twice, it is easy to prove that the number of nodes accessed by Algorithm~\ref{alg:bbs-alg} is optimal to compute ARSP.
And since $m-1$ orthogonal range queries are performed on aggregated R-trees for each instance in $I$, the expected time complexity of Algorithm~\ref{alg:bbs-alg} is $O(nm\log{n})$.

\section{Sublinear-time Algorithm\\for Weight Ratio Constraints}\label{sec:eclprobalg}

In this section, we use preprocessing to accelerate ARSP computation when $\calF$ is a set of linear scoring functions whose weights are described by weight ratio constraints.
Formally, let $R = \prod^{d-1}_{i = 1}[l_i, h_i]$ denote a set of user-specified ranges, weight ratio constraints $R$ on $\simplex^{d-1}$ require $\omega[d] > 0$ and $l_i \le \omega[i]/\omega[d] \le h_i$ holds for every $1 \le i < d$.
%For brevity of notation, we will use $R = \prod^{d-1}_{i = 1}[l_i, h_i]$ to denote the set of weight ratio ranges.
Liu~\etal have investigated this special $\calF$-dominance on traditional datasets, renamed as \textit{eclipse-dominance}, and defined the \textit{eclipse} query to retrieve the set of all non-\textit{eclipse-dominated} objects~\cite{DBLP:conf/icde/Liu0ZP021}.
We refer the readers to their paper for the wide applications of eclipse query.
Although we focus on uncertain datasets, our methods can also be used to design improved algorithms for eclipse query processing as shown in our experiments.

\subsection{Reduction to Half-space Reporting Problem}
\label{sec:reduction}

Given a set of user-specified ranges $R = \prod^{d-1}_{i = 1}[l_i, h_i]$ and two instances $t$ and $s$, let $\omega^*$ be the optimal solution of the following linear programming (LP) problem,
\begin{equation}
%	\small
	\begin{aligned}
		\text{minimize} \quad & h(\omega) = \sum^d_{i = 1} (s[i] - t[i]) \times \omega[i] \\
		\text{subject to} \quad & l_i \le \omega[i]/\omega[d] \le h_i \qquad 1 \le i < d \\
		& \omega[d] > 0,  \quad \sum^d_{i = 1} \omega[i] = 1. 
	\end{aligned}
	\label{eq:lp}
\end{equation}
Under weight ratio constraints $R$, the $\calF$-dominance test condition stated in Theorem~\ref{thm:F-dominace-V} can be equivalently represented as determining whether $h(\omega^*) \ge 0$.
The crucial observation is that the sign of $h(\omega^*)$ can be determined more efficiently without solving problem~(\ref{eq:lp}).
%To be specific, since $\omega[d] > 0$, changing the object function $h(\omega) = \sum^d_{i = 1}(t[i] - s[i]) \times \omega[i]$ into $h'(r) = \sum^{d-1}_{i = 1}(t[i] -st[i]) \times r[i] + (s[d] - t[d])$, where $r[i] = \omega[i] / \omega[d]$, guarantees that $h(\omega^*) \ge 0$ iff $h'(r*) \ge 0$.
%Here $r^*$ is the optimal solution of the problem~(\ref{eq:lp}) with $h'(r)$ as the objective function.
To be specific, let $r^*$ be the optimal solution of the following LP problem,
\begin{equation}
	%	\small
	\begin{aligned}
		\text{minimize} \quad & h'(r) = \sum^{d - 1}_{i = 1} (s[i] - t[i]) \times r[i] + s[d] - t[d]\\
		\text{subject to} \quad & l_i \le r[i] \le h_i \qquad 1 \le i < d
	\end{aligned}
	\label{eq:lp2}
\end{equation}
The following lemma proves that $h'(r^*) \ge 0$ if and only if $h(\omega^*) \ge 0$.

\begin{lemma}
	Let $r^*$ and $\omega^*$ be the optimal solutions of LP problems~(\ref{eq:lp}) and~(\ref{eq:lp2}), respectively.
	$h'(r^*) \ge 0$ if and only if $h(\omega^*) \ge 0$.
\end{lemma}

\begin{IEEEproof}
	We first prove that if $h(\omega^*) \ge 0$, then $h'(r^*) \ge 0$.
	Let $\omega = (r^*[1]/(\sum^{d-1}_{i=1}r^*[i] + 1), \cdots, r^*[d - 1]/(\sum^{d-1}_{i=1}r^*[i] + 1), 1/(\sum^{d-1}_{i=1}r^*[i] + 1))$.
	For any $1 \le i < d$, $\omega[i] / \omega[d] = r^*[i] \in [l_i, h_i]$.
	And $\sum^d_{i = 1} \omega[i] = 1$.
	Therefore, $\omega$ is a feasible solution of LP problem~(\ref{eq:lp}).
	Hence, $h'(r^*) = (\sum^{d-1}_{i = 1} r^*[i] + 1) \times h(\omega) \ge (\sum^{d-1}_{i = 1} r^*[i] + 1) \times h(\omega^*) \ge 0$.
	
	Next, we prove that if $h'(r^*) \ge 0$, then $h(\omega^*) \ge 0$.
	Let $r = (\omega^*[1]/\omega^*[d], \cdots, \omega^*[d-1]/\omega^*[d])$.
	For any $1 \le i < d$, $\omega^*[i]/\omega^*[d] \in [l_i, h_i]$.
	Hence, $r$ is a feasible solution of LP problem~(\ref{eq:lp2}).
	Thus, $h(\omega^*) = \omega^*[d] \times h'(r) \ge \omega^*[d] \times h'(r^*) \ge 0$.
\end{IEEEproof}

Since each $r[i] = \omega[i] / \omega[d]$ can be choose independently from the corresponding interval $[l_i, h_i]$, $r^*$ can be directly determined in $O(d)$ time.
Based on this, we can perform $\calF$-dominance test more efficiently.

\begin{theorem}[Efficient $\calF$-dominance test]
    Let $\calF$ be a set of linear scoring functions whose weights are described by weight ratio constraints $R = \prod^{d-1}_{i = 1}[l_i, h_i]$, an instance $t$ $\calF$-dominates another instance $s$ if and only if
    %\[\sum^{d-1}_{i = 1}(\mathbf{1}[s[i] > t[i]] \times l_i + (1 - \mathbf{1}[s[i] > t[i]]) \times h_i)\times (s[i] - t[i]) \ge t[d] - s[d],\]
    %\begin{equation*}
	%	\begin{aligned}
	%		&t[d] - s[d] \le \\
	%		&\sum^{d-1}_{i = 1}(\mathbf{1}[s[i] > t[i]] \times l_i + (1 - \mathbf{1}[s[i] > t[i]]) \times h_i)\times (s[i] - t[i]) 
	%	\end{aligned}
    %\end{equation*}
    $t[d] - s[d] \le \sum^{d-1}_{i = 1}(\mathbf{1}[s[i] > t[i]] \times l_i + (1 - \mathbf{1}[s[i] > t[i]]) \times h_i)\times (s[i] - t[i])$,
    where $\mathbf{1}[\cdot]$ is the indicator function.
	\label{thm:R-dominance}
\end{theorem}

%Now, consider the set $I$ of all instances as a set of points in the data space $[0, 1]^d$, the $i$-th attributes as coordinates in the $i$-th dimension.
According to Theorem~\ref{thm:R-dominance}, we present a reduction of finding all instances in $I$ that $\calF$-dominate instance $t$ to a series of $2^{d-1}$ half-space reporting problem~\cite{agarwal2017simplex}, which aims to preprocess a set of points in $\real^d$ into a data structure so that all points lying below or on a query hyperplane can be reported quickly.
%The above procedure actually reduces the problem of finding all instances that $\calF$-dominate $t$ to a series of $2^{d-1}$ half-space range searching problem~\cite{agarwal2017simplex}.
%Consider an instance $t$.
We partition the data space $\mathbb{R}^d$ into $2^{d-1}$ regions using $d - 1$ hyperplanes $x[i] = t[i]$ ($1 \le i < d$).
Then, each resulted region can be identified by a $(d - 1)$-bit code such that the $i$-th bit is {0} if the $i$-th attributes of instances in this region are less than $t[i]$, and {1} otherwise.
We refer the region whose identifier is $k$ in decimal as region $k$.
Suppose $t \in T_i$, let $I_{t, k}$ denote the set of instances of other uncertain objects contained in region $k$.
As an example, $I_{t, 0} = \{s \in I \setminus T_i \mid \forall 1 \le i < d. s[i] < t[i]\}$.
It is easy to verify for each $0 \le k < 2^{d-1}$, when performing $\calF$-dominance test between instances in $I_{t, k}$ and $t$, the results of $d-1$ indicator functions in Theorem~\ref{thm:R-dominance} are identical for all instances in $I_{t,k}$.
Geometrically, all instances in $I_{t, k}$ that $\calF$-dominate $t$ must lie below or on the following hyperplane,
\begin{equation}
	h_{t, k} : x[d] = \sum^{d-1}_{i=1}((1 - \lvert k \rvert_2[i])\times l_i + \lvert k \rvert_2[i]\times h_i) \times (t[i] - x[i]) + t[d],
	\label{eq:half-space}
\end{equation}
where $\lvert k \rvert_2 [i]$ is $i$-th bit of the binary of number $k$.

\begin{figure}
	\subfigure[{Partition data space $\mathbb{R}^2$ with $t_{2,3}$, and two hyperplanes $h_{t_{2,3}, 0}, h_{t_{2,3}, 1}$ of $t_{2,3}$.}]{
		\includegraphics[width=.42\linewidth]{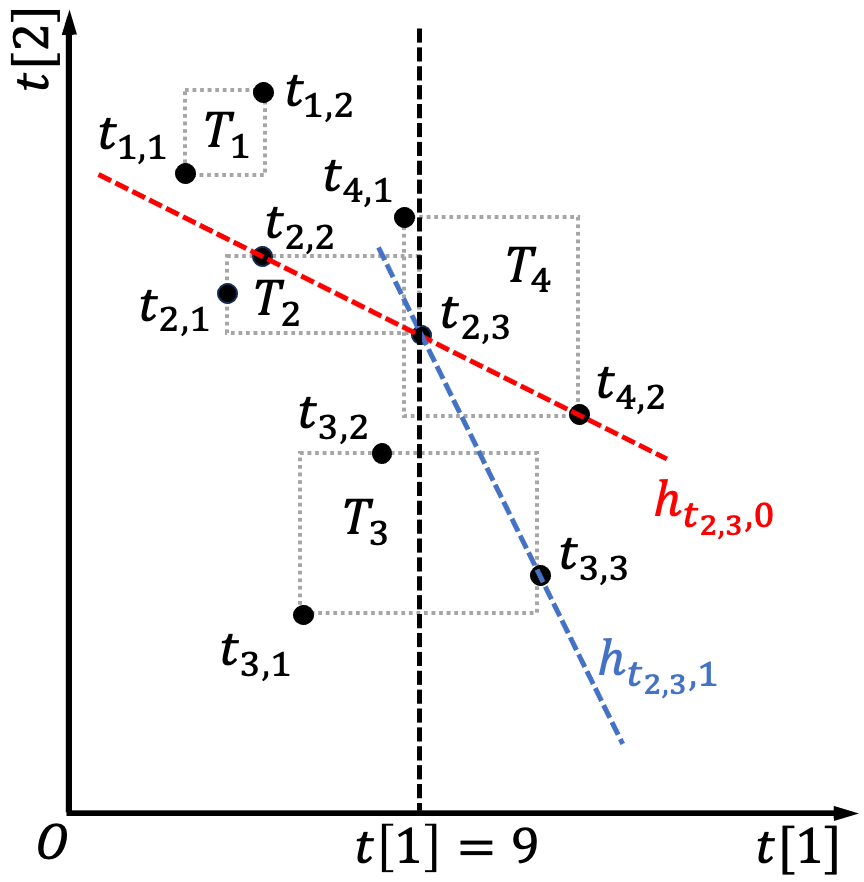}
		\label{fig:range-search}
	}\hspace{1ex}
	\subfigure[{Dual hyperplanes in $I^*_{t_{2,3}, 0}$, two faces $f_1, f_2$ of $\mathcal{A}(I^*_{t_{2,3}, 0})$, and dual point $h^*_{t_{2,3}, 0}$.}]{
		\includegraphics[width=.42\linewidth]{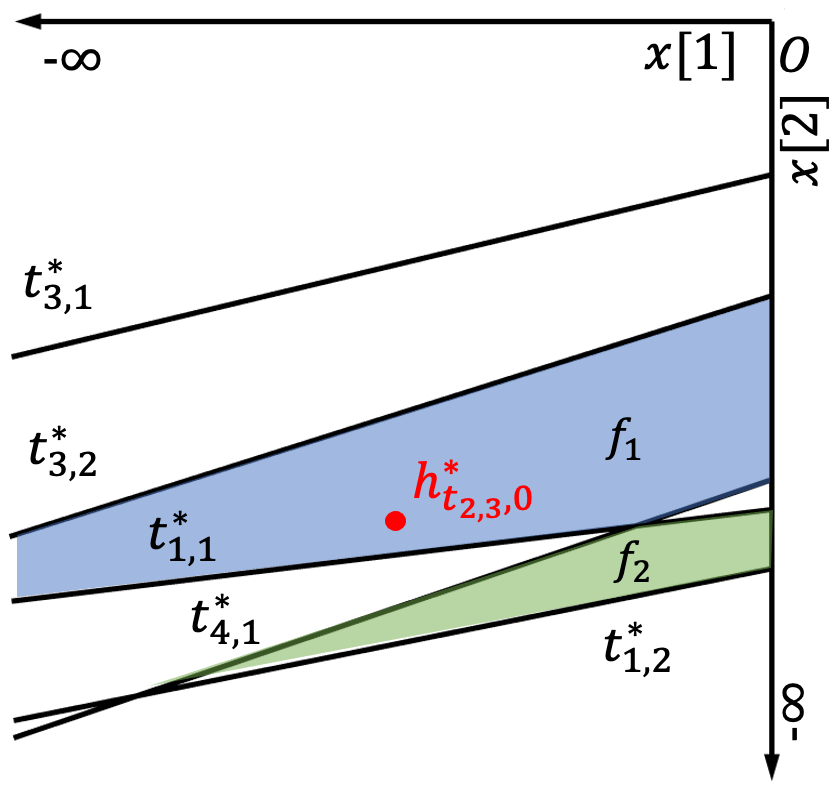}
		\label{fig:point-lcation}
	}\vspace{-2mm}
	\caption{{An illustration of the reduction to half-space reporting problem and performing point location queries in dual space.}}
	\vspace{-2mm}
	\label{fig:hp-dual}
\end{figure}

\begin{example}\label{ep1}
	%See Fig.~\ref{fig:range-search} for an illustration of the reduction.
	The uncertain dataset used in Example~\ref{ep} is plotted in Fig.~\ref{fig:range-search}.
	Consider instance $t_{2, 3}$.
	Data space $\mathbb{R}^2$ is partitioned with the line $t[1] = t_{2,3}[1] = 9$.
	Region 0 contains the set of instances $I_{t_{2,3}, 0} = \{s \in I \setminus T_2 \mid s[1] \le t_{2,3}[1]\}= \{t_{1, 1}, t_{1, 2}, t_{3, 1}, t_{3, 2}, t_{4,1}\}$, and region 1 contains the set of instances $I_{t_{2, 3}, 1} = \{t_{3,3}, t_{4,2}\}$.
%	Note that since other instances of $T_2$ do not affect $\Pr_{\rm rsky}(t_{2, 3})$, we omit them in $I_{t_{2, 3}, 0}$ and $I_{t_{2,3}, 1}$.
	Given weight ratio constraints $R = [0.5, 2]$, according to equation~(\ref{eq:half-space}), the hyperplane $h_{t_{2,3}, 0}$ in region 0 is $t[2] = -0.5t[1] + 16.5$ and the hyperplane $h_{t_{2,3}, 1}$ in region 1 is $t[2] = -2t[1] + 30$. 
	Since $t_{3, 1}$ and $t_{3,2}$ ly below or on $h_{t_{2, 3}, 0}$, we know $t_{3, 1} \prec_\calF t_{2, 3}$ and $t_{3,2} \prec_\calF t_{2,3}$.
	And since $t_{3,3}$ lies on $h_{t_{2,3}, 1}$, we know $t_{3,3} \prec_\calF t_{2,3}$.
\end{example}

%The above procedure actually reduces the problem of finding all instances that $\calF$-dominate $t$ to a series of $2^{d-1}$ half-space range searching problem
The half-space reporting problem can be efficiently solved using the well-known \textit{point-hyperplane duality}~\cite{mark2008computational}.
The duality maps a point $p = (p[1], \cdots, p[d]) \in \real^d$ into the hyperplane $p^* : x[d] = p[1]x[1] + \cdots + p[d-1]x[d-1] - p[d]$, and a hyperplane $h : x[d] = \alpha[1]x[1] + \cdots + \alpha[d-1]x[d-1] - \alpha[d]$ into the point $h^* = (\alpha[1] \cdots, \alpha[d])$.
It is proved that if $p$ lies above (\resp, below, on) $h$, then $h^*$ lies above (\resp, below, on) $p^*$.
Thus, the dual version of the half-space reporting problem becomes that given a set of $n$ hyperplanes in $\real^d$ and a query point $q$, report all hyperplanes lying above or through $q$.
Let $H$ be the set of $n$ hyperplanes in $\real^d$, the {\it arrangement} of $H$, denoted by $\mathcal{A}(H)$, is a subdivision of $\real^d$ into {\it faces} of dimension $k$ for $0 \le k \le d$.
Each face in $\mathcal{A}(H)$ is a maximal connected region of $\real^d$ that lies in the same subset of $H$.
For a query point $q$, let $\lambda(q, H)$ denote the set of hyperplanes in $H$ lying above or through $q$.
It is easy to verify that all points $p$ lying on the same face $f$ of $\mathcal{A}(H)$ have the same $\lambda(p, H)$, denoted by $\lambda(f, H)$.
Thus, with a precomputation of $\lambda(f, H)$ for each face $f$ of $\mathcal{A}(H)$ and the following structure for point location in $\mathcal{A}(H)$, $\lambda(q, H)$ can be found in logarithmic time.

\begin{theorem}[Structure for Point Location~\cite{DBLP:journals/iandc/Meiser93}]
	Given a set $H$ of $n$ hyperplanes in $\real^d$ and a query point $q$, there is a data structure of size $O(n^{d + \varepsilon})$ which can be constructed in $O(n^{d+\varepsilon})$ expected time for any $\varepsilon > 0$, so that the face of $\mathcal{A}(H)$ containing $q$ can be located in $O(\log{n})$ time.
	\label{thm:fast-point-location}
\end{theorem}

%Based on this efficient algorithm for finding all instances that $\calF$-dominate $t$, we design an improved algorithm for the ARSP problem.
Accordingly, after building the point location structure for the set of dual hyperplanes of each $I_{t, k}$, by locating the dual point of $h_{t, k}$, we can find all instances in $I_{t, k}$ that $\calF$-dominate $t$ efficiently.
However, according to equation~(\ref{eq:rskyprob-def}), in order to compute $\Pr_{\rm rsky}(t)$, we need to further calculate the cumulative probability of instances from the same uncertain object.
In what follows, we propose an efficient algorithm to compute ARSP by modifying the above algorithm.

In the preprocessing stage, for each instance $t \in I$, say $t \in T_i$, the algorithm partitions $I \setminus T_i$ into $2^{d-1}$ sets $I_{t, k} = \{s \in I \setminus T_i \mid s$ in region $k$ derived by partitioning $[0, 1]^d$ with $t\}$ ($0 \le k < 2^{d-1}$).
Then, for each set $I_{t, k}$, it computes the set of dual hyperplanes $I^*_{t, k} = \{s^* \mid s \in I_{t, k}\}$ and builds the point location structure on $I^*_{t, k}$.
Finally, it constructs $\mathcal{A}(I^*_{t, k})$ and records an array $\sigma_f = \langle \sigma_f[1], \cdots, \sigma_f[m] \rangle$ for each face $f$ of $\mathcal{A}(I^*_{t, k})$,  where
$\sigma_f[j] = \sum_{s^* \in \lambda(f, I^*_{t, k}) \wedge s \in T_j}p(s)$,
\ie, the sum of probabilities over all instances of object $T_j$ lying below or on the hyperplane $p^*$, where $p^*$ is the dual hyperplane of some point $p$ lying in face $f$.

In the query processing stage, given weight ratio constraints $R = \prod^{d-1}_{i = 1}[l_i, h_i]$, the algorithm processes each instance $t$ as follows.
It first initializes $\Pr_{\rm rsky}(t) = p(t)$ and $\sigma[i] = 0$ for $1 \le i \le m$, where $\sigma[i]$ is for recording the sum of existence probabilities of instances from object $T_i$ that $\calF$-dominate $t$ found so far.
Then, for each $0 \le k < 2^{d-1}$, it compute the dual point $h^*_{t, k}$ of the hyperplane $h_{t, k}$ defined in equation~(\ref{eq:half-space}), and performs point location query $h^*_{t, k}$ on the structure built on $I^*_{t, k}$.
Let $f$ be the face returned by the point location query $h^*_{t, k}$, for $1 \le j \le m$, it updates $\Pr_{\rm rsky}(t)$ to $\Pr_{\rm rsky}(t) \times (1 - \sigma[j] - \sigma_f[j])/(1 - \sigma[j])$ and adds $\sigma_f[j]$ to $\sigma[j]$.
After all queries, it returns $\Pr_{\rm rsky}(t)$ as the final rskyline probability of $t$.
Since each point location query can be performed in $O(\log{n})$ time and the update of $\Pr_{\rm rsky}(t)$ requires $O(m)$ time for each $\sigma_f$, the time complexity of this algorithm is $O(2^dmn\log{n})$.

\begin{example}
	Continue with Example~\ref{ep1}.
	For instance $t_{2, 3}$, in the preprocessing stage, the algorithm will compute dual hyperplanes of $I_{t_{2,3}, 0}$ and $I_{t_{2,3}, 1}$, build point location structures on $I^*_{t_{2,3}, 0}$ and $I^*_{t_{2,3}, 1}$, and record $\sigma_f$ for each face $f$ of $\mathcal{A}(I^*_{t_{2,3},0})$ and $\mathcal{A}(I^*_{t_{2,3},1})$.
	As an example, hyperplanes in $I^*_{t_{2,3}, 0}$ are plotted in Fig.~\ref{fig:point-lcation}.
	For face $f_1$, it records $\sigma_{f_1}[1] = \sigma_{f_1}[4] = 0$, $\sigma_{f_1}[3] = p(t_{3,1}) + p(t_{3,2}) = 2/3$ since $t^*_{3,1}$ and $t^*_{3,2}$ lie above or through every point in $f_1$ and for face $f_2$ it records $\sigma_{f_2}[1] = p(t_{1,1}) = 1/2$, $\sigma_{f_2}[3] = p(t_{3,1}) + p(t_{3,2}) = 2/3$, $\sigma_{f_2}[4] = p(t_{4,1}) = 1/2$ since $t^*_{1, 1}, t^*_{3,1}, t^*_{3,2}$, and $t^*_{4,1}$ lie above or through every point in $f_1$.

	Then, given weight ratio constraints $R = [0.5, 2]$, to compute $\Pr_{\rm rsky}(t_{2,3})$, the algorithm first initializes $\Pr_{\rm rsky}(t_{2,3}) = p(t_{2,3}) = 1/3$ and $\sigma[i] = 0$ ($1 \le i \le 4$), then performs point location query $h^*_{t_{2,3}, 0}$ and $h^*_{t_{2,3}, 1}$ on $I^*_{t_{2,3}, 0}$ and $I^*_{t_{2,3}, 1}$ respectively to update $\Pr_{\rm rsky}(t_{2,3})$ and $\sigma[i]$ ($1 \le i \le 4$).
	By locating $h^*_{t_{2,3}, 0} = (-0.5, -16.5)$, which is the dual point of $h_{t_{2,3}, 0} : t[2] = -0.5t[1] + 16.5$, face $f_1$ is returned.
	Since only $\sigma_{f_1}[3] \ne 0$, the algorithm updates $\Pr_{\rm rsky}(t_{2,3})$ to $\Pr_{\rm rsky}(t_{2, 3}) * (1 - \sigma[3] - \sigma_{f_1}[3])/(1 - \sigma[3]) = 1/9$ and updates $\sigma[3]$ to $\sigma[3] + \sigma_{f_1}[3] = 2/3$.
\end{example}

\subsection{Sublinear-time Algorithm}

To achieve a sublinear query time, the following two bottlenecks of the above algorithm should be addressed.
First, for each instance, $2^{d-1}$ arrays are accessed and it seems unrealistic to merge them efficiently based on equation~(\ref{eq:rskyprob-def}).
Second, since instances are sequentially processed, the query time can not be less than $n$. 
%Since all of them should be scanned at least once, the time complexity is $\Omega(n)$.
In subsequent, we introduce two strategies to overcome these two inefficiencies.

%\subsubsection*{\textbf{Multi-level strategy}}
\noindent{\bf Multi-level strategy.}
The reason why the above algorithm has to access $2^{d-1}$ arrays for each instance $t$ is that the query point $h^*_{t, k}$ is different for each $I^*_{t, k}$.
Hence, it needs to perform $2^{d-1}$ different point location queries to retrieve $\sigma_f$ from each $I_{t, k}$.
Different from general linear constraints, the reduction ensures that given weight ratio constraints $R = \prod^{d - 1}_{i = 1}[l_i, h_i]$, the number of point location queries performed for each instance is always $2^{d-1}$.
In this case, we show how to resolve this issue with the help of \textit{multi-level strategy}~\cite{DBLP:journals/ipl/Bentley79}.

%The algorithm recursively constructs a $2^{d-1}$-level structure on the set of dual hyperplanes $I^* = \{t^* \mid t \in I\}$.
A $2^{d-1}$-level structure on the set of dual hyperplanes $I^* = \{t^* \mid t \in I\}$ is a recursively defined point location structure.
%Each level of the data structure is a point location structure on one of $I^*_{t, k}$ ($0 \le k < 2^{d-1}$).
To be specific, an one-level structure is a point location tree built on $I^*$.
And each face $f$ of $\mathcal{A}(I^*)$ records the following information:
(1) an array $\sigma_f = \langle \sigma_f[j] \mid 1 \le j \le m \rangle$, where $\sigma_f[j] = \sum_{s^* \in \lambda(f, I^*) \wedge s \in T_j}p(s)$, \ie, the sum of probabilities over all instances of object $T_j$ lying below or on the hyperplane $p^*$, where $p^*$ is the dual hyperplane of some point $p$ lying in face $f$,
(2) a product $\beta_f = \prod^m_{j = 1, \sigma_f[j] \ne 1}(1 - \sigma_f[j])$, and
(3) a count $\chi_f = |\{j \mid \sigma_f[j] = 1\}|$ are also recorded for each face $f \in \mathcal{A}(I^*)$.
A $k$-level structure is an one-level structure built on $I^*$ and each face $f$ of $\mathcal{A}(I^*)$ additionally contains an associated $(k-1)$-level structure built on $\lambda(f, I^*)$.

After constructing the $2^{d-1}$-level structure on $I^*$, the algorithm processes weight ratio constraints $R = \prod^{d-1}_{i = 1}[l_i, h_i]$ as follows.
For each instance $t$, it initializes hyperplanes set $H$ as $I^*$ and integer $k$ as zero.
While $k < 2^{d-1}$, it generate the dual point $h^*_{t, k}$ according to equation~(\ref{eq:half-space}) and performs point location query $h^*_{t, k}$ on structure built on $H$.
Then, let $f$ be the returned face, it updates $H$ as $\lambda(f, H)$ and $k$ as $k + 1$.
Note that query $h^*_{t, k}$ helps to find all instances lying below or on $h_{t, k}$ in the result of the first $k - 1$ queries.
Let $f$ be the last face returned.
According to the information recorded for $f$, $\Pr_{\rm rsky}(t)$ is calculated as follows.
If $\chi_f = 0$, then $\Pr_{\rm rsky}(t) = \beta_f \cdot p(t)/(1 - \sigma_f[i])$, or if $\chi_f = 1 \wedge \sigma_f[i] = 1$, $\Pr_{\rm rsky}(t) = \beta_f \times p(t)$, otherwise, $\Pr_{\rm rsky}(t) = 0$.
Since for each instance $t \in I$,  $\Pr_{\rm rsky}(t)$ can be computed in constant time after performing $2^{d-1}$ point location queries, the total time complexity of the multi-level structure based algorithm for ARSP computation is $O(2^{d-1}n\log{n})$ time.
%This also leads to an algorithm with logarithmic query time and polynomial preprocessing time for rskyline probability query.
%The formal definition of rskyline probability query is given as follows.

%\noindent{\bf RSkyline Probability Query (RSPQ)}\\
%\noindent{\bf Input:} an uncertain dataset $\calD$, a set of weight ratio ranges $R$, and a query instance $q$. \\
%\noindent{\bf Output:} for the query instance $q$, the probability that $q$ is not $\calF$-dominated by any instances in $I$, \ie,
%\[{\rm Pr}_{\rm rsky}(q)= \prod^m_{i = 1} (1 - \sum_{t \in T_i, t \prec_\calF q}\Pr(t)).\]

%\begin{theorem}
%	\emph{RSPQ} belongs to the complexity class $\mathrm{PsL}${\rm ~\cite{DBLP:journals/tcs/GaoLML20}}.
%\end{theorem}

%\subsubsection*{\textbf{Shift strategy}}
\noindent{\bf Shift strategy.}
The major obstacle to the second bottleneck is that the $2^{d-1}$ point location queries $h^*_{t, k}$ ($0 \le k < 2^{d-1}$) are different for each instance $t$.
Geometrically speaking, $R = \prod^{d-1}_{i = 1}[l_i, h_i]$ is a $(d-1)$-dimensional hyper-rectangle.
Let $V = \{v = (v[1], \cdots, v[d-1]) \mid \forall 1 \le i < d, v[i] \in \{l_i, h_i\}\}$ be the set of $R$'s vertices.
For $0 \le k < 2^{d-1}$, we call a vertex $v \in V$ the $k$-vertex of $R$ if there are $k$ vertices before $v$ in the lexicographical order of vertices in $V$.
For example, $(l_1, \cdots, l_{d-1})$ is the 0-vertex of $R$ and $(h_1, \cdots, h_{d-1})$ is the $(2^{d-1}-1)$-vertex of $R$.
Let $v_k$ denote the $k$-vertex of $R$.
According to equation~(\ref{eq:half-space}), $h^*_{t, k} = (-v_k[1], \cdots, -v_k[d-1], -(\sum^{d-1}_{i=1}v_k[i]t[i] + t[d]))$.
For any two instances $t$ and $s$, each pair $h*_{t, k}$ and $h^*_{s, k}$ differs only in the last dimension.
In what follows, we introduce the \textit{shift strategy} to unify the procedures of performing point location queries for all instances by making their $h^*_{t, k}[d]$ the same for each $0 \le k < 2^{d-1}$.

Specifically, for each instance $t \in I$, say $t \in T_i$, the algorithm first creates a shifted dataset $I_t$ by treating $t$ as the origin, \ie, $I_t = \{s - t \mid s \in I \setminus T_i\}$.
Then, it merges all sets $I_t$ into a key-value pair set $\mathcal{I} = \{(s, \langle t \mid s \in I_t \rangle) \mid s \in \bigcup_{t \in I} I_t\}$.
Finally, it construct a $2^{d-1}$-level structure on the set of dual hyperplanes $\mathcal{I}^* = \{s^* \mid (s, -) \in \mathcal{I}\}$ as stated above, except that the information recorded in the one-level structure for each face $f$ of $\mathcal{A}(\mathcal{I}^*)$ is redefined as $\Pr_f = \langle \Pr_f[t] \mid t \in I \rangle$, where ${\rm Pr}_f[t] = \prod^m_{j= 1, j \ne i}(1 - \sum_{s^* \in \lambda(f, \mathcal{I}^*) \wedge s + t \in T_j}p(s))$.

Given constraints $R = \prod^{d-1}_{i=1}[l_i, h_i]$, the algorithm generates $2^{d-1}$ point location queries $h^*_k  = (v_k[1], \cdots, v_k[d-1], 0)$ ($0 \le k < 2^{d-1})$ and executes them on the $2^{d-1}$-level structure built on $\mathcal{I}^*$.
Let $f$ be the last face returned, $\Pr_{\rm rsky}(t)$ of each instance $t \in I$ is computed as $p(t) \times \Pr_f(t)$.
Since the algorithm performs a total of $2^{d-1}$ point location queries, the query time is at most $O(2^{d-1}\log{n} + n)$, where the additional linear time is required for reporting the final result.
% all point location queries can be executed in $O(2^{d-1}\log{n})$ time.
%And after that, ARSP can be reported in linear time with respect to the result size.

\section{Experiments}\label{sec:experiments}

In this section, we present the experimental study for the ARSP problem.

\subsection{Experimental Setting}

\noindent{\bf Datasets.}
We use both real and synthetic datasets for expe-riments.
The real data includes three datasets.
IIP~\cite{DBLP:journals/vldb/JinYCYL10} contains 19,668 sighting records of icebergs with 2 attributes: melting percentage and drifting days.
Each record in IIP has a confidence level according to the source of sighting, including R/V (radar and visual), VIS (visual only), RAD (radar only).
We treat each record as an uncertain object with one instance and convert these three confidence levels to probabilities 0.8, 0.7, and 0.6 respectively.
%CAR~\cite{DBLP:journals/tkde/LiuYYL13} contains 310,345 cars with 4 attributes: price, power, mileage, registration year.
CAR~\cite{DBLP:journals/tkde/LiuYYL13} contains 184,810 cars with 4 attributes: price, power, mileage, registration year.
To convert CAR into an uncertain dataset, we organize cars with the same model into an uncertain object $T$ and for each $t \in T$, we set $p(t) = 1/|T|$, \ie, when a customer wants to rent a specific model of car, any car of that model in the dataset will be offered with equal probability.
%NBA~\cite{DBLP:journals/tkde/KimIP12} includes 994,031 game records of 3,502 players with 7 metrics: points, assists, steals, blocks, rebounds, minutes, field goals made.
%NBA~\cite{DBLP:journals/tkde/KimIP12} includes 547,125 game records of 3,287 players with 7 metrics: points, assists, steals, blocks, rebounds, minutes, field goals made.
NBA~\cite{DBLP:journals/tkde/KimIP12} includes 354,698 game records of 1,878 players with 8 metrics: points, assists, steals, blocks, turnovers, rebounds, minutes, field goals made.
%NBA~\cite{DBLP:journals/tkde/KimIP12} includes 23,160 season records for 3,505 players with 8 metrics: points, assists, steals, blocks, turnovers, rebounds, minutes, field goals made.
We treat each player as an object $T$ and each record of the player as an instance $t$ of $T$ with $p(t) = 1/|T|$.
%To covert CAR and NBA into uncertain datasets, for each object $T$, we assign each instance of $T$ the existence probability $1/|T|$.
%contains 28,475 technical statistics (season records) of 3,707 players with 5 professional metrics: points, assists, rebounds, steals, and blocks, extracted from \url{https://www.nba.com/stats/}.
%We consider each player as an uncertain object and his season records as instances with the same existence probability of that object.

The synthetic datasets are generated with the same procedure described in~\cite{DBLP:conf/vldb/PeiJLY07, DBLP:conf/pods/AtallahQ09, DBLP:journals/tkde/KimIP12, DBLP:conf/cikm/LiuZXLL15}.
%We first use the standard data generation tool~\cite{DBLP:conf/icde/BorzsonyiKS01} to generate $m$ centers $c_1, \cdots , c_m$ in $[0, 1]^d$ based on independent (IND) or anti-correlated (ANTI) distribution.
Let $m$ be the number of uncertain objects.
%For $1 \le i \le m$, we first generate center $c_i$ of object $T_i$ in $[0, 1]^d$ according to independent (IND) or anti-correlated (ANTI) distribution~\cite{DBLP:conf/icde/BorzsonyiKS01}.
%Due to space limitations, we do not include correlated synthetic datasets as they are the least challenging for dominance-based queries~\cite{DBLP:journals/tkde/KimIP12, DBLP:journals/pvldb/CiacciaM17}.
For $1 \le i \le m$, we first generate center $c_i$ of object $T_i$ in $[0, 1]^d$ following independent (IND), anti-correlated (ANTI), or correlated (CORR) distribution~\cite{DBLP:conf/icde/BorzsonyiKS01}.
Then, we generate a hyper-rectangle $R_i$ centered at $c_i$.
And all instances of $T_i$ will appear in $R_i$.
The edge length of $R_i$ follows a normal distribution in range $[0, l]$ with expectation $l/2$ and standard deviation $l/8$.
% center, we construct a hyper-rectangle whose edge length follows a normal distribution $\mathscr{N}(l/2, l/8)$.
And the number of $T_i$'s instances follows a uniform distribution over interval $[1, cnt]$.
We generate $n_i$ instances uniformly within $R_i$ and assign each instance the existence probability $1/n_i$.
Finally, we remove one instance from the first $\phi \times m$ objects so that for any $1 \le j \le \phi \times m$, $\sum_{t \in T_j} p(t) < 1$.
Therefore, the expected number of instances in a synthetic dataset is $({cnt}/{2} - \phi)\times m$.
%Most experimental settings follow those in~\cite{DBLP:conf/vldb/PeiJLY07, DBLP:journals/tkde/KimIP12}.

\noindent{\bf Constraints.}
Our experiments consider two methods to generate linear constraints on weights.
WR specifies weak rankings on weight attributes~\cite{DBLP:journals/cor/EumPK01}.
Given the number of constraints $c$, it requires $\omega[i] \ge \omega[i + 1]$ for every $1 \le i \le c$.
IM generates constraints in an interactive manner~\cite{DBLP:journals/pvldb/QianGJ15}.
Specifically, it first chooses a weight $\omega^*$ randomly in $\simplex^{d-1}$.
Then, for each $1 \le i \le c$, it generates two objects $t_i, s_i$ uniformly in $[0, 1]^d$, divide $\simplex^{d-1}$ into two subspaces with $\sum^d_{j = 1}(t_i[j] - s_i[j]) \times \omega[j] = 0$, and selects the one containing $\omega^*$ as the $i$-th input constraint.
The main difference between these two methods is that the preference region generated by WR always has $d$ vertices, while the number of vertices of the preference region generated by IM usually increases with $c$.

%To eliminate the bias of the generated datasets and constraints, we repeat all experiments 10 times for each parameter configuration and report the average as final results.

\noindent{\bf Algorithms.}
We implement the following algorithms in C++ and the source code is available at~\cite{github}.
\begin{enumerate}[$\bullet$]
	\item {ENUM}: the first baseline algorithm in Section~\ref{subsec:bsl}.
	\item {LOOP}: the second baseline algorithm in Section~\ref{subsec:bsl}.
	\item {KDTT}: the kdtree-traversal algorithm in Section~\ref{subsec:tt}.
	\item {KDTT+}: the kdtree-traversal algorithm incorporating preorder traversal into tree construction in Section~\ref{subsec:tt}. 
	\item {QDTT+}: the quadtree-traversal algorithm incorporating preorder traversal into tree construction in Section~\ref{subsec:tt}. 
	\item {B\&B}: the branch-and-bound algorithm in Section~\ref{subsec:bb}.
	\item {DUAL} ({-M/S}): the dual-based algorithm in Section~\ref{sec:eclprobalg}, where {-M} is for multi-level strategy, {-S} is for shift strategy.
\end{enumerate}
%All algorithms are complied by GNU G++ 7.5.0 with -O2 optimization and 
All experiments are conducted on a machine with a 3.5-GHz Intel(R) Core(TM) i9-10920X CPU, 256GB main memory, and 1TB hard disk running CentOS 7.

\subsection{Effectiveness of ARSP}

We verify the effectiveness of ARSP on the NBA dataset.
To facilitate analysis, we extract game records in 2021 from NBA and consider 3 attributes for each player: rebound, assist, and points.
We still treat each player as an object $T$ and each record of the player as an instance $t$ of $T$ with $p(t) = 1/|T|$.
We set $\calF = \{\omega[1] \text{Rebound} + \omega[2] \text{Assist} + \omega[3] \text{Point} \mid \omega[1] \ge \omega[2] \ge \omega[3]\}$.
Table~\ref{tab:top-14} reports the top-14 players in rskyline probability ranking along with their rskyline probabilities.
%As a comparison, we also conduct the traditional rskyline analysis by computing the average statistics for each player and retrieving rskyline on the aggregated dataset, which is called aggregated rskyline for short hereafter.
As a comparison, we also conduct the traditional rskyline query on the aggregated dataset, which is obtained by computing the average statistics for each player.
We call the result aggregated rskyline for short hereafter and mark players in the aggregated rskyline with a ``*'' sign in Table~\ref{tab:top-14}.
%We plot score distributions of four players in Fig.~\ref{fig:score}.}

\begin{table}[ht]
	\centering
	\caption{{Top-14 players in rskyline probability ranking.}}
	\label{tab:top-14}
	\begin{tabular}{|c|c|c|c|}
		\hline
		\textbf{Player} & $\Pr_{\rm rsky}$ & \textbf{Player} & $\Pr_{\rm rsky}$ \\ \hline \hline
		* Russell Westbrook & 0.349 & * Rudy Gobert & 0.142 \\ \hline
		* Nikola Jokic & 0.331 & * Clint Capela  & 0.134 \\ \hline
		Giannis Antetokounmpo & 0.292 & Nikola Vucevic & 0.126 \\ \hline
		James Harden & 0.213 & Andre Drummond & 0.109 \\ \hline
		Joel Embiid & 0.186 & Julius Randles & 0.109 \\ \hline
		Luka Doncic & 0.168 & Kevin Durant & 0.101 \\ \hline
		* Domantas Sabonis & 0.162 & * Jonas Valanciunas & 0.095 \\ \hline
	\end{tabular}
\end{table}

\begin{figure}[ht]
	\centering
	\subfigure[$S_{\omega_1}(\cdot)$]{
		\includegraphics[width=.3\linewidth]{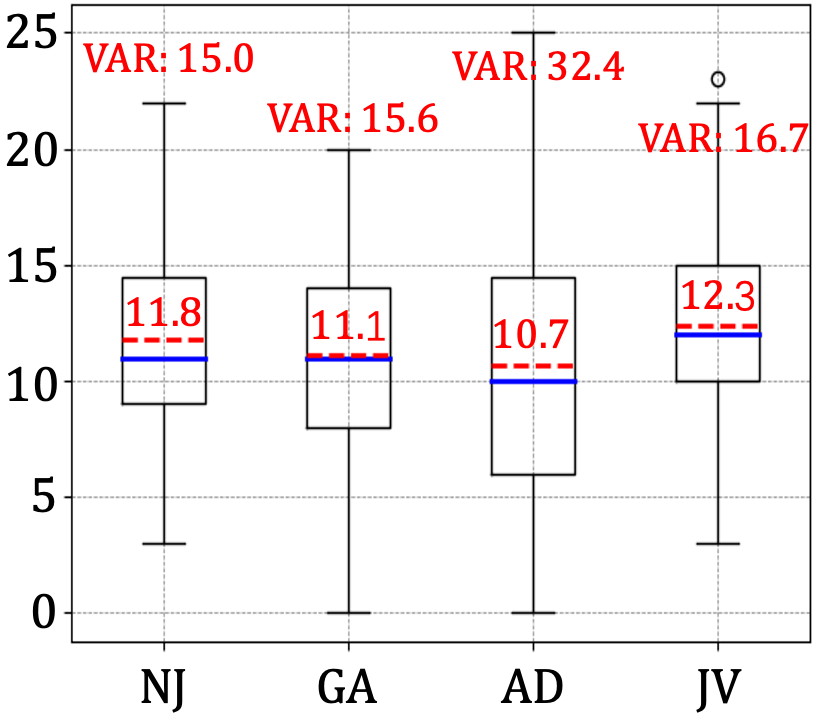}
		\label{fig:score1}
	}\hspace{-1ex}
	\subfigure[$S_{\omega_2}(\cdot)$]{
		\includegraphics[width=.3\linewidth]{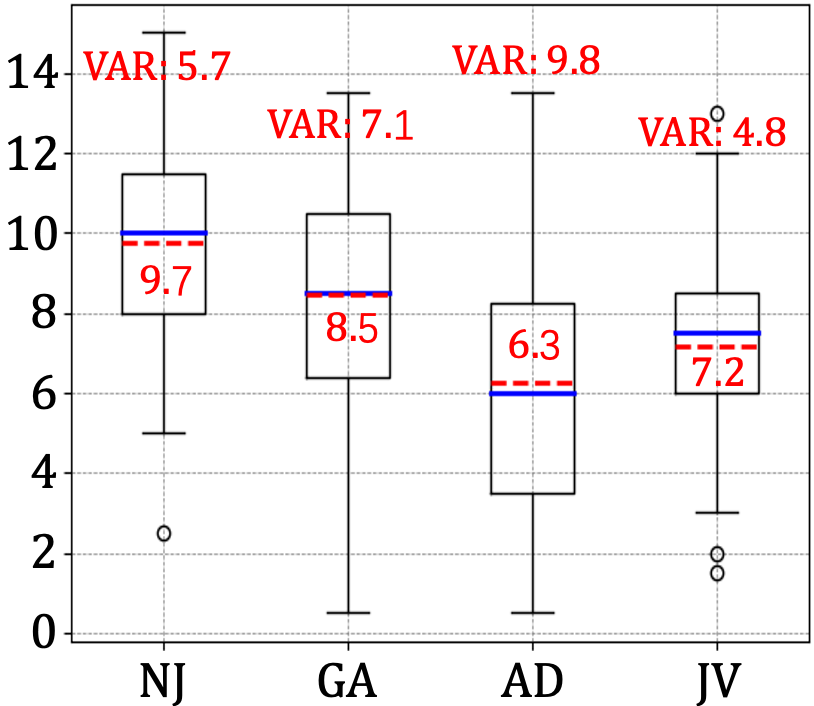}
		\label{fig:score2}
	}\hspace{-1ex}
	\subfigure[$S_{\omega_3}(\cdot)$]{
		\includegraphics[width=.3\linewidth]{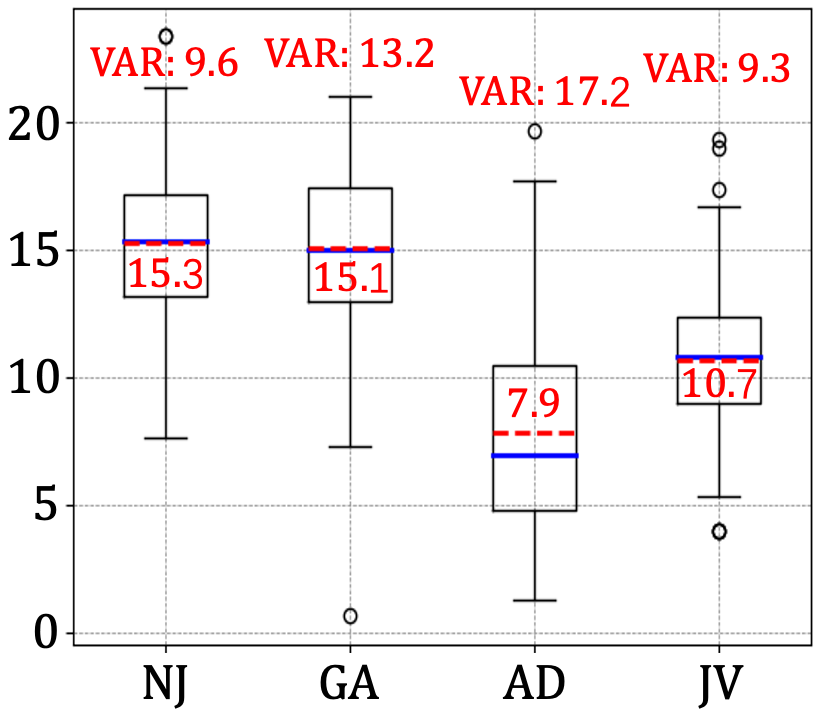}
		\label{fig:score3}
	}
	\vspace{-2mm}
	\caption{{Boxplots of players' scores under $\omega_1 = (1, 0, 0)$, $\omega_2 = (1/2, 1/2, 0)$, and $\omega_3 = (1/3, 1/3, 1/3)$, where average is marked with red dotted lines.}}
	\label{fig:score}
	\vspace{-2mm}
\end{figure}

%均值大，方差小，会让一个元组即在agg rsky中，也具有大的rsky 概率，比如RW NJ
%均值会被支配，方差相对较大，会让一个元组不属于agg rsky，但rsky 概率大，比较NJ和 GA
%某一维度上均值大，方差大，会让一个元组出现在agg rsky中，但是会让它的rsky 概率很小，比较NJ和JV
%AD被JV F支配，但是AD的方差更大，让他rsky 概率更大

%First, players with high rskyline probabilities are also in the aggregate rskyline (really good players)
%Second, some player not in aggregate rskyline but have high rskyline probability, higher than player in the aggregated rskyline (high variance, plot)

We first observe that rskyline probabilities can reflect the difference between two incomparable players in the aggregated dataset under $\calF$.
%players in the aggregated rskyline have different rskyline probabilities, \eg, Nikola Jokic and Jonas Valanciunas.
Theorem~\ref{thm:F-dominace-V} claims that $t \prec_\calF s$ if and only if $\forall \omega \in V,$ $S_{\omega}(t) \le S_{\omega}(s)$.
Here $V = \{\omega_1 = (1, 0, 0), \omega_2 = (1/2, 1/2, 0), \omega_3 = (1/3, 1/3, 1/3)\}$.
See Fig.~\ref{fig:score} for scores of Nikola Jokic (NJ) and Jonas Valanciunas (JV) under weights in $V$.
NH not only has a good average performance so that he is in the aggregated rskyline, but also performs best in some games so that he has a high rskyline probability.
As for JV, his average performance under $\omega_1$ is great, making him belong to the aggregated rskyline.
But his large performance variance under $\omega_1$ and relatively poor performance under $\omega_2$ and $\omega_3$ suggests that many of his records are $\calF$-dominated by other players' records.
This results in his rskyline probability being pretty low.
Therefore, compared to aggregated rskyline players with high rskyline probabilities, who consistently perform well, those with low rskyline probabilities are more likely to have many records being $\calF$-dominated by other players' records, which may be less attractive.
%With the rskyline probability, Jokic is ranked higher than Valanciunas because the former is more likely to perform well in future games.

Second, we find that players not in the aggregated rskyline but have high rskyline probabilities is also appealing.
For example, Giannis Antetokounmpo is $\calF$-dominated by Nikola Jokic in the aggregated dataset, but his rskyline probability is higher than another four aggregated rskyline players.
Compared to NJ, his scores (GA) have slightly lower averages and higher variances.
In other words, he has some records, like Nikola Jokic's, which $\calF$-dominates most of other players' records and he also has some records that are $\calF$-dominated by many of other players' records.
Besides, the large performance variance also explains why Andre Drummond (AD) is $\calF$-dominated by Jonas Valanciunas (JV) but has a higher rskyline probability.
This suggests that looking for players with high rskyline probabilities can find excellent players with slightly lower averages but higher variances in performance.

Finally, a set of players with specified size can be retrieved by performing top-$k$ queries on ARSP, while the size of the aggregated rskyline is uncontrollable.
From these observations, we conclude that ARSP provides a more comprehensive view on uncertain datasets than the aggregated rskyline.

\begin{table}[ht]
	\centering
	\caption{{Top-14 players in skyline probability ranking.}}
	\label{tab:top-14-sky}
	\begin{tabular}{|c|c|c|c|}
		\hline
		\textbf{Player} & $\Pr_{\rm sky}$ & \textbf{Player} & $\Pr_{\rm sky}$ \\ \hline \hline
		Nikola Jokic & 0.557 & LeBron James & 0.308 \\ \hline
		Russell Westbrook & 0.537 & Domantas Sabonis & 0.283 \\ \hline
		Giannis Antetokounmpo & 0.479 & Stephen Curry & 0.266 \\ \hline
		James Harden & 0.447 & Kevin Durant & 0.257 \\ \hline
		Luka Doncic & 0.398 & Nikola Vucevic & 0.236 \\ \hline
		Joel Embiid & 0.339 & Julius Randle & 0.224 \\ \hline
		Trae Young & 0.309 & Damian Lillard & 0.208 \\ \hline
	\end{tabular}
\end{table}

We also compare the distinction between uncertain skyline queries and uncertain rskyline queries.
Similar to Table~\ref{tab:top-14}, Table~\ref{tab:top-14-sky} also reports the top-14 players in skyline probability ranking along with their skyline probabilities.
By analyzing these results, we obtain several interesting observations.
First, the rskyline probability of an uncertain object is typically smaller than its skyline probability because the function set $\mathcal{F}$ improves the dominance ability of each instance.
But excellent players like Nikola Jokic, Russell Westbrook always have both high rankings in skyline probability and rskyline probability.
Second, uncertain rskyline queries can better serve the specific preferences of individual users.
Given different inputs $\calF$ from different users, rskyline probabilities of uncertain objects are variant, however, skyline probabilities of uncertain objects always remain the same.
As stated in~\cite{DBLP:journals/pvldb/CiacciaM17}, a skyline object may be $\calF$-dominated by other objects.
Therefore, an object with high skyline probabilities may have low rskyline probabilities, making it less attractive under $\calF$.
For instance, Trae Young's skyline probability is 0.309 (ranked 7th) but his rskyline probability under $\calF = \{\omega[1] \text{Rebound} + \omega[2] \text{Assist} + \omega[3] \text{Point} \mid \omega[1] \ge \omega[2] \ge \omega[3]\}$ is only 0.029 (ranked 31st).

\subsection{Experimental Results under Linear Constraints.}\label{sec:exp-linear}

%3X5 synthetic

\begin{figure*}[t]
	\centering
	\includegraphics[height=0.12in]{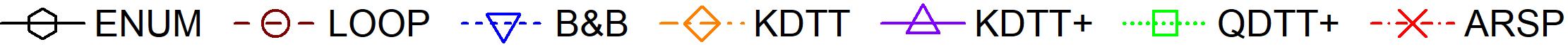}
	\\
	\subfigure[{IND, vary $m$}]{
		\includegraphics[width=0.186\linewidth]{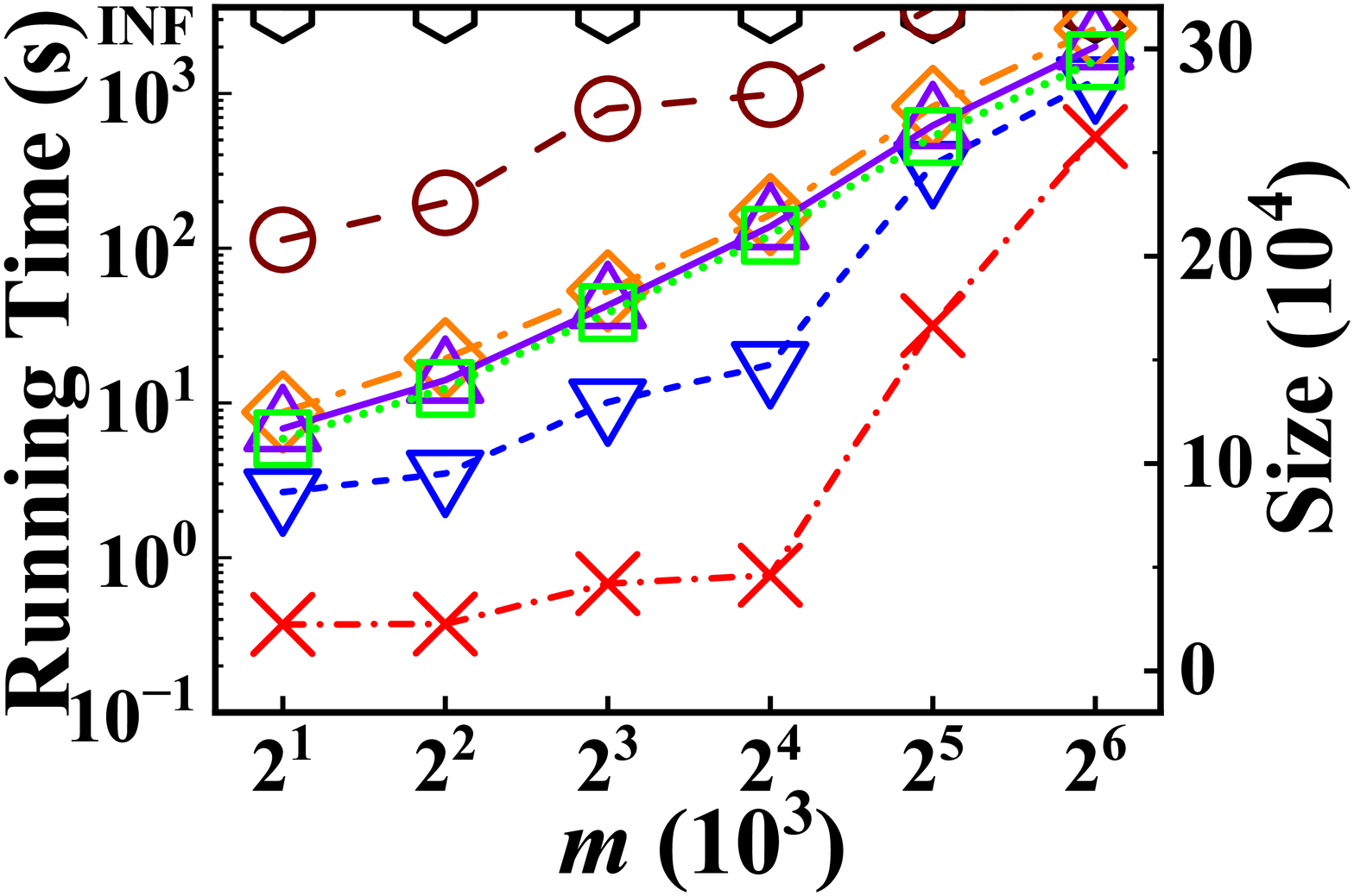}
		\label{fig:ind_m}
	}\hfill
	\subfigure[{ANTI, vary $m$}]{
		\includegraphics[width=0.186\linewidth]{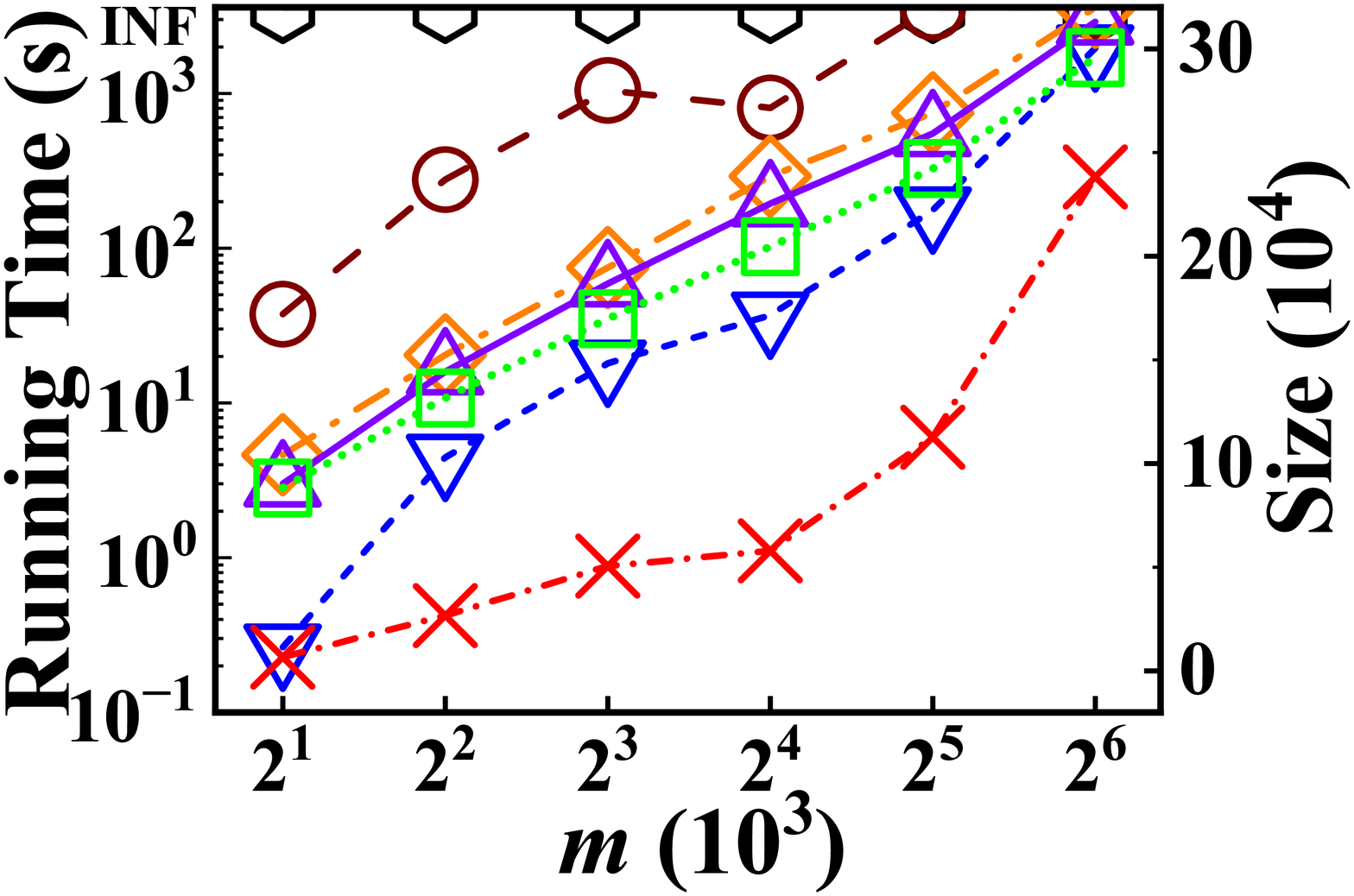}
		\label{fig:anti_m}
	}\hfill
	\subfigure[{CORR, vary $m$}]{
		\includegraphics[width=0.186\linewidth]{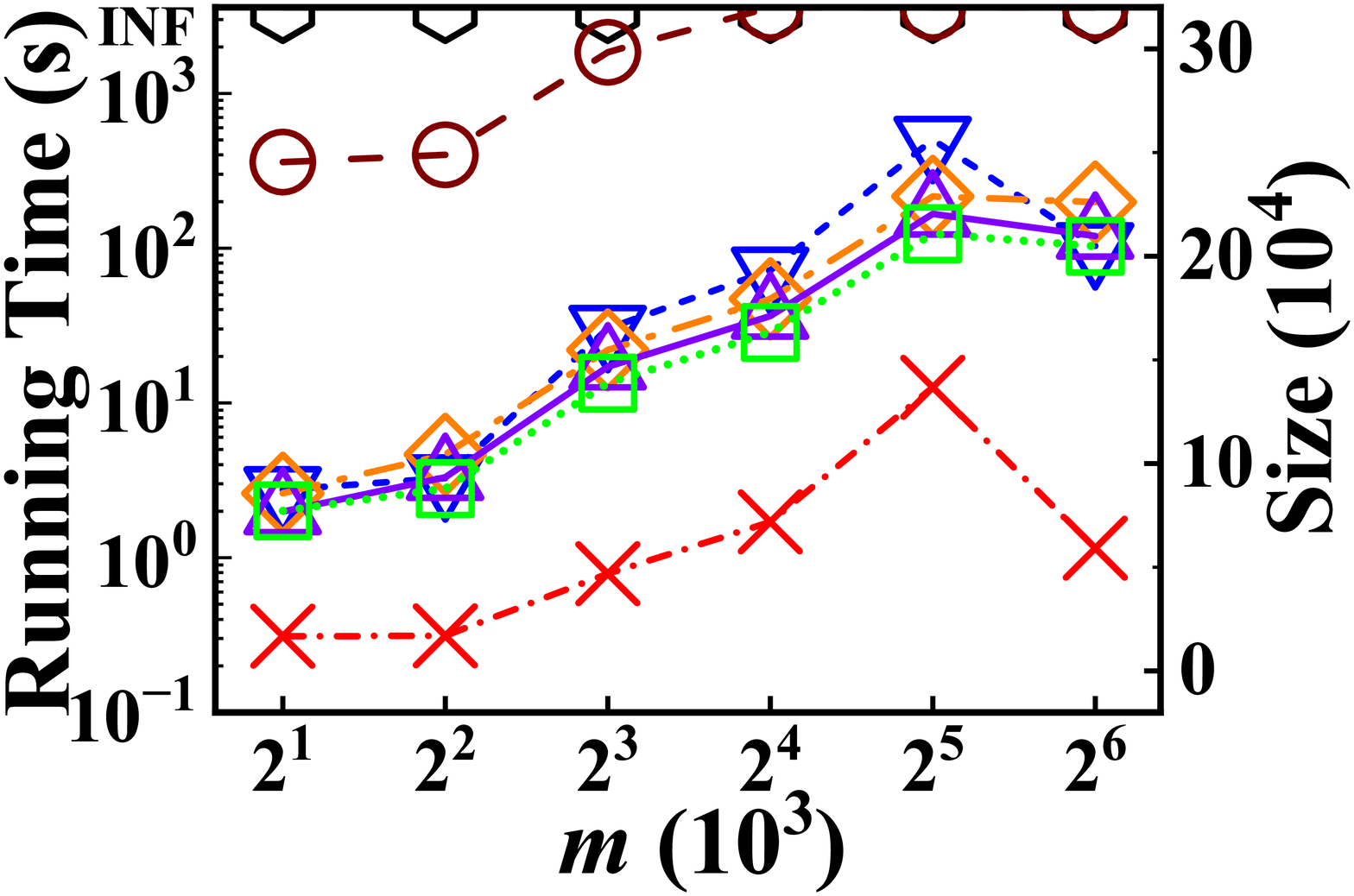}
		\label{fig:corr_m}
	}\hfill
	\subfigure[{IND, vary $cnt$}]{
		\includegraphics[width=0.186\linewidth]{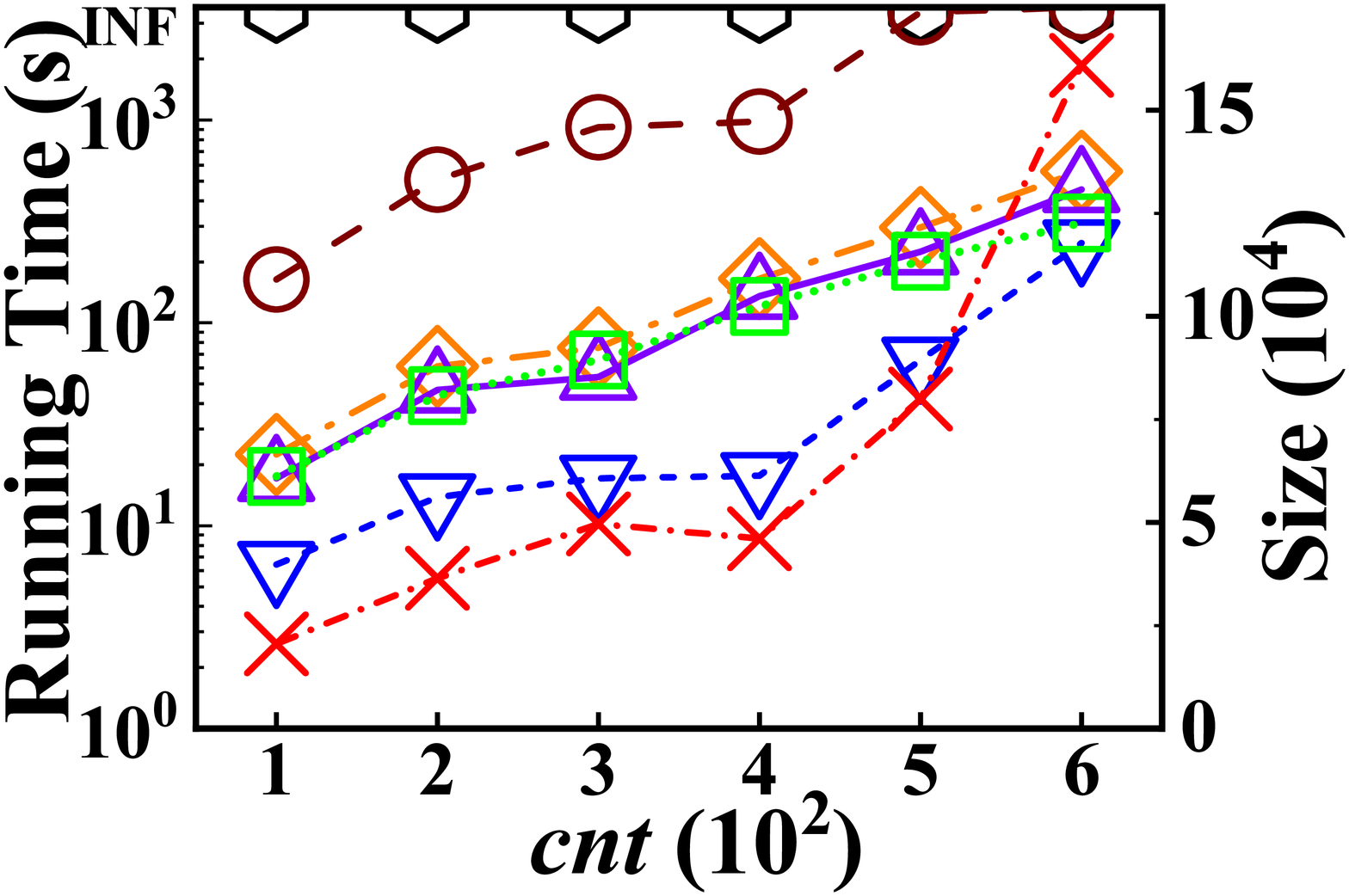}
		\label{fig:ind_cnt}
	}\hfill
	\subfigure[{ANTI, vary $cnt$}]{
		\includegraphics[width=0.186\linewidth]{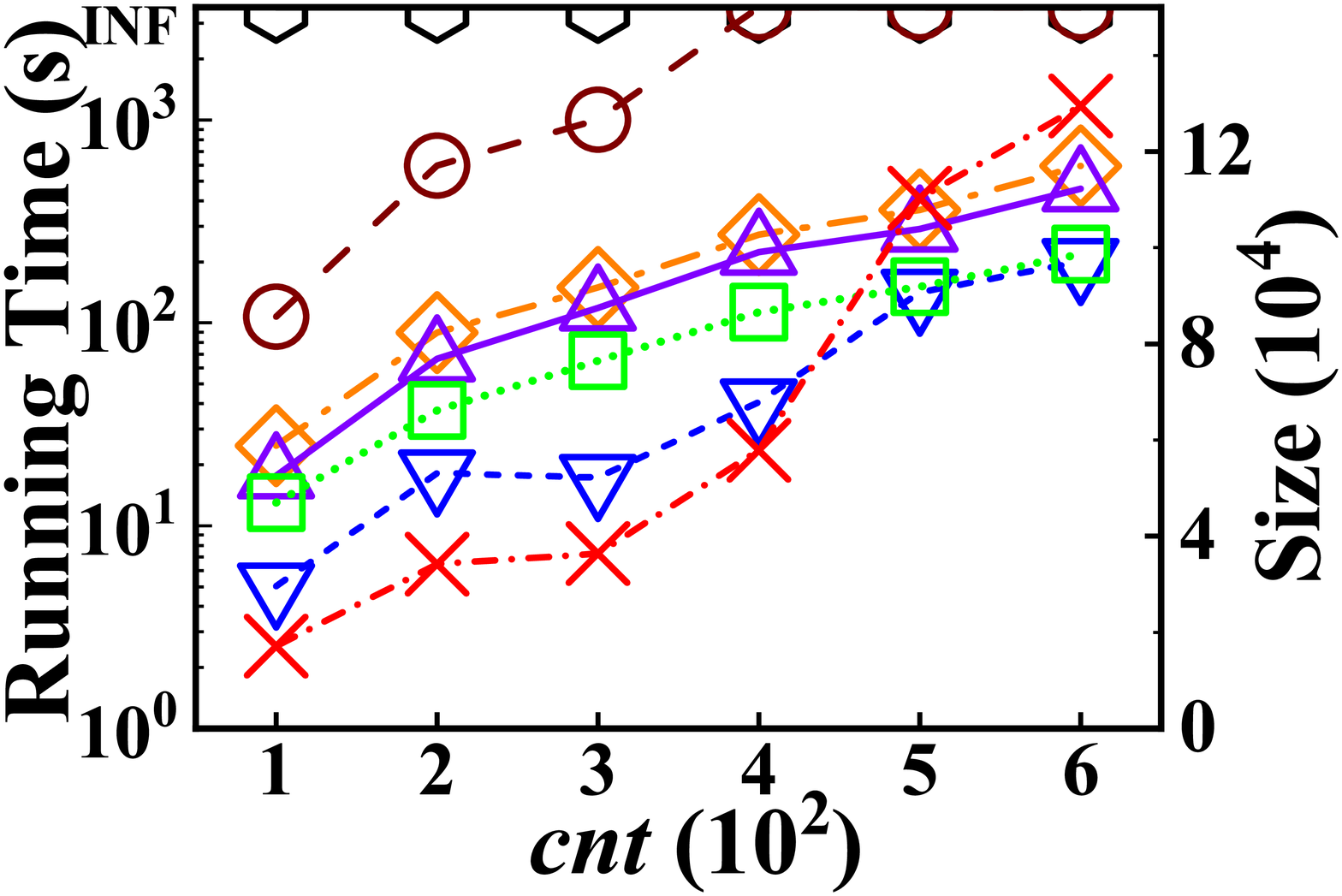}
		\label{fig:anti_cnt}
	}
	\subfigure[{CORR, vary $cnt$}]{
		\includegraphics[width=0.186\linewidth]{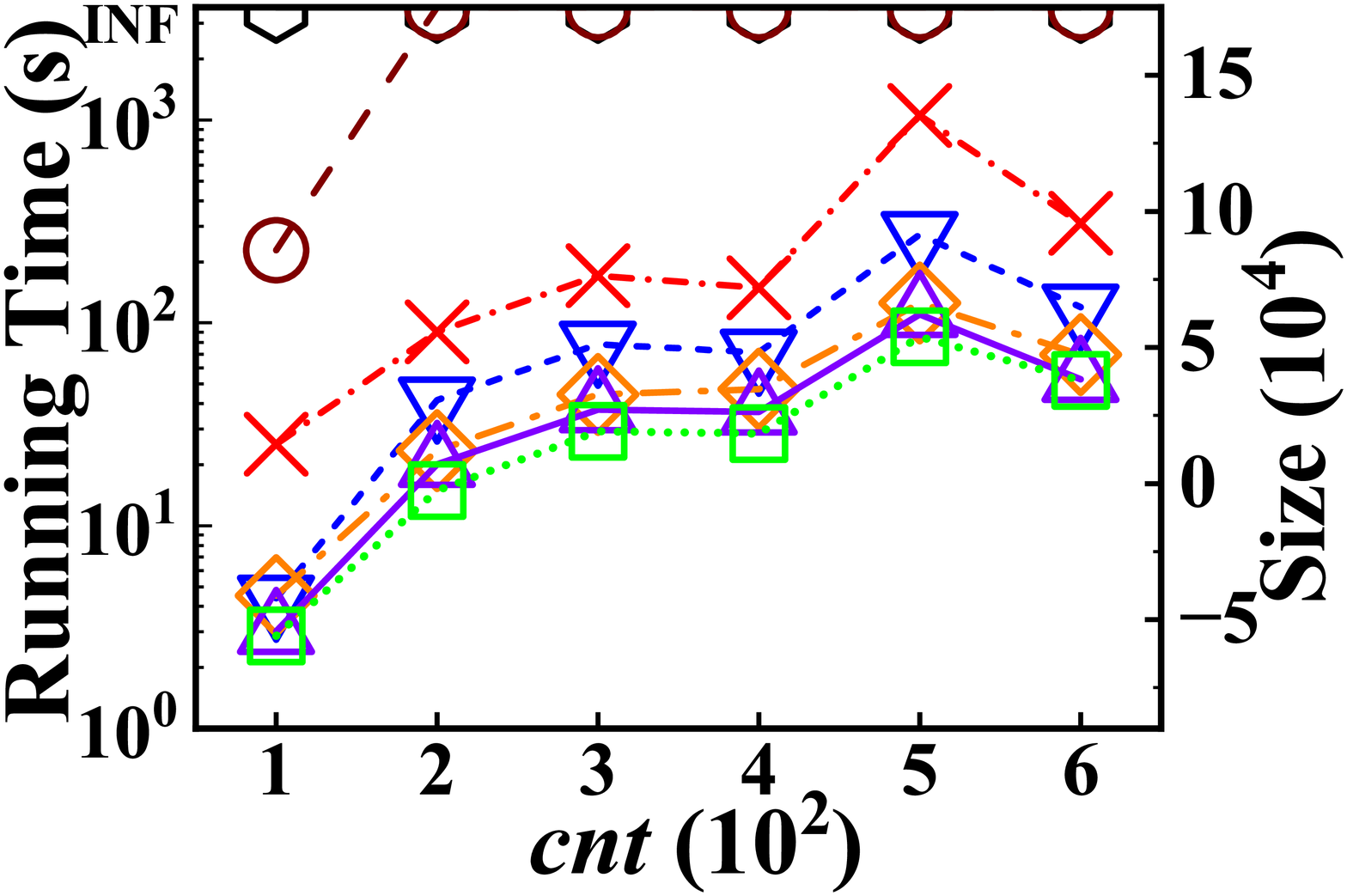}
		\label{fig:corr_cnt}
	}\hfill
	\subfigure[{IND, vary $d$}]{
		\includegraphics[width=0.186\linewidth]{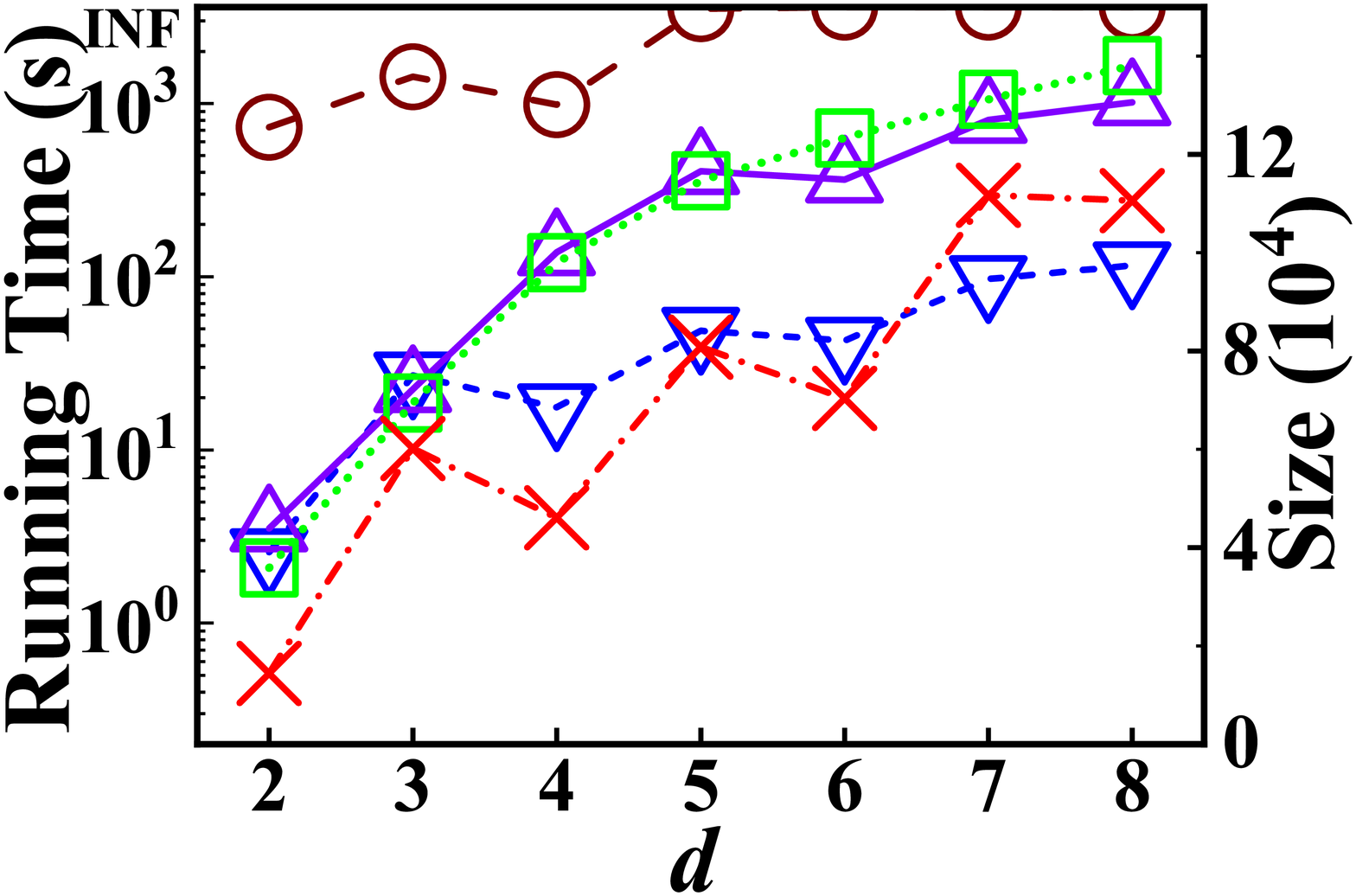}
		\label{fig:ind_d}
	}\hfill
	\subfigure[{ANTI, vary $d$}]{
		\includegraphics[width=0.186\linewidth]{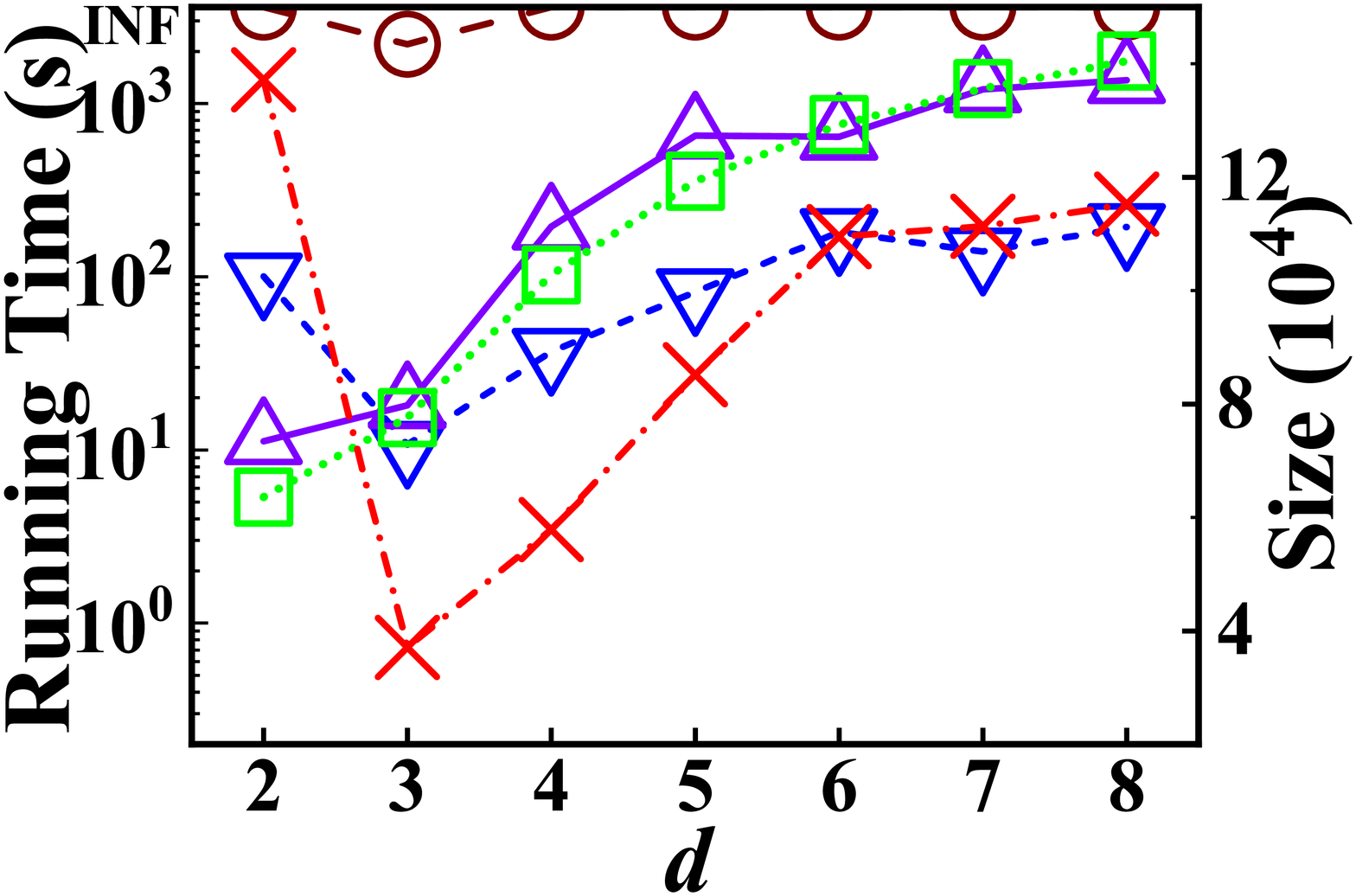}
		\label{fig:anti_d}
	}\hfill
	\subfigure[{CORR, vary $d$}]{
		\includegraphics[width=0.186\linewidth]{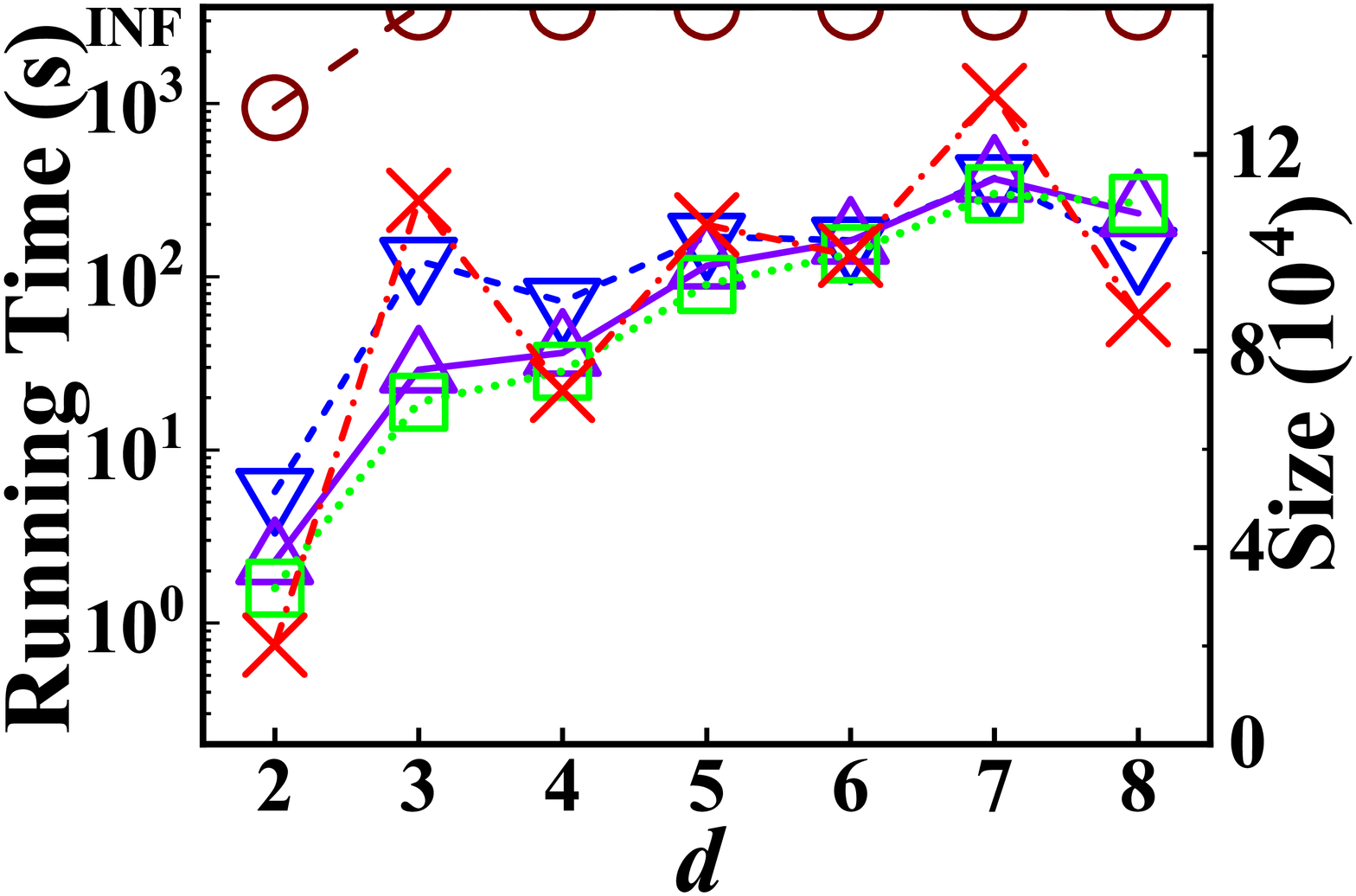}
		\label{fig:corr_d}
	}\hfill
	\subfigure[{IND, vary $l$}]{
		\includegraphics[width=0.186\linewidth]{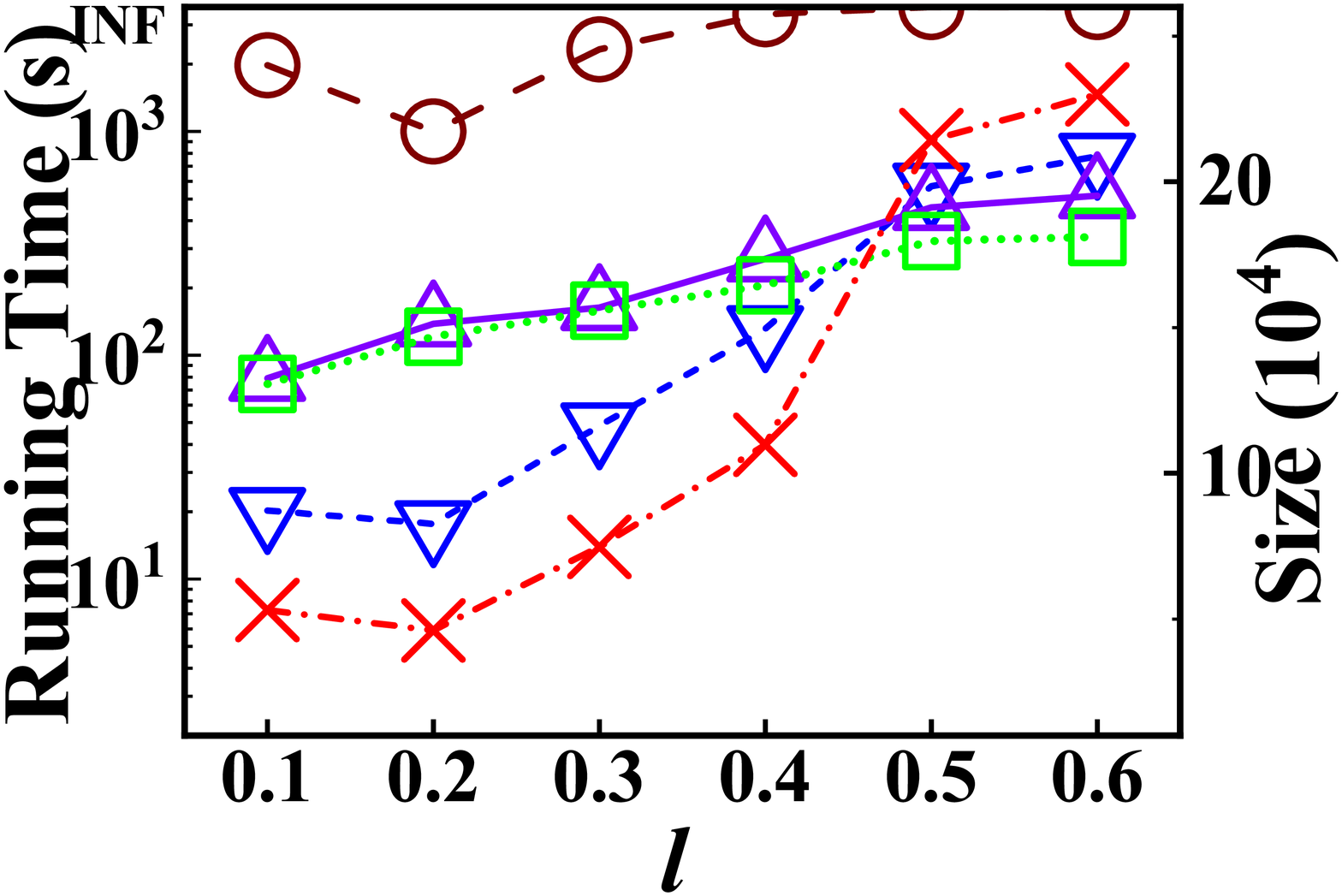}
		\label{fig:ind_l}
	}
	\subfigure[{ANTI, vary $l$}]{
		\includegraphics[width=0.186\linewidth]{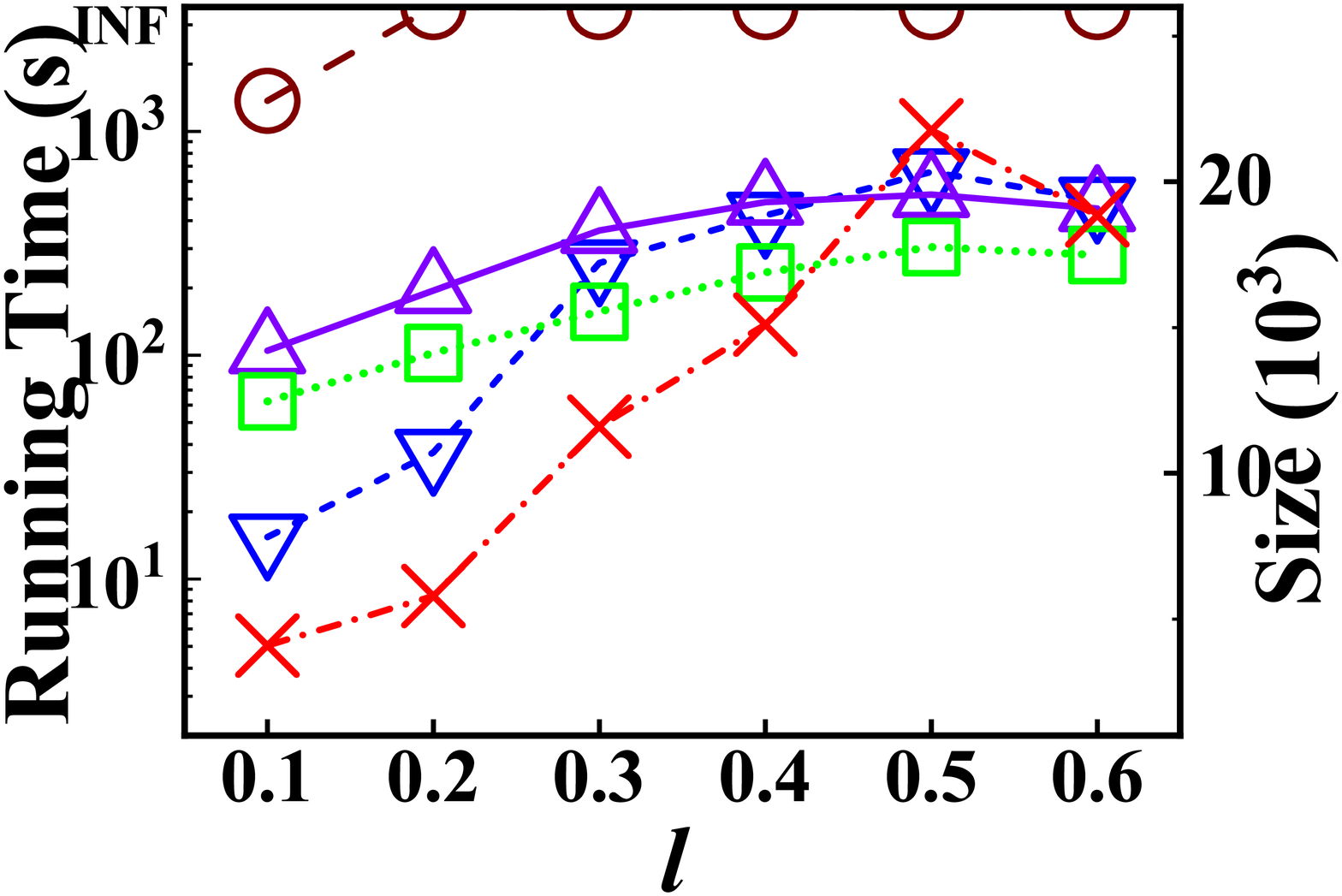}
		\label{fig:anti_l}
	}\hfill
	\subfigure[{CORR, vary $l$}]{
		\includegraphics[width=0.186\linewidth]{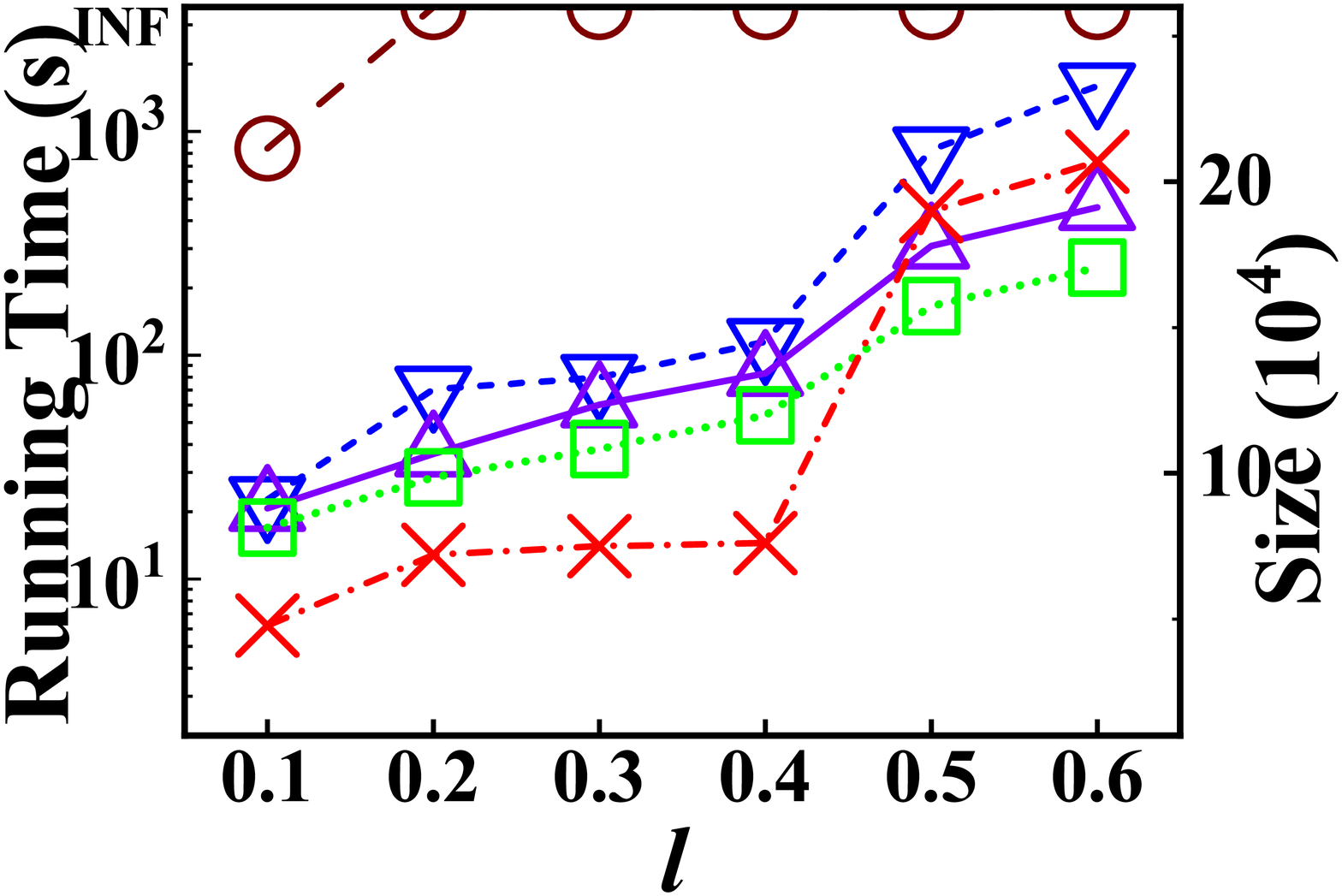}
		\label{fig:corr_l}
	}\hfill
	\subfigure[{IND, vary $\phi$}]{
		\includegraphics[width=0.186\linewidth]{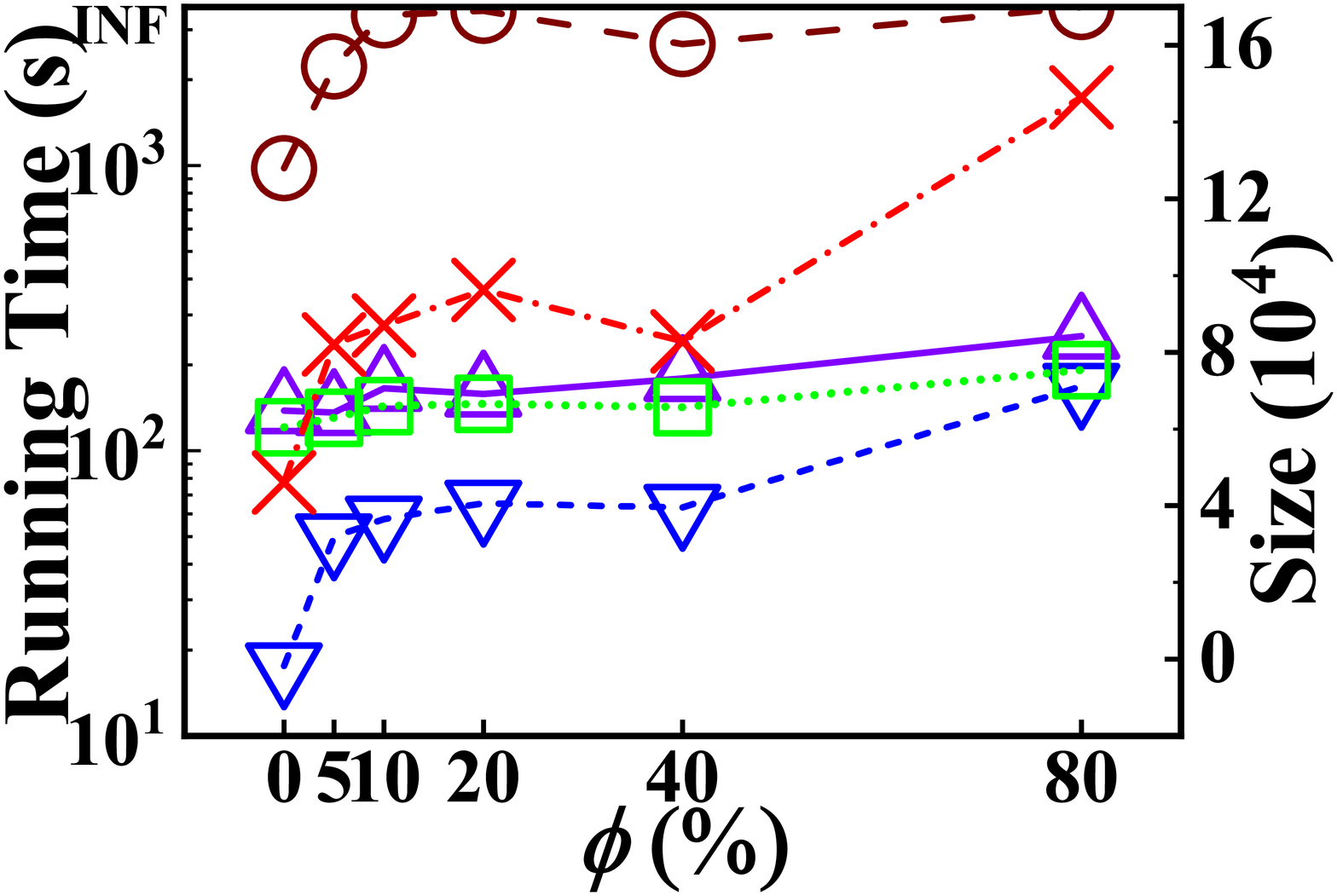}
		\label{fig:ind_phi}
	}\hfill
	\subfigure[{ANTI, vary $\phi$}]{
		\includegraphics[width=0.186\linewidth]{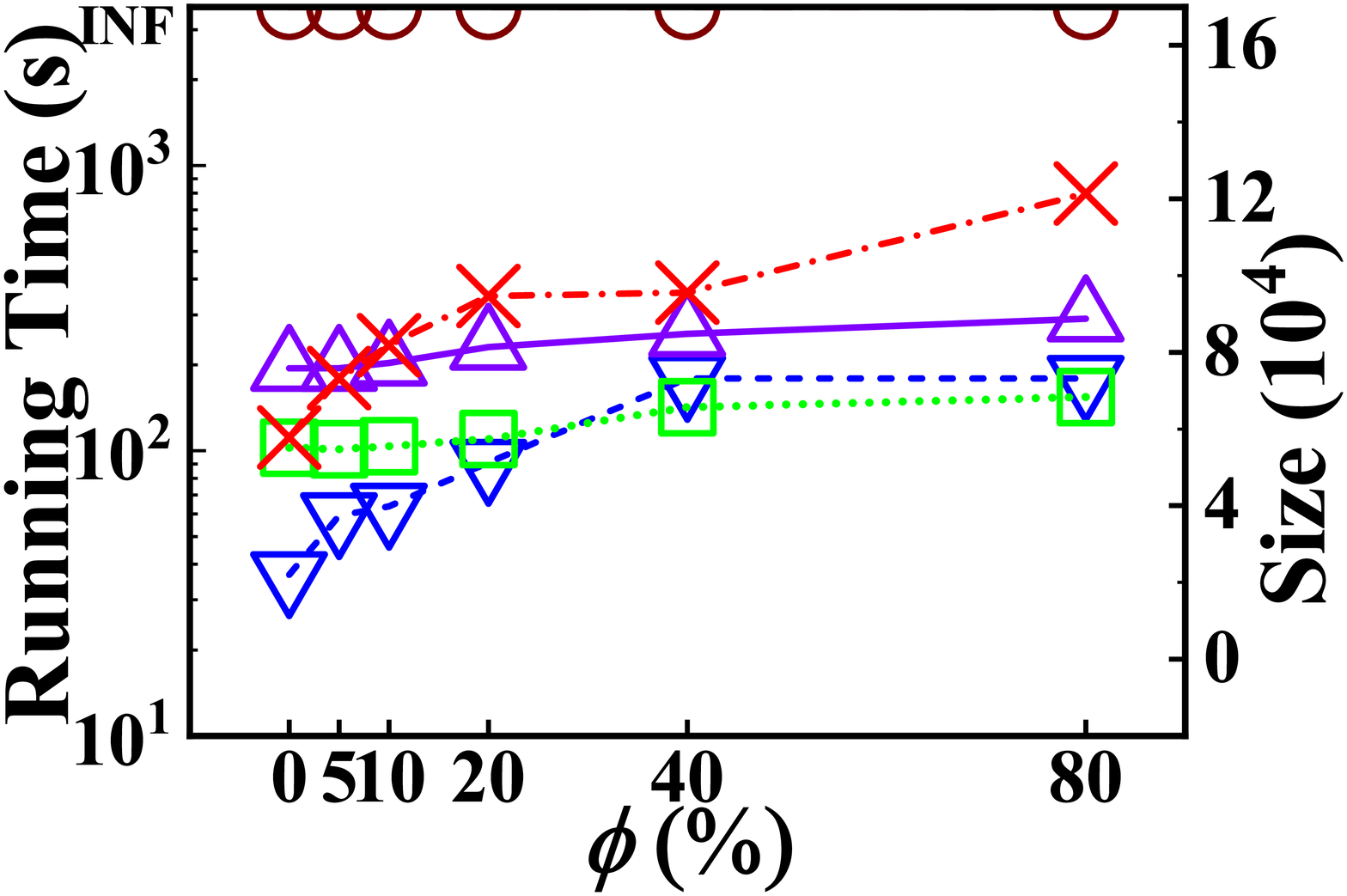}
		\label{fig:anti_phi}
	}\hfill
	\subfigure[{CORR, vary $\phi$}]{
		\includegraphics[width=0.186\linewidth]{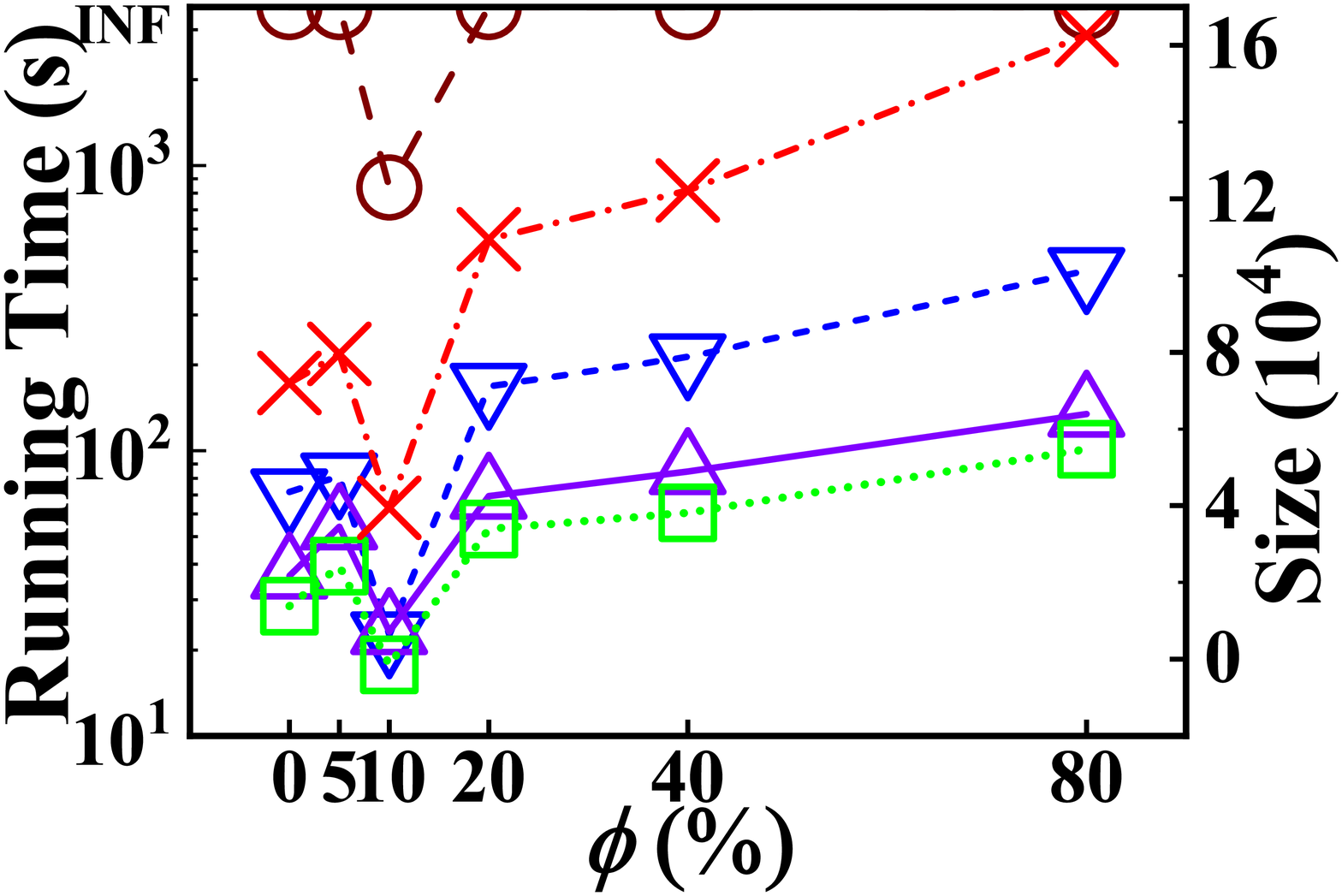}
		\label{fig:corr_phi}
	}
	\includegraphics[height=0.12in]{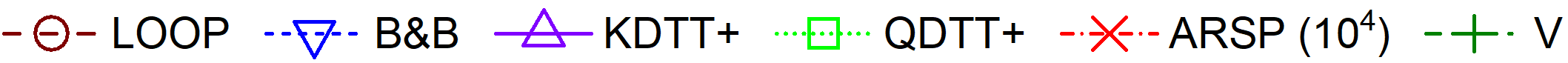}
	\\
	\subfigure[{IND ($d = 6$), vary $c$}]{
		\includegraphics[width=0.18\textwidth]{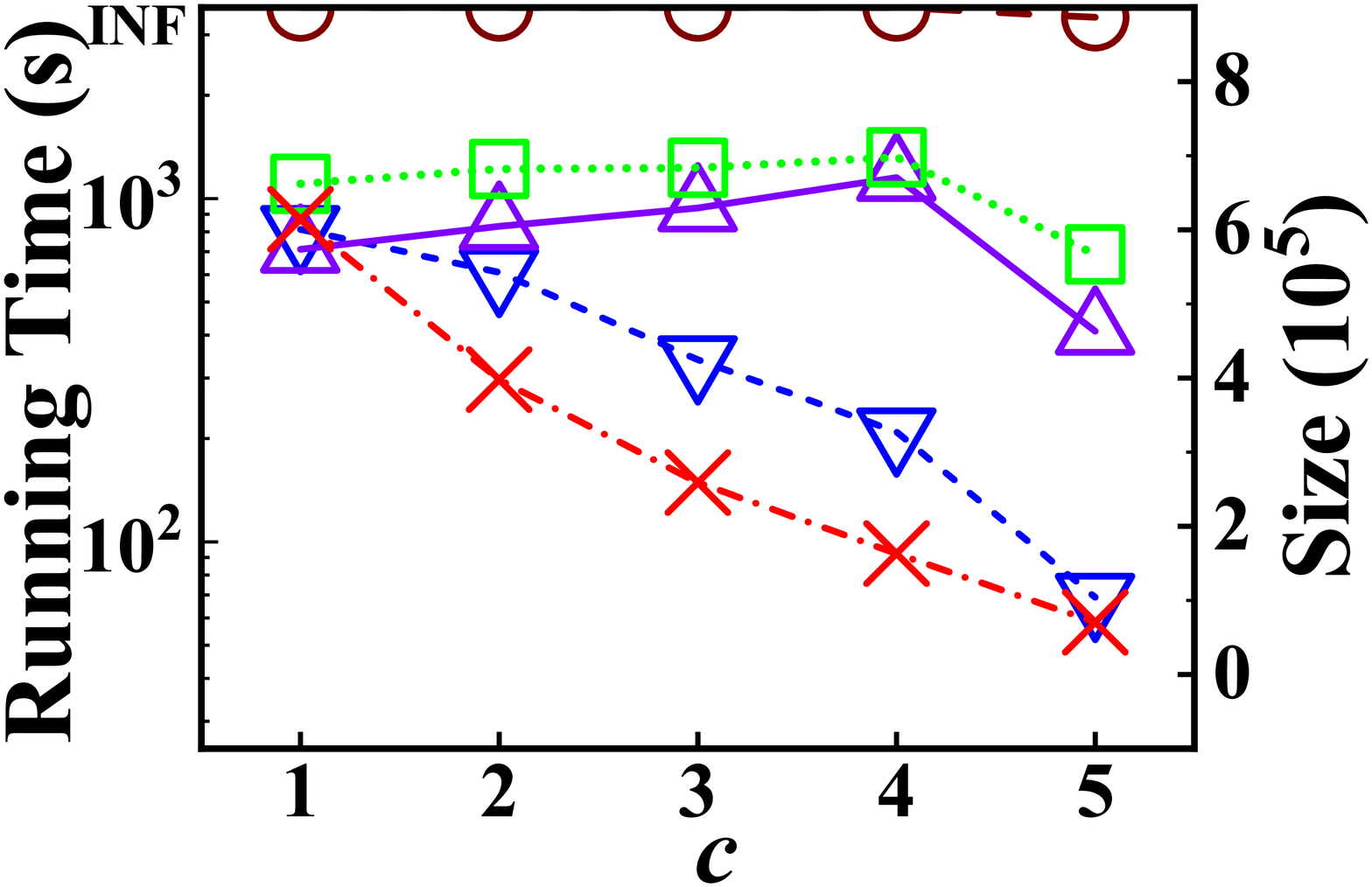}
		\label{fig:ind_wr_c}
	}\hfill
	\subfigure[{ANTI ($d = 6$), vary $c$}]{
		\includegraphics[width=0.18\textwidth]{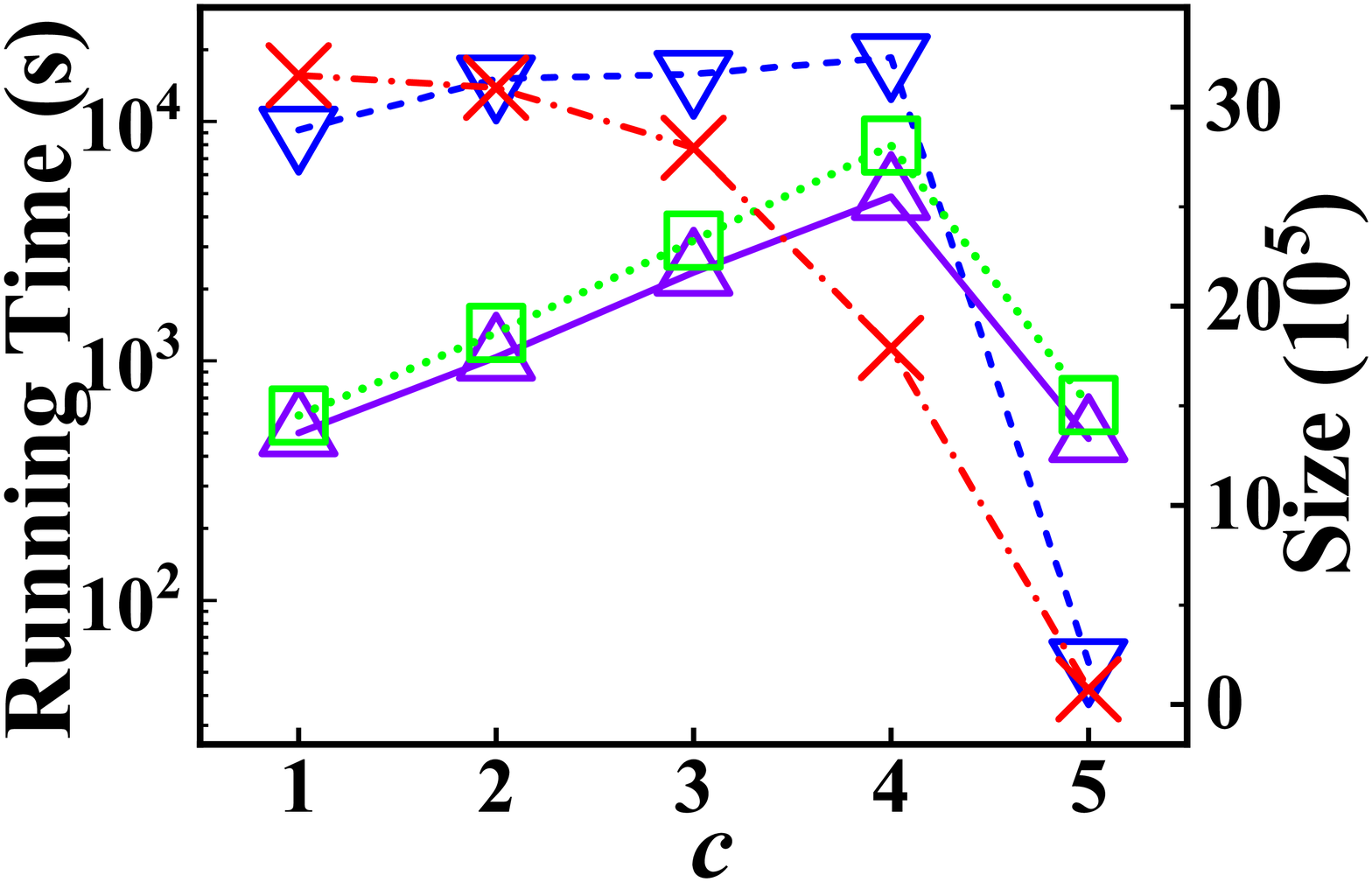}
		\label{fig:anti_wr_c}
	}\hfill
	\subfigure[{IND, IM, vary $m$}]{
		\includegraphics[width=0.18\textwidth]{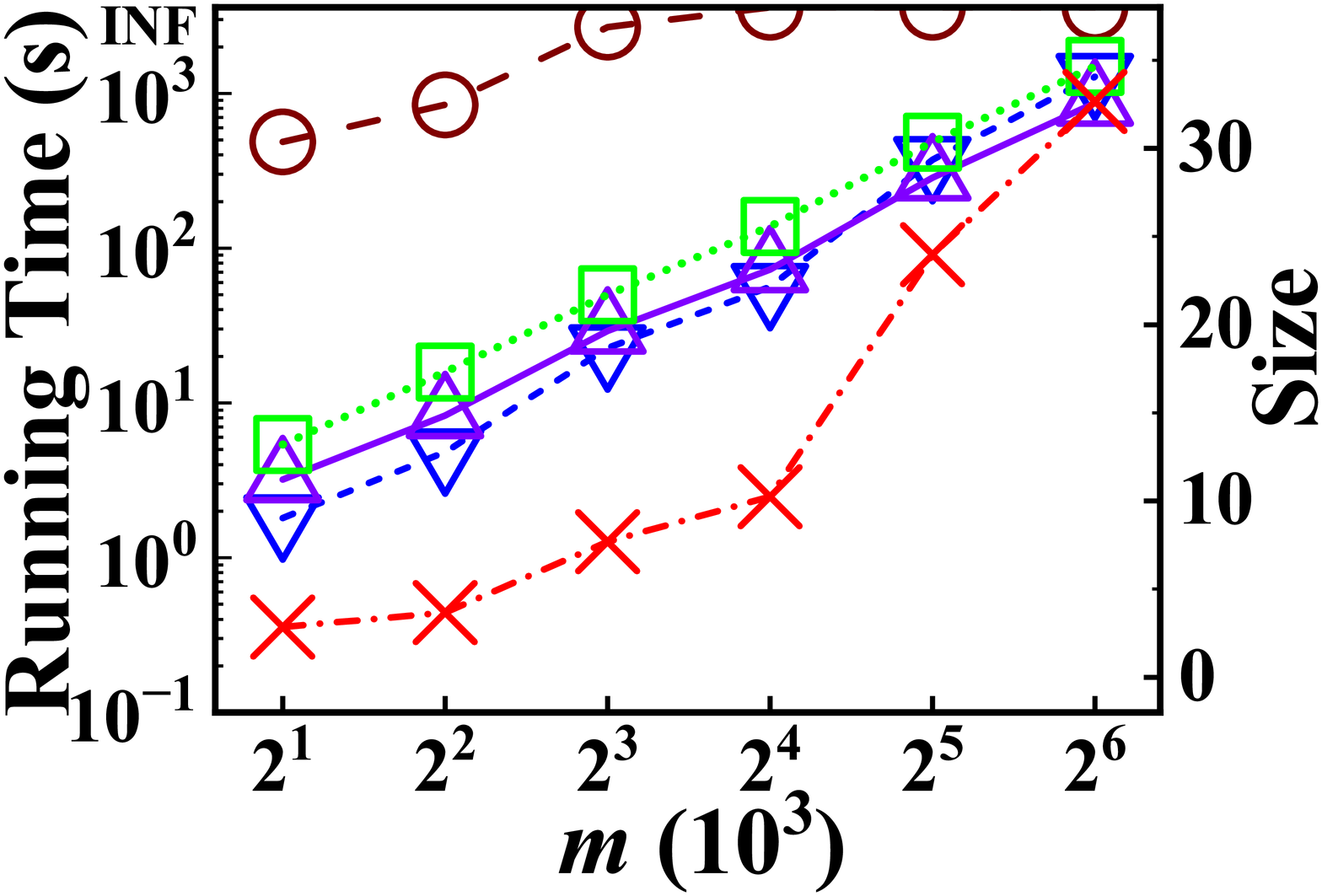}
		\label{fig:ind_im_m}
	}\hfill
	\subfigure[{IND, IM, vary $d$}]{
		\includegraphics[width=0.18\textwidth]{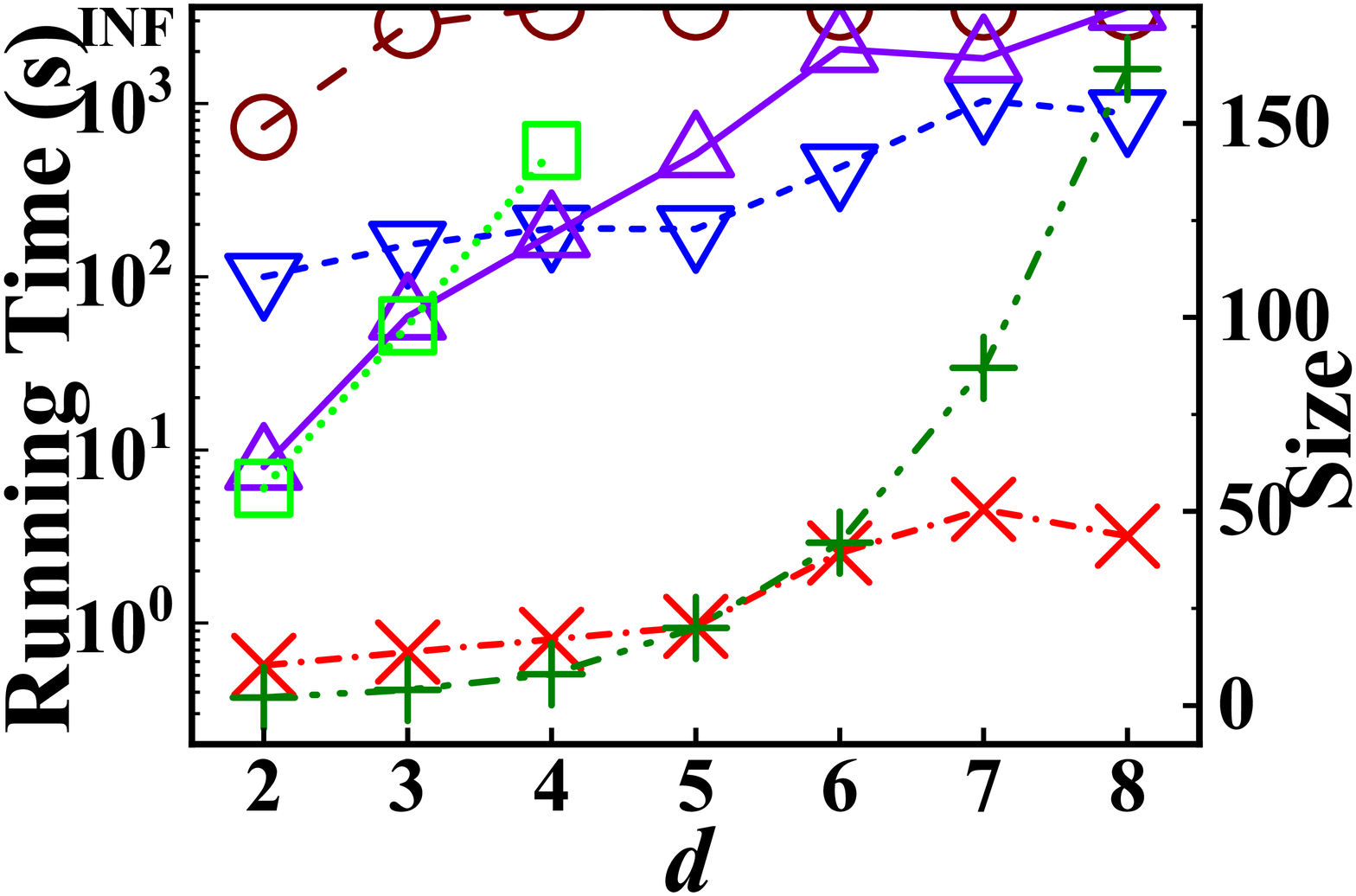}
		\label{fig:ind_im_d}
	}\hfill
	\subfigure[{IND, IM, vary $c$}]{
		\includegraphics[width=0.18\textwidth]{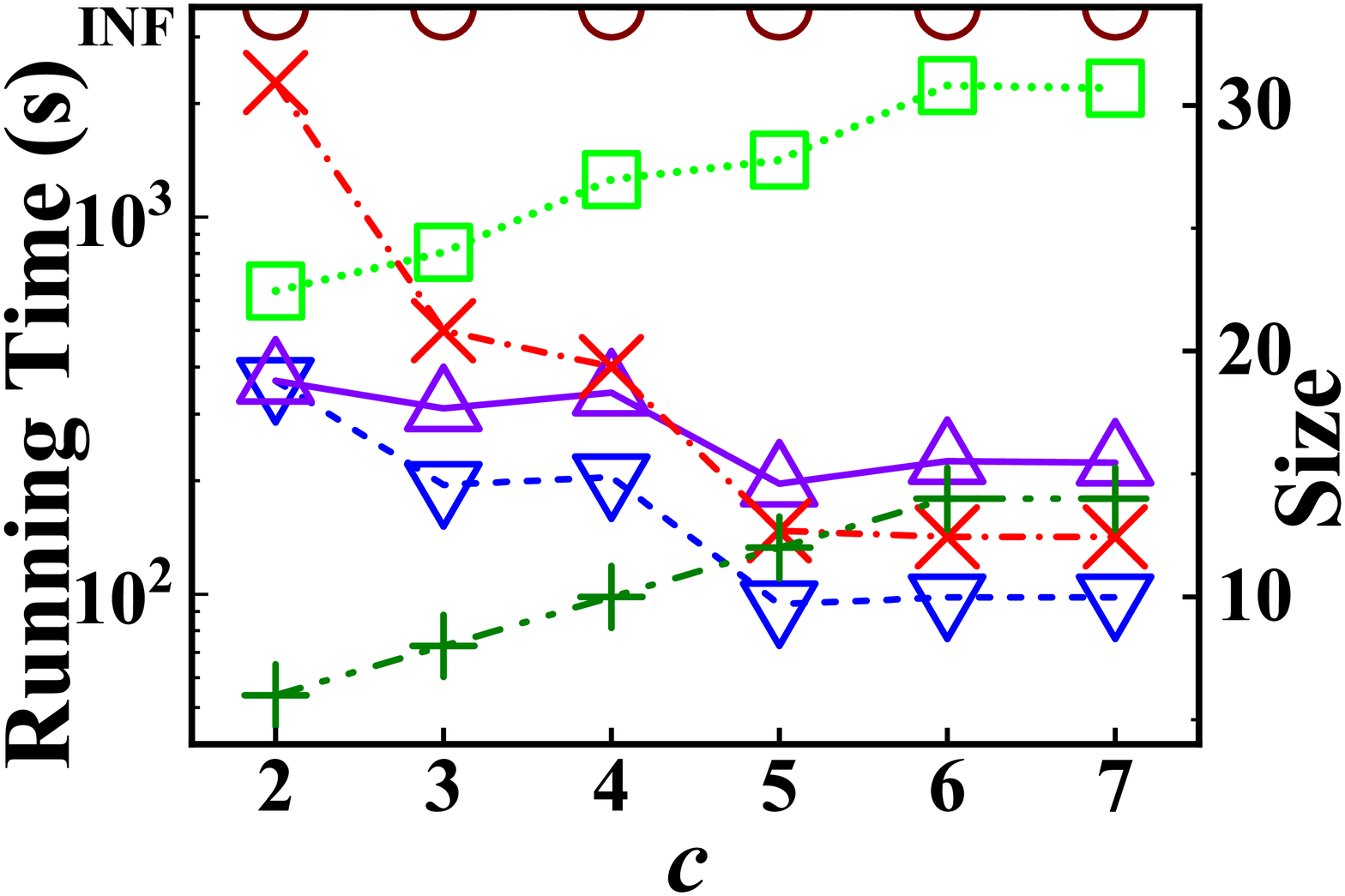}
		\label{fig:ind_im_c}
	}
	\vspace{-2mm}
	\caption{{Running time of different algorithms and the size of ARSP on synthetic datasets.}}
	\label{fig:synthetic-data}	\vspace{-2mm}
\end{figure*}

Fig.~\ref{fig:synthetic-data} and~\ref{fig:real-data} show the running time of different algorithms and the size of ARSP on real and synthetic datasets.
The size of ARSP is the number of instances with none-zero rskyline probabilities.
Following~\cite{DBLP:conf/vldb/PeiJLY07, DBLP:journals/tkde/KimIP12}, the default values for object cardinality $m$, instance count $cnt$, data dimensionality $d$, region length $l$, and percentage $\phi$ of objects with $\sum_{t \in T} p(t) < 1$ of synthetic datasets are set as $m = 16$K, $cnt = 400$,  $d = 4$, $l = 0.2$, and $\phi = 0$.
Unless otherwise stated, WR is used to generate $c = d - 1$ input linear constraints for $\calF$.
All datasets are index by R-trees in main memory.
Since the construction time of the R-tree is a one-time cost for all subsequent queries, it is not included in the query time.
And the query time limit (INF) is set as 3,600 seconds.

Fig.~\ref{fig:synthetic-data} (a)-(c) present the results on synthetic datasets with $m$ varying from $2$K to $64$K.
Based on the generation process, the number of instances $n$ increases as $m$ grows.
Thus, the running time of all algorithms and the size of ARSP increase.
Due to the exponential time complexity, ENUM never finishes within the limited time.
All proposed algorithms outperform LOOP by around an order of magnitude because LOOP always performs a large number of $\calF$-dominance tests and does not include any pruning strategy.
B\&B runs fastest on IND and ANTI with the help of the incremental mapping and pruning strategies.
As $m$ grows, the gap narrows because the more objects, the more aggregated R-trees are queried per instance.
KDTT+ and QDTT+ are more effective on CORR because pruning is triggered earlier by objects near the origin during space partitioning, \eg, when $m = 2$K, QDTT+ prunes 13 child nodes of the root node on CORR, compared with 9 on IND and 5 on ANTI.
%For example, in QDTT+ 13 children of the root is pruned on CORR, while 9 on IND, 5 on ANTI.
%13 corr, 9 inde, 5 anti
%The pruning strategies in KDTT+ and QDTT+ rely on space partition.
%It is more effective on CORR because most subtrees can be pruned by objects near the origin.
Although with similar strategies, {QDTT+} performs better than {KDTT+}.
The reason is that space is recursively divided into $2^d$ regions in QDTT+, which results in a smaller MBR and thus a greater possibility of being pruned.
%Algorithms based on traversing space-partition trees work well on CORR.
%This is because the preorder of trees make them reach to a cell containing ...
%Since correlated distribution is the least challenging for dominance-based queries~\cite{DBLP:journals/tkde/KimIP12, DBLP:journals/pvldb/CiacciaM17}, the performance of all algorithms improve on correlated datasets.
%KDTT+ and QDTT+ run faster than B\&B on correlated datasets.
%This is because during computation, pruning of the subtree happen when these algorithms reach a node that completely contains the minimum-bounding rectangle of an object.
%For example, when $m = 2$K, QDTT+ prunes 15 subtrees of the root node. 
Results also demonstrate our optimization techniques significantly improve the experimental performance of {KDTT}.
As shown in	Fig.~\ref{fig:synthetic-data} (d)-(f), the relative performance of all algorithms remains basically unchanged with respect to $cnt$.
And the size of ARSP also increases as $cnt$ grows since the larger $cnt$, the more instances in $I$ and the less likelihood of an instance being $\calF$-dominated by all instance of an object.

\begin{figure*}[t]
	\centering
	\includegraphics[height=0.12in]{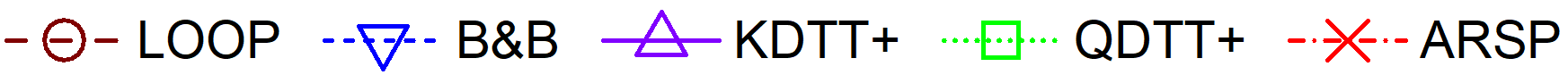}
	\\
	\subfigure[{IIP, vary $m$}]{
		\includegraphics[width=0.186\linewidth]{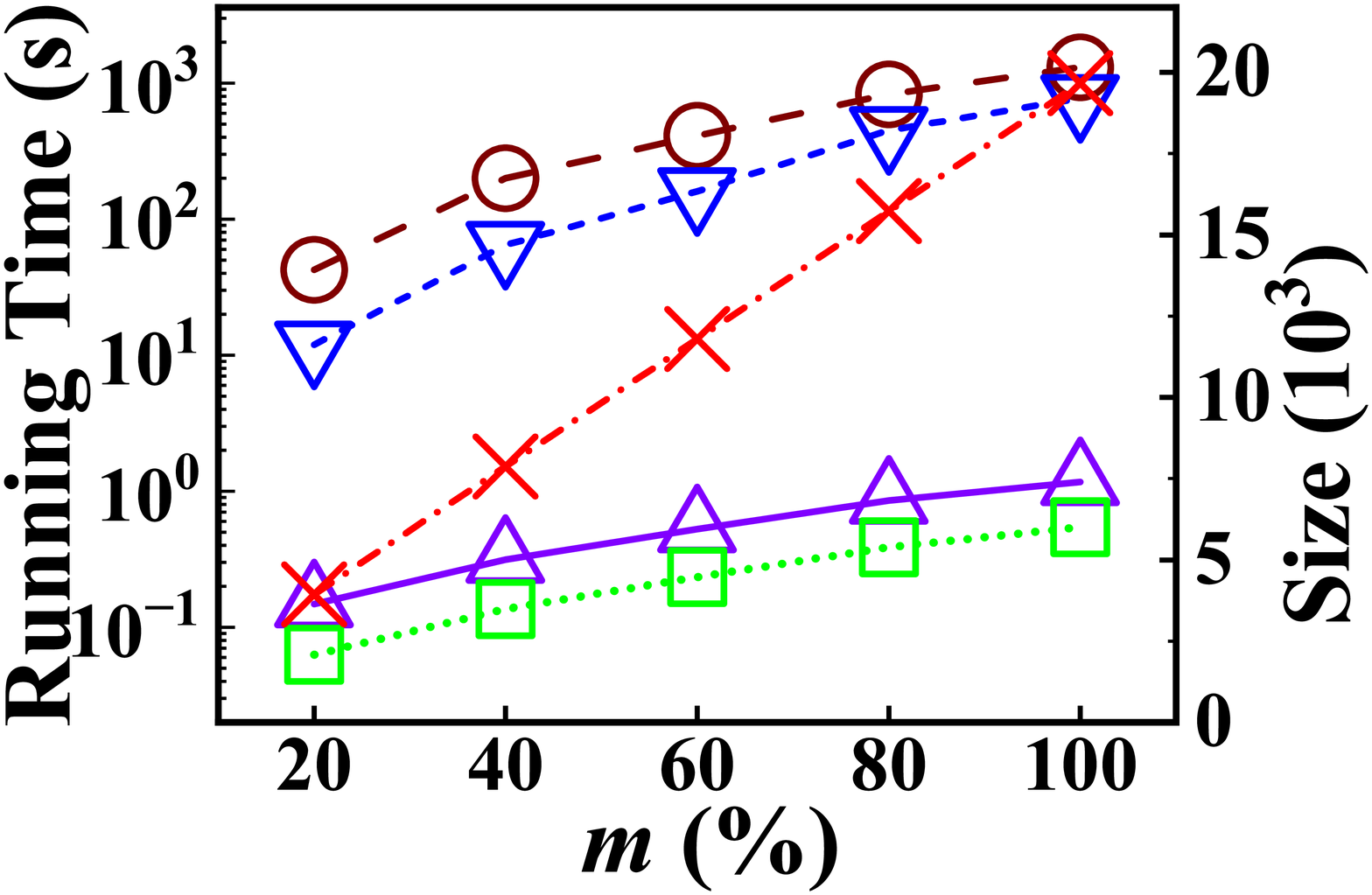}
		\label{fig:iip_m}
	}\hfill
	\subfigure[{CAR, vary $m$}]{
		\includegraphics[width=0.186\linewidth]{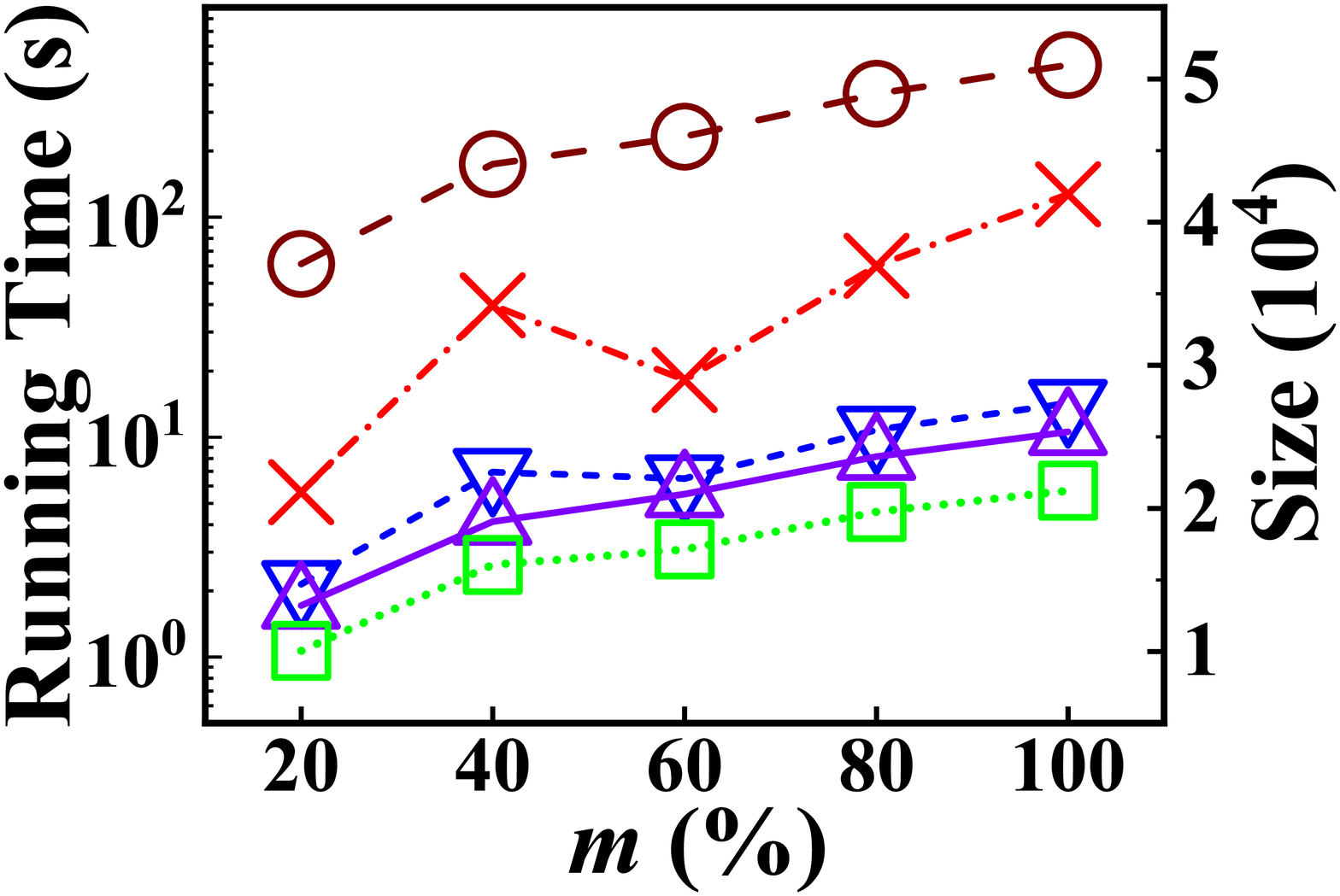}
		\label{fig:car_m}
	}\hfill
	\subfigure[{NBA, vary $m$}]{
		\includegraphics[width=0.186\linewidth]{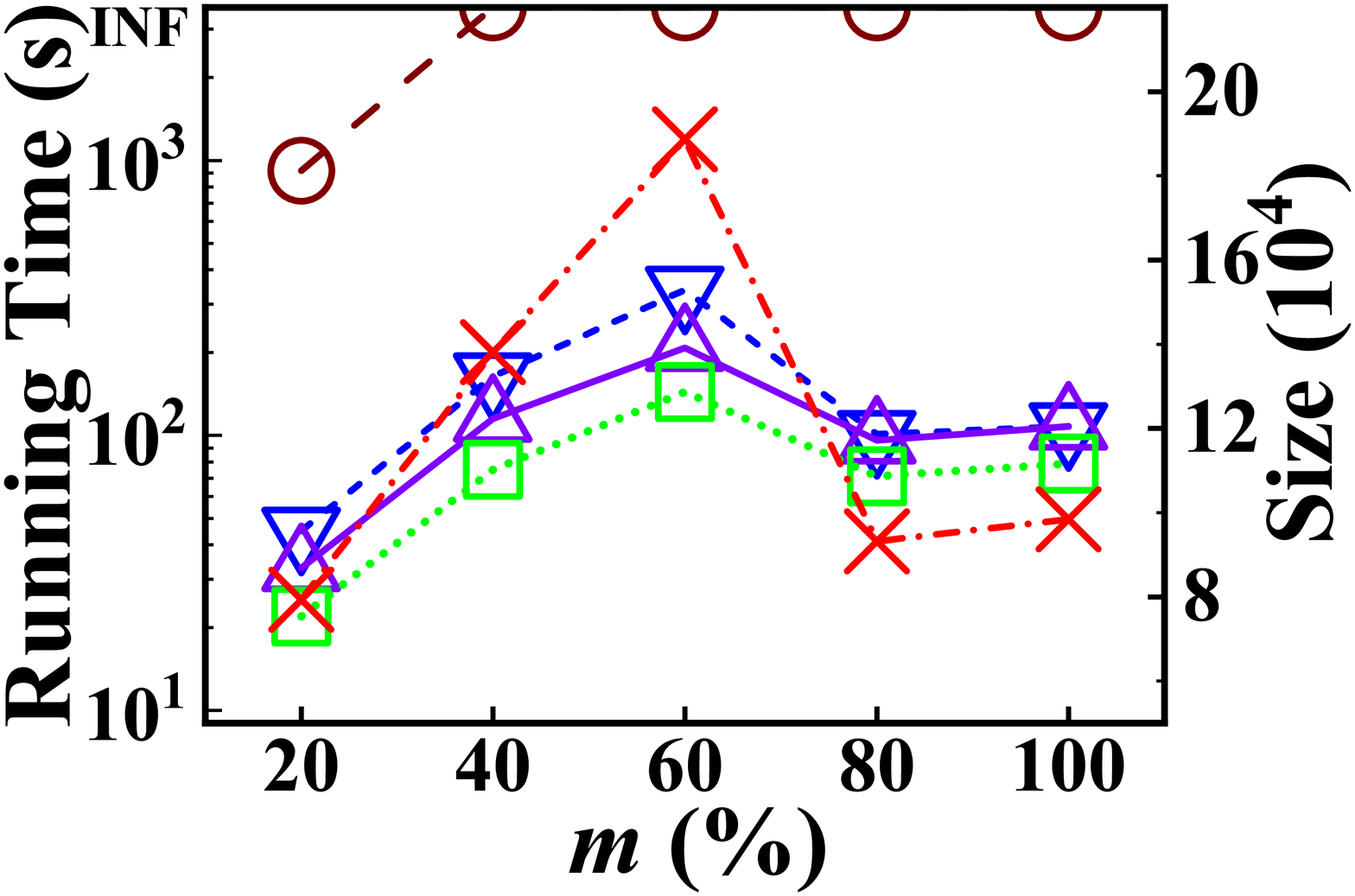}
		\label{fig:nba_m}
	}\hfill
	\subfigure[{NBA, vary $d$}]{
		\includegraphics[width=0.186\linewidth]{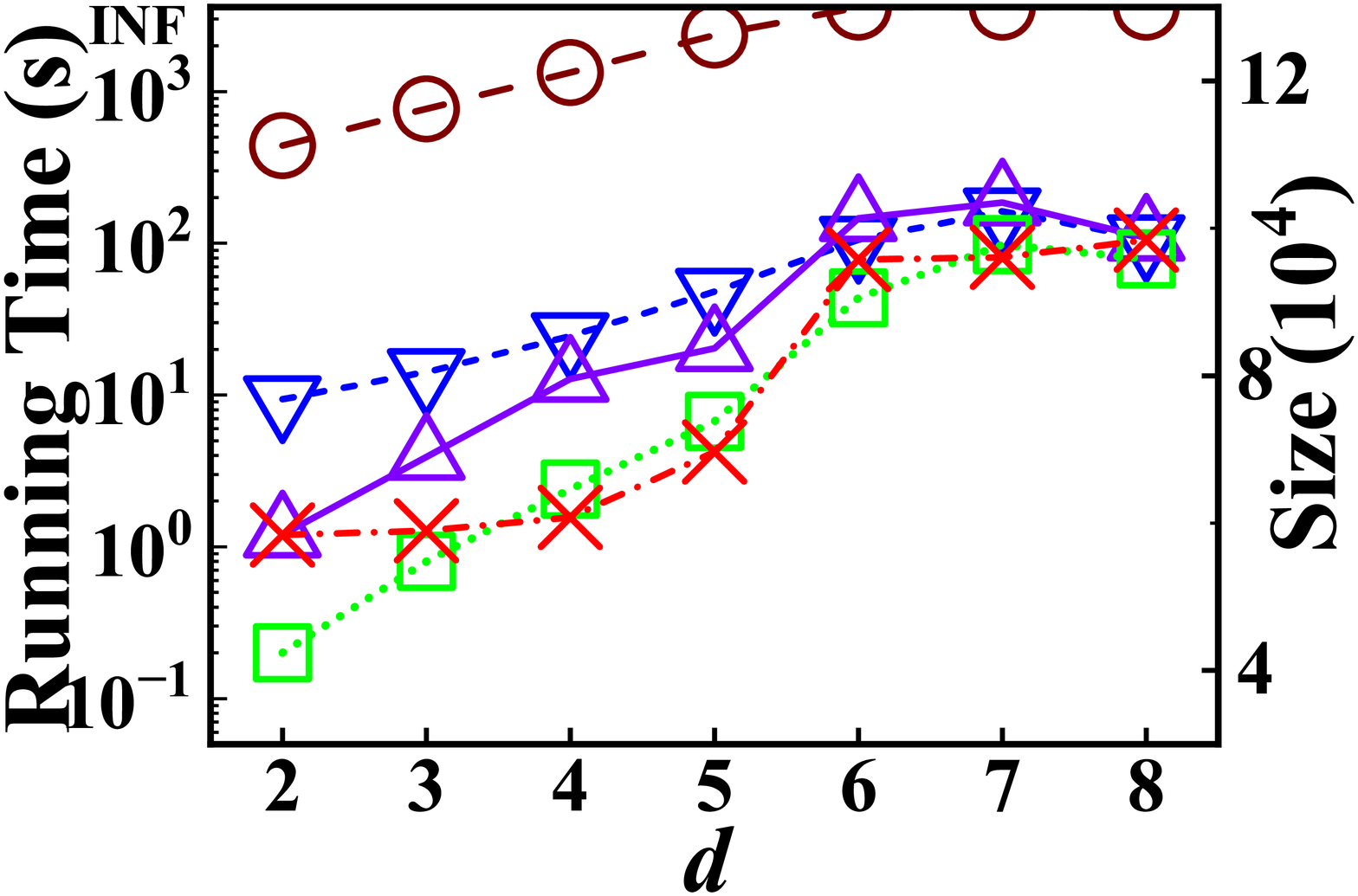}
		\label{fig:nba_d}
	}\hfill
	\subfigure[{NBA, vary $c$}]{
		\includegraphics[width=0.186\linewidth]{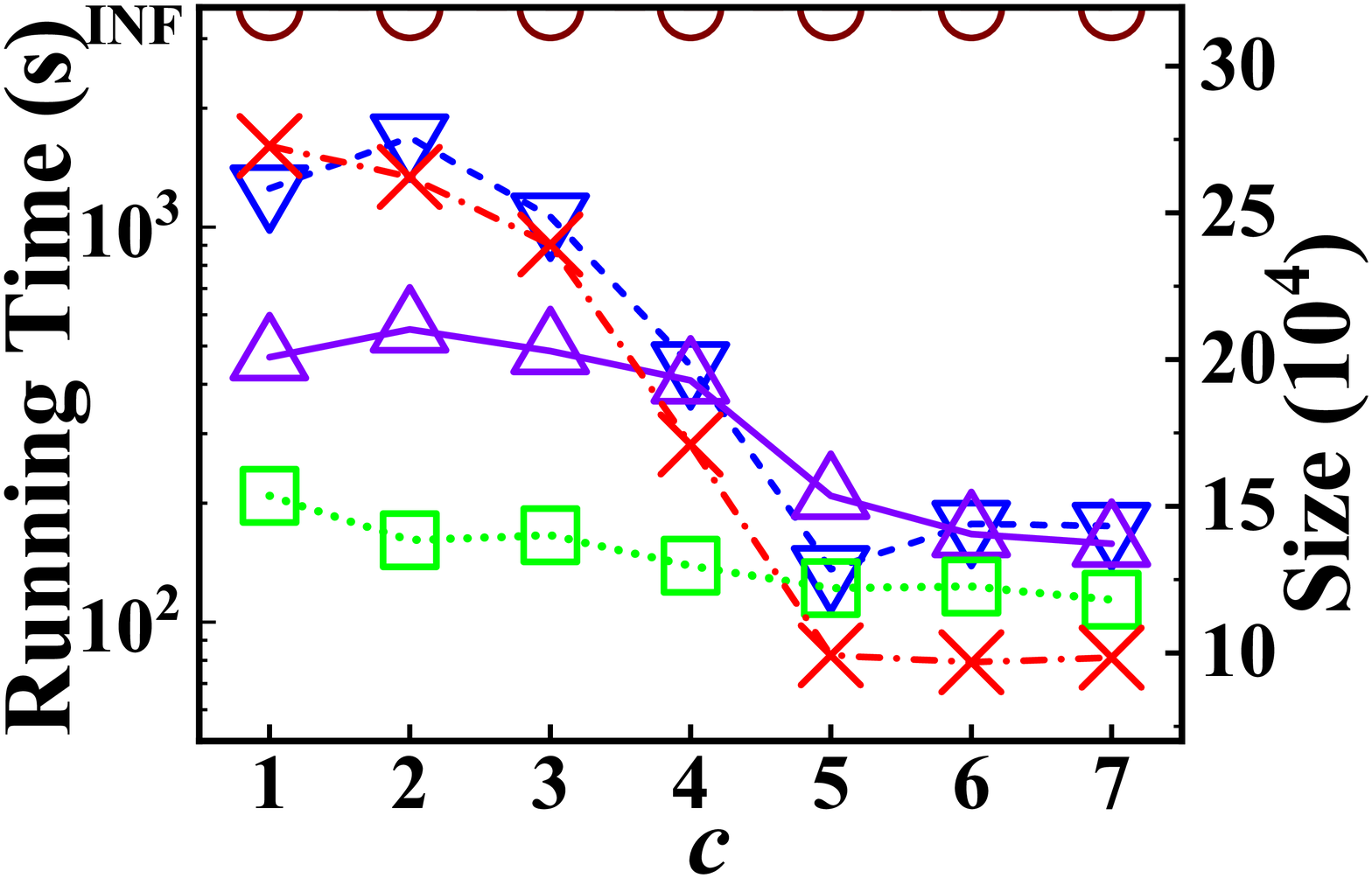}
		\label{fig:nba_c}
	}
	\vspace{-2mm}
	\caption{{Running time of different algorithms and the size of ARSP on real datasets.}}
	\label{fig:real-data}
	\vspace{-2mm}
\end{figure*}

Having established {ENUM} is inefficient to compute ARSP, henceforth it is excluded from the following experiments.
The curve of {KDTT} is also omitted as it is always outperformed by {KDTT+}.
Fig.~\ref{fig:synthetic-data} (g)-(i) plot the results on synthetic datasets with varying dimensionality $d$.
With the increase of $d$, the cost of $\calF$-dominance test increases.
Thus, the running time of all algorithms increases.
QDTT+ and KDTT+ are more efficient than B\&B on low-dimensional datasets, but their scalability is relatively poor.
This is because when $d$ grows, the dataset becomes sparser, causing the subtrees pruned during the preorder traversal get closer to leaf nodes in {KDTT+} and {QDTT+}.
Moreover, the exponential growth in the number of child nodes of {QDTT+} also causes its inefficiency on high-dimensional datasets.
When the dataset becomes sparser, an instance is more likely not to be $\calF$-dominated by others.
Therefore, the size of ARSP increases with higher dimensionality.

Fig.~\ref{fig:synthetic-data} (j)-(l) show the effect of $l$ by varying $l$ from 0.1 to 0.6.
As $l$ increases, the number of instances $\calF$-dominated by all instances of an object decreases.
Thus, the size of ARSP and the running time of all algorithms increase.
Compared to others, {B\&B} is more sensitive to $l$ since it determines not only the number of instances to be processed but also the time consumed in querying aggregated R-trees.

Fig.~\ref{fig:synthetic-data} (m)-(o) show the runtime of different algorithms and the size of ARSP on synthetic datasets with different $\phi$.
According to equation~\ref{eq:rskyprob-def}, the more objects $T$ with $\sum_{t \in T}p(t) < 1$, the less instances with zero rskyline probabilities.
Hence, the running time and the size of ARSP both increase as $\phi$ increases.
Similar to $l$, $\phi$ also affects B\&B greatly since the larger $\phi$, the fewer instances are added to the pruning set $P$.
%3 + 2 c

%\begin{figure*}[t]
%	\centering
%	\includegraphics[height=0.12in]{figures/new/color_legend_2}
%	\\
%	\subfigure[\textcolor{blue}{IND, WR, vary $c$}]{
%		\includegraphics[width=0.18\textwidth]{figures/new/inde_vary_c}
%		\label{fig:ind_wr_c}
%	}
%	\subfigure[\textcolor{blue}{ANTI, WR, vary $c$}]{
%		\includegraphics[width=0.18\textwidth]{figures/new/anti_vary_c}
%		\label{fig:anti_wr_c}
%	}
%	\subfigure[\textcolor{blue}{IND, IM, vary $m$}]{
%		\includegraphics[width=0.18\textwidth]{figures/new/inde_vary_m_im}
%		\label{fig:ind_im_m}
%	}
%	\subfigure[\textcolor{blue}{IND, IM, vary $d$}]{
%		\includegraphics[width=0.18\textwidth]{figures/new/inde_vary_d_im}
%		\label{fig:ind_im_d}
%	}
%	\subfigure[\textcolor{blue}{IND, IM, vary $c$}]{
%		\includegraphics[width=0.18\textwidth]{figures/new/inde_vary_c_im}
%		\label{fig:ind_im_c}
%	}	\vspace{-2mm}
%	\caption{\textcolor{blue}{Running time of different algorithms and the size of ARSP  different constraints.}}
%	\label{fig:constraints}	\vspace{-2mm}
%\end{figure*}

Fig.~\ref{fig:synthetic-data} (p)-(q) plot the effect of $c$ on IND and ANTI ($d = 6$).
Results on CORR are omitted, in which the running time of all algorithms and the size of ARSP decrease as $c$ grows.
The reason is that the preference region narrows with the growth of $c$, which enhances the $\calF$-dominance ability of each instance.
Therefore, more instances are pruned during the computation.
Whereas, this also results in the need to perform more $\calF$-dominance tests to compute rskyline probabilities of unpruned instances.
%The running time first increases and then decreases.
The trends of the running time on IND and ANTI reflect the compromise of these two factors.
B\&B perform inconsistently as its pruning strategy is more effective on IND.

Fig.~\ref{fig:synthetic-data} (r)-(t) show the results under linear constraints generated by IM.
Running time of all algorithms and the size of ARSP show similar trends to WR under all parameters except $c$.
As stated above, the $\calF$-dominance ability of each instance improves with the growth of $c$.
Thus, as shown in Fig.~\ref{fig:ind_im_c}, the running time of all algorithms decreases, except QDTT+.
This is because the number of vertices of the preference region generated by IM increases as $c$ grows (see the curve of $V$), thus leading to the dimensional disaster of the quadtree.
This also accounts for the failure of QDTT+ when $d \ge 5$ in Fig.~\ref{fig:ind_im_d}.

%5 real

Experimental results on real datasets confirm the above ob-servations.
Fig.~\ref{fig:real-data} (a) shows the results on IIP with varying $m$.
As introduced in the datasets' description, each records in IIP is treated as an uncertain object with one instance.
This means $\phi = 1$, \ie, every object $T$ in IIP satisfies $\sum_{t \in T} p(t) < 1$.
Thus, the size of ARSP is the number of input instances.
And B\&B almost degenerates into LOOP, since no instances are pruned and no computations are reused.
Fig.~\ref{fig:real-data} (b)-(c) show the results on CAR and NBA with different $m$.
It is noticed that attribute variance is pretty large in these two datasets, \eg, in NBA, about half of the players got zero points in some games but more than 20 points in other games.
Therefore, relative performance of algorithms are similar to synthetic datasets with large $l$.
This also holds for results on NBA with different $d$ and $c$ which are shown in Fig.~\ref{fig:real-data} (d)-(e).

\subsection{Experimental Results under Weight Ratio Constraints.}

Since the data structure stated in Theorem~\ref{thm:fast-point-location} is theoretical in nature, we introduce a specialized version of {DUAL-MS} for $d = 2$ to avoid this.	
Recall that for each instance $t$, we reduce the computation of $\Pr_{\rm rsky}(t)$ to $2^{d-1}$ half-space reporting problems in section~\ref{sec:reduction}. 
When $d = 2$, we notice that these two half-space queries can be reinterpreted as a continuous range query.
See Fig.~\ref{fig:specilized_version} for an illustration.
When processing $t_{2, 3}$,  we can regard $t_{2, 3}$ as the origin, ray $y = t_{2, 3}[2], x \ge t_{2, 3}[1]$ as the base.
Then each instance can be represented by an angle, \eg, $\theta = \pi + \arctan\frac{12 - 5}{9 - 6}$ for $t_{3,1}$.
In such case, the two query halfspaces $h_{t_{2,3}, 0} : t[2] \le -0.5t[1] + 16.5$ and $h_{t_{2,3},1} : t[2] \le -2t[1] + 30$ can be transformed to the range query $[\pi - \arctan\frac{1}{2}, 2\pi - \arctan2]$ with respect to angle.
After the transformation, we can simply use a binary search tree to organize the instances instead of the point location tree.
We give an implementation of this specialized {DUAL-MS} and evaluate its performance on {IIP} dataset.
As a comparison, we attach a simple preprocessing strategy to {KDTT+}, which removes all instances with zero skyline probability from $I$ in advance.
Fig.~\ref{fig:dual_ms_iip} shows the running time of these two algorithms.
It is noticed that although the query efficiency is improved, the huge preprocessing time and memory consumption prevents its application on big datasets.

\begin{figure}[ht]
	\centering
	\subfigure[Specialized version]{
		\includegraphics[width=.42\linewidth]{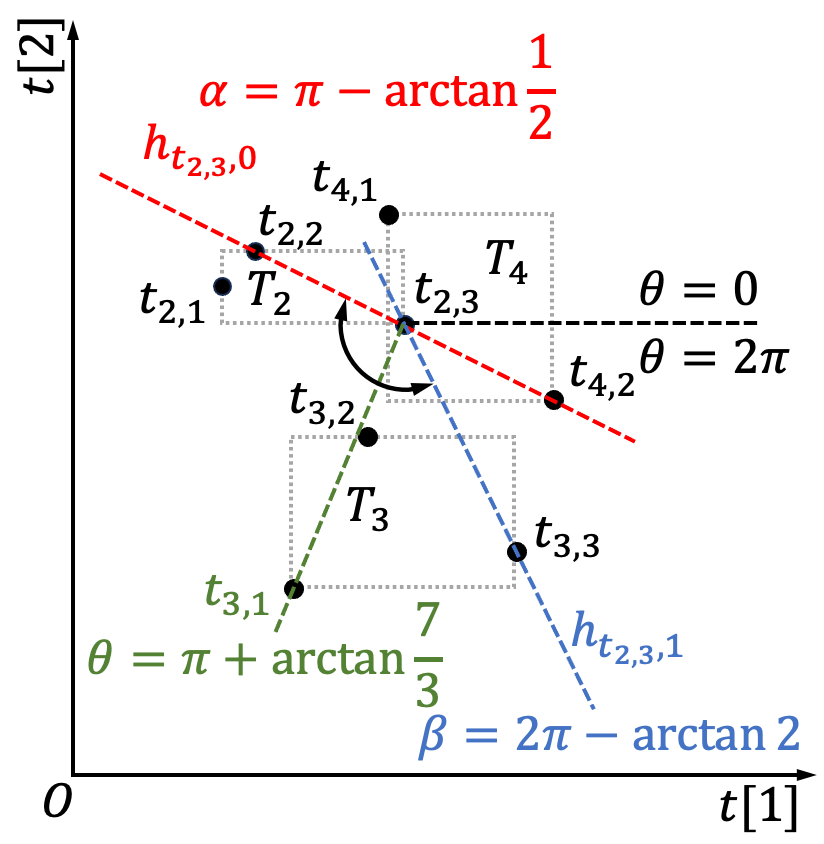}
		\label{fig:specilized_version}
	}\hfill
	\subfigure[Running time on IIP]{
		\includegraphics[width=.47\linewidth]{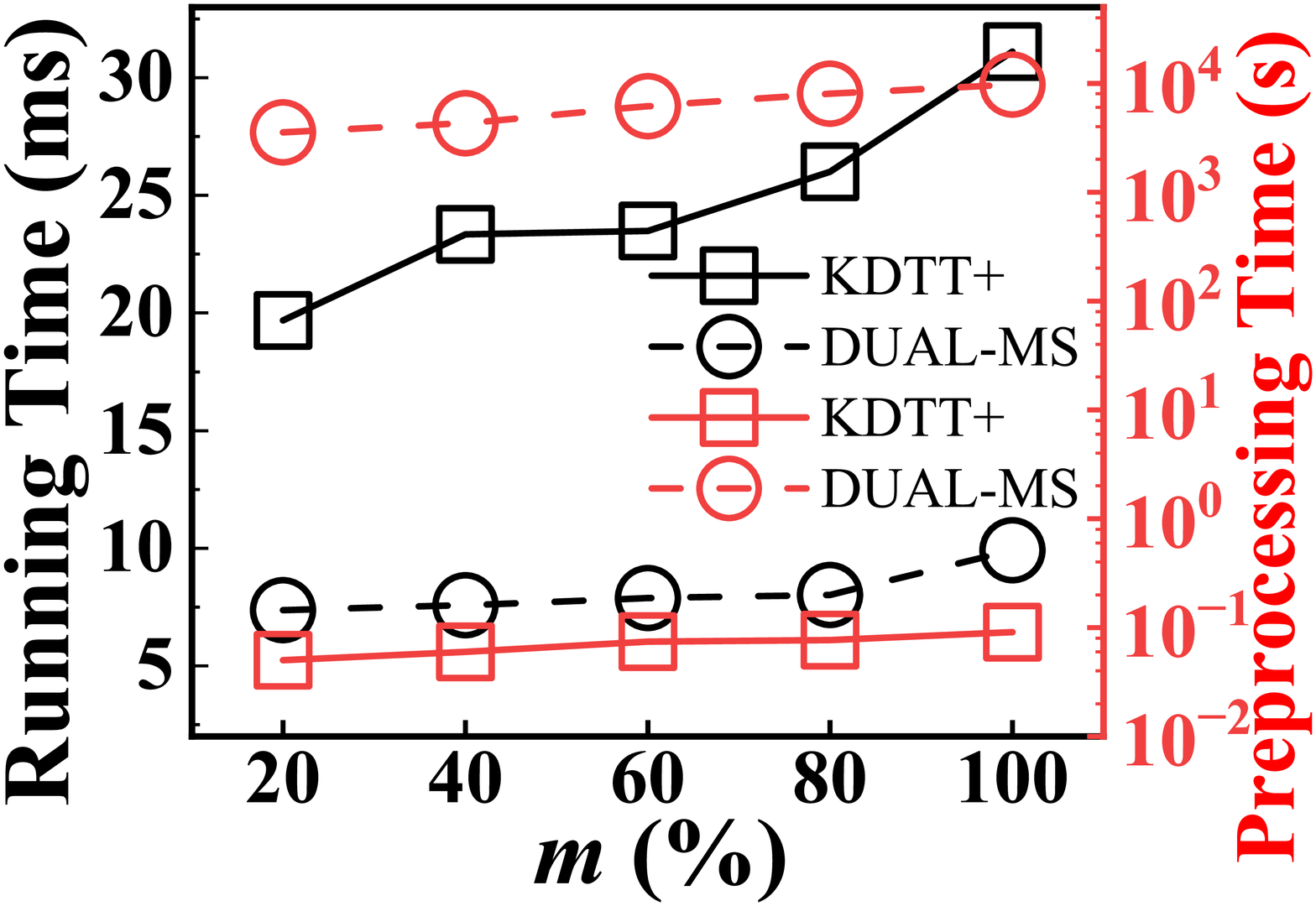}
		\label{fig:dual_ms_iip}
	}\vspace{-2mm}
	\caption{A specialized version of {DUAL-MS} for $d = 2$ and its running time on IIP dataset.}
	\vspace{-2mm}
\end{figure}

%To verify the efficiency of {DUAL-MS}, we give a practical implementation of {DUAL-MS} when $d = 2$ in~\cite{gao2023computing}.
%Experimental results show that although the query efficiency is improved, the huge preprocessing time and memory consumption prevents its application on big datasets.
The above drawbacks of {DUAL-MS} are alleviated when it comes to process eclipse queries.
This is because eclipse is always a subset of skyline $S$, which has a logarithmic size in expectation.
Meanwhile, the multi-level strategy is no longer needed since for each object $t \in S$, $t$ belongs to the eclipse of $D$ iff all point location queries on $S^*_{t, k}$ ($0 \le k < 2^{d-1}$) return emptiness.
We implement the {DUAL-S} for eclipse query processing, in which we use a $kd$-tree to index the dataset constructed by performing shifted strategy.
For comparison, we also implement the state-of-the-art index-based algorithm {QUAD}~\cite{DBLP:conf/icde/Liu0ZP021} in C++ and compare their efficiency and scalability with respect to data cardinality $n$, data dimensionality $d$, and ratio range $q$.
Similar to~\cite{DBLP:conf/icde/Liu0ZP021}, the defaulted value is set as $n = 2^{14}$, $d = 3$, and $q = [0.36, 2.75]$.

\begin{figure}[ht]
	\centering
	\subfigure[Effect of $n$]{
		\includegraphics[width=.3\linewidth]{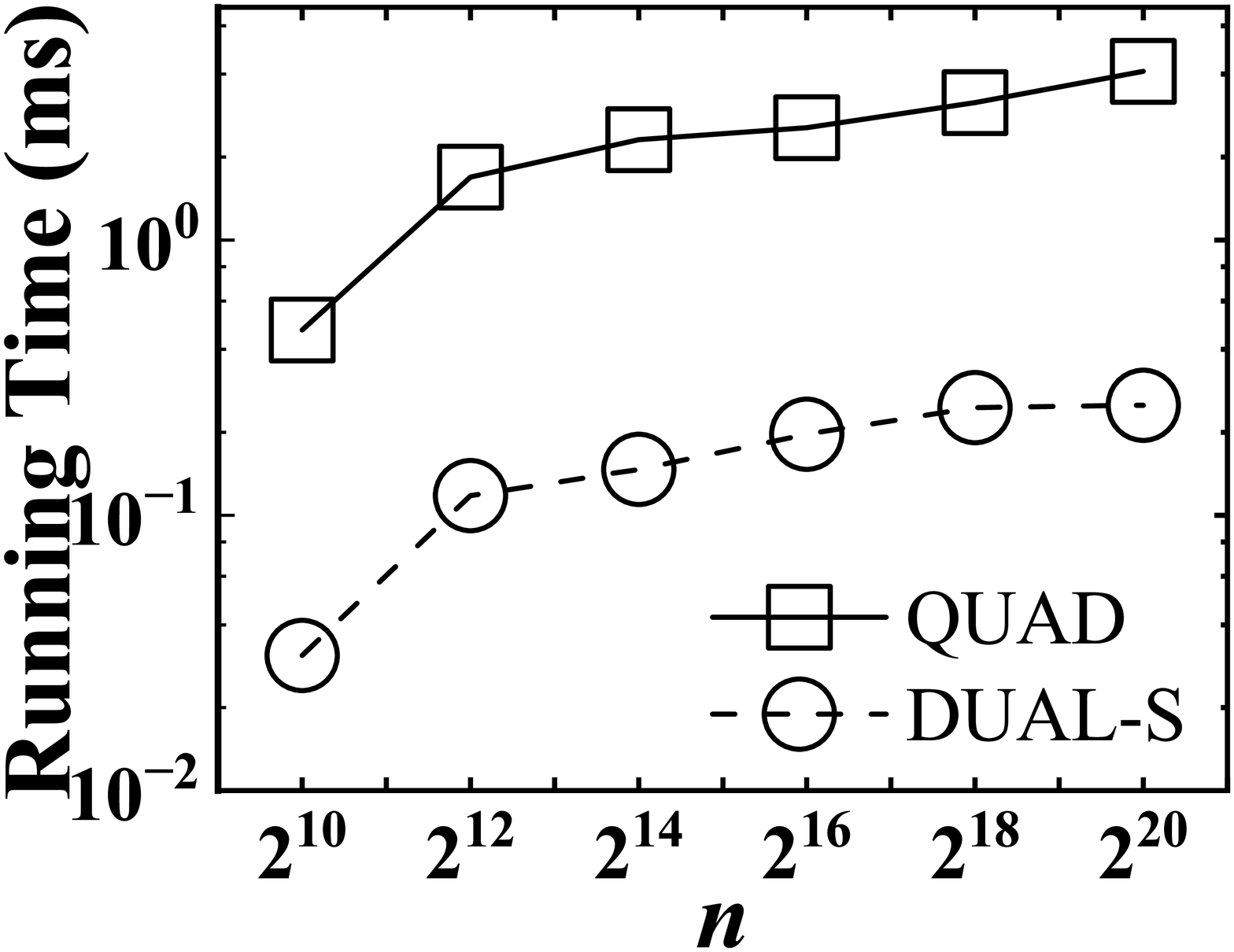}
		\label{fig:eclipse_n}
	}\hspace{-1ex}
	\subfigure[Effect of $d$]{
		\includegraphics[width=.3\linewidth]{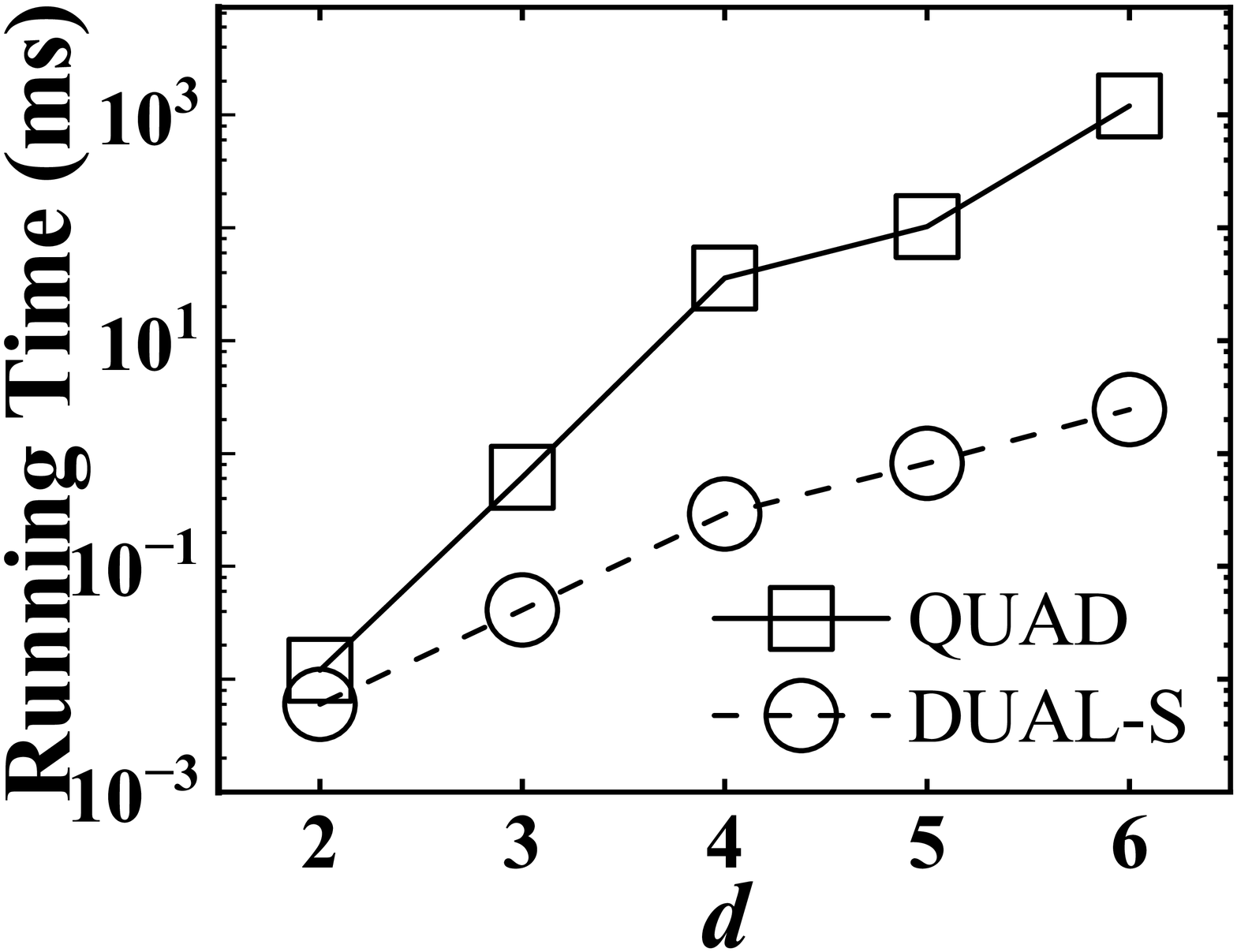}
		\label{fig:eclipse_dim}
	}\hspace{-1ex}
	\subfigure[Effect of $q$]{
		\includegraphics[width=.3\linewidth]{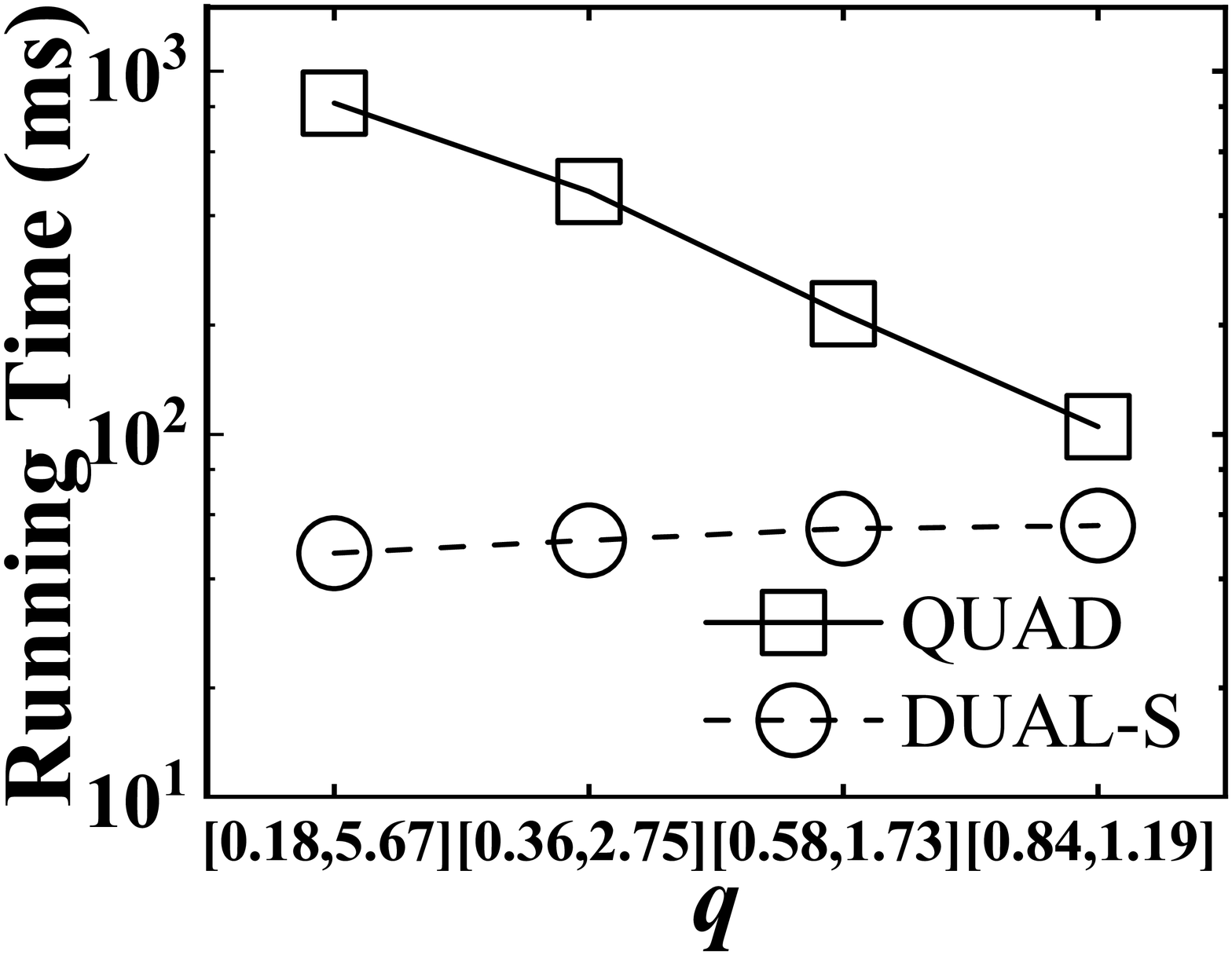}
		\label{fig:eclipse_q}
	}\vspace{-2mm}
	\caption{Running time for eclipse query ({IND}).}
	\label{fig:eclipse}
	\vspace{-2mm}
\end{figure}

As shown in Fig.~\ref{fig:eclipse} (a)-(b), the running time of these two algorithms increases as $n$ and $d$ grows.
{DUAL-S} outperforms {QUAD} by at least an order of magnitude and even more on high-dimensional datasets.
The reason is that {QUAD} needs to iterate over the set of hyperplanes returned by the window query performed on its Intersection Index, and then reports all objects with zero order vector as the final result.
This takes $O(s^2)$ time, where $s$ is the skyline size of the dataset.
But {DUAL-S} excludes an object from the result if there is a query returns non-empty result, which only take $O(s)$ time.
Moreover, the hyperplane quadtree adopted in {QUAD} scales poorly with respect to $d$ for the following two reasons.
On the one hand, the tree index has a large fan-out since it splits all dimensions at each internal node.
On the other hand, the number of intersection hyperplanes of a node decreases slightly relative to its parent, especially on high-dimensional datasets, which results in an unacceptable tree height.
Moreover, as shown in Fig.~\ref{fig:eclipse_q}, {QUAD} is more sensitive to the ratio range than {DUAL-S} because the number of hyperplanes returned by the window query actually determines the running time.

\section{Related Work}\label{sec:relatedwork}

In this section, we elaborate on two pieces of previous work that are most related to ours.

\noindent{\bf Queries on uncertain datasets.}
Pei \etal first studied how to conduct skyline queries on uncertain datasets~\cite{DBLP:conf/vldb/PeiJLY07}.
They proposed two algorithms to identify objects whose skyline probabilities are higher than a threshold $p$.
Considering inherent limitations of threshold queries, Atallah and Qi first addressed the problem of computing skyline probabilities of all objects~\cite{DBLP:conf/pods/AtallahQ09}.
They proposed a $\tilde{O}(n^{2-1/(d+1)})$-time algorithm by using two basic all skyline probabilities computation methods, weighted dominance counting method and grid method, to deal with frequent and infrequent objects, respectively. 
With a more efficient sweeping method for infrequent objects, Atallah~\etal improved the time complexity to $\tilde{O}(n^{2-1/d})$~\cite{DBLP:journals/tods/AtallahQY11}.
However, the utilities of these two algorithms are limited to 2D datasets because of a hidden factor exponential in the dimensionality of the dataset, which came from the high dimensional weighted dominance counting algorithm.
To get rid of this, Afshani~\etal calculated skyline probabilities of all instances by performing a pre-order traversal of a modified KD-tree~\cite{DBLP:journals/mst/AfshaniAALP13}.
With the well-know property of the KD-tree, it is proved that the time complexity of their algorithm is $O(n^{2-1/d})$.
More practically, Kim \etal introduced an in-memory Z-tree structure in all skyline probabilities computation to reduce the number of dominance tests, which has been experimentally demonstrated to be efficient~\cite{DBLP:journals/tkde/KimIP12}.
However, it is non-trivial to revise these algorithms for computing all skyline probabilities to address the problem studied in this paper.
This is because all of them rely on the fact that the dominance region of an instance is a hyper-rectangle, which no longer holds under $\calF$-dominance.

%TODO:uncertain topk和topk uncertain sky prob
Somehow related to what we study in this paper are those works on top-$k$ queries on uncertain datasets~\cite{DBLP:conf/icde/SolimanIC07, DBLP:journals/dpd/WangSY16, DBLP:conf/icde/HuaPZL08, DBLP:conf/icde/YiLKS08, DBLP:conf/sigmod/GeZM09}.
Under the possible world model, top-$k$ semantics are unclear, which give rise to different definitions, \eg, to compute the most likely top-$k$ set, the object with high probability to rank $i$-th, the objects having a probability greater than a specified threshold to be included in top-$k$, \etc
Our work differs from theirs as an exactly input weight is required in these studies, whereas we focus on finding a set of non-$\calF$-dominated objects where $\calF$ is a set of user-specified scoring functions.
In other word, our work can be regarded as extending theirs by relaxing the input preference into a region.

\noindent{\bf Operators with restricted preference.}
Given a set of monotone scoring functions $\calF$, Ciaccia and Martinenghi defined $\calF$-dominance and  introduced two restricted skyline operators, \textsc{ND} for retrieving the set of non-$\calF$-dominated objects and \textsc{PO} for finding the set of objects that are optimal according to at least one function in $\calF$.
%that an object $t$ $\calF$-dominates another object $s$ if $t$ scores better than $s$ for any $f \in \calF$~\cite{DBLP:journals/pvldb/CiacciaM17}.
%Based on $\calF$-dominance, they
And they designed several linear programming based algorithms for these two queries, respectively.
Mouratidis and Tang extended \textsc{PO} under top-$k$ semantic when $\calF$ is a convex preference polytope $\Omega$, \ie, they studied the problem of identifying all objects that appear in the top-$k$ result for at least one $\omega \in \Omega$~\cite{DBLP:journals/pvldb/MouratidisT18}.
%They first disqualified records $\calF$-dominated by $k$ or more others, and then determined the top-$k$-th in each partition of $\Omega$ among the remaining candidates.
Liu \etal investigated a case of $\calF$-dominance where $\calF$ consists of $d-1$ constraints on the weight ratio of other dimensions to the user-specified reference dimension~\cite{DBLP:conf/icde/Liu0ZP021}.
They defined eclipse query as retrieving the set of all non-eclipse-dominated objects and proposed a series of algorithms.
These works only consider datasets without uncertainty, and we extend above dominance-based operators to uncertain datasets.
Their techniques can not be applied to our problem since the introduction of uncertainty makes the problem challenging as for each instance, we now need to identify all instances that $\calF$-dominate it.

\section{Conclusions}\label{sec:conclusions}

In this paper, we study the problem of computing ARSP to aid multi-criteria decision making on uncertain datasets.
%We study this problem from both complexity and algorithm perspective.
We first prove that no algorithm can compute ARSP in truly subquadratic time without preprocessing, unless the orthogonal vectors conjecture fails.
Then, we propose two efficient algorithms to compute ARSP when $\calF$ is a set of linear scoring functions whose weights are described by linear constraints.
We use preprocessing techniques to further improve the query time under weight ratio constraints.
Our thorough experiments over real and synthetic datasets demonstrate the effectiveness of ARSP and the efficiency of our proposed algorithms.
For future directions, there are two possible ways.
On the one hand, conducting rskyline analysis on datasets with continuous uncertainty remains open, where it becomes expensive to make the integral for computing the dominance probability.
On the other hand, it is still worthwhile to investigate concrete lower bounds of the ARSP problem under some specific dimensions.

\bibliographystyle{IEEEtran}
\bibliography{ARSP_ref.bib}
	
\end{document}